\documentclass[3p,10pt]{elsarticle}

\usepackage{latexsym}
\usepackage{xspace}
\usepackage{amssymb}
\usepackage{amsmath}
\usepackage{stmaryrd}
\usepackage{hyperref}
\usepackage{url}
\usepackage{graphicx}
\input{xy}
\xyoption{all}
\usepackage{calc}
\usepackage{tikz}
\usepackage{flushend}
\usepackage{color}
\usepackage[ruled,vlined,linesnumbered]{algorithm2e}
\usepackage{multirow}
\usepackage{array}

\newtheorem{theorem}{Theorem}[section]
\newtheorem{lemma}[theorem]{Lemma}
\newtheorem{proposition}[theorem]{Proposition}
\newtheorem{corollary}[theorem]{Corollary}

\newtheorem{Definition}[theorem]{Definition}
\newtheorem{Example}[theorem]{Example}
\newtheorem{Remark}[theorem]{Remark}

\newenvironment{example}{\begin{Example}\begin{em}}{\end{em}\end{Example}}

\newproof{proof}{Proof}

\SetKwInput{Input}{Input}
\SetKwInput{Output}{Output}
\SetKwInput{Purpose}{Purpose}
\SetKw{Break}{break}
\SetKw{Continue}{continue}
\SetKw{Return}{return}
\SetKw{New}{new\,}

\def\eqref#1{(\ref{#1})}

\def\tuple#1{\langle#1\rangle}

\newcommand{\E}{\exists}

\newcommand{\mL}{\mathcal{L}}

\newcommand{\fALCreg}{\mbox{$f\!\mathcal{ALC}_{reg}$}\xspace}

\newcommand{\mR}{\mathcal{R}}

\newcommand{\mI}{\mathcal{I}}

\newcommand{\mZ}{\mathcal{Z}}

\newcommand{\mM}{\mathcal{M}}
\newcommand{\mMp}{{\mathcal{M}'}}
\newcommand{\mMdp}{{\mathcal{M}''}}
\newcommand{\mME}{{\mathcal{M}/_{\!E}}}
\newcommand{\DeltaM}{\Delta^{\!\mathcal{M}}}
\newcommand{\DeltaMp}{\Delta^{\!\mathcal{M}'\!}}
\newcommand{\DeltaMdp}{\Delta^{\!\mathcal{M}''\!}}
\newcommand{\DeltaME}{\Delta^{\!\mathcal{M}/_{\!E}}}

\newcommand{\myend}{\mbox{}\hfill{\footnotesize$\Box$}}

\newcommand{\comment}[1]{}

\newcommand{\fand}{\varotimes}

\newcommand{\fto}{\Rightarrow}
\newcommand{\fequiv}{\Leftrightarrow}

\newcommand{\SP}{\Sigma_P}
\newcommand{\SA}{\Sigma_A}
\newcommand{\Ss}{\Sigma_S}
\newcommand{\So}{\Sigma_O}

\newcommand{\fKP}{$\mathit{fK}_{\!\mL}^{\Phi}$\xspace}

\newcommand{\unprocessed}{\mathit{unprocessed}}

\newcommand{\FfAs}{FfAs\xspace}
\newcommand{\mIE}{{\mathcal{I}/_{\!E}}}

\journal{arXiv}

\begin{document}
\sloppy
	
\begin{frontmatter}
		
\title{Minimizing fuzzy Kripke models over residuated lattices}

\author{Linh Anh Nguyen}
\ead{nguyen@mimuw.edu.pl}
\ead{nalinh@ntt.edu.vn}				
\address{Institute of Informatics, University of Warsaw, Banacha 2, 02-097 Warsaw, Poland}
\address{Faculty of Information Technology, Nguyen Tat Thanh University, Ho Chi Minh City, Vietnam}

\begin{abstract}
Model minimization is a fundamental problem in formal methods, as smaller models can substantially reduce the complexity of model analysis and reasoning while preserving relevant semantic properties. This work studies the minimization of finite fuzzy Kripke models over residuated lattices while preserving global or local assertions in fuzzy modal logics, which are either fuzzy propositional dynamic logic or a restriction thereof, optionally enriched with nominals, converse modalities, and/or the universal program constructor. We prove that, whenever the multiplication operator of the underlying residuated lattice is right-continuous, the quotient model induced by the greatest fuzzy auto-bisimulation is a minimal reduction preserving global assertions. If the underlying fuzzy modal logic includes the universal program constructor, this quotient model is also a minimal reduction preserving local assertions and a minimal fuzzy Kripke model bisimilar to the original model. In the case where the universal program constructor is absent and the multiplication operator is additionally strictly monotone with respect to~1, we present the first algorithm for constructing a minimal reduction preserving local assertions; the resulting model is also a minimal bisimilar model.
\end{abstract}


\end{frontmatter}

%===============================================================================

\section{Introduction}
\label{section: intro}

Minimization is an important topic in the study of formal models. Replacing a complex structure with a smaller, semantically equivalent one improves its comprehensibility and often reduces the computational resources required for subsequent analysis. Consequently, minimization has long played an important role in areas such as automata theory, model checking, concurrency theory, and formal verification. Classical examples include the minimization of deterministic finite automata and the reduction of labeled transition systems or Kripke models modulo behavioral equivalences such as bisimulation. In each case, the objective is to construct the smallest model that preserves the properties of interest.

In description logics, minimizing an interpretation is not intended to improve the efficiency of reasoning, since reasoning is performed with respect to a knowledge base rather than a particular interpretation. Nevertheless, minimized interpretations are valuable in ontology engineering, especially when presenting compact countermodels showing that a queried assertion or axiom is not entailed by a knowledge base. 
A similar observation applies to modal logics. The minimization of a Kripke model does not directly accelerate logical inference, which is likewise knowledge-based rather than model-based. Instead, its primary benefit lies in the generation of concise countermodels. Smaller models are easier to inspect and understand, making them particularly useful in interactive theorem-proving environments, model exploration, and explanation systems, where they provide clear evidence that a formula is not a logical consequence of a given knowledge base.

Traditionally, the objective of model minimization is to preserve a chosen notion of semantic equivalence, such as language equivalence for automata or bisimilarity for transition systems and Kripke models. In many applications, however, preserving the truth of a selected class of formulas is preferred. This observation motivates the search for minimization techniques tailored to the logical language under consideration, producing smaller models while retaining exactly those properties relevant to the intended reasoning task.

In this work, we study the problem of minimizing a fuzzy Kripke model over a residuated lattice while preserving its global or local assertions. The underlying fuzzy modal logic is either fuzzy propositional dynamic logic or a restriction thereof, optionally enriched with nominals, converse modalities, and/or the universal program constructor.

\subsection{Related work}

The minimization and reduction of finite fuzzy automata (\FfAs) have attracted considerable attention in the literature. Existing research on exact minimization has focused primarily on deterministic \FfAs \cite{LI20071423,Halamish2015,YANG202172,DEMENDIVILGRAU2024109108,ShamsizadehZG24}. In contrast, for nondeterministic \FfAs, most studies address state reduction rather than the construction of a truly minimal automaton \cite{BASAK2002223,PEEVA20084152,CSIP.10,SCI.14,SCB.18}. An exception is the work~\cite{LITFS2015}, which deals with exact minimization of nondeterministic \FfAs over the G\"odel structure. This situation is not surprising, since even the minimization of nondeterministic finite (crisp) automata is PSPACE-complete~\cite{Jiang1993}, implying that the corresponding problem for nondeterministic \FfAs is PSPACE-hard. 
Another line of research considers approximate state reduction, aiming to construct a smaller fuzzy automaton that approximates the behavior of a given one \cite{Belohlvek2009OnAM,YangL.19,StanimirovicMC.22,FSS-D-25-00029}. Regarding the underlying algebraic structures, the approaches proposed in \cite{BASAK2002223,LI20071423,Belohlvek2009OnAM,Halamish2015,YangL.19,YANG202172,StanimirovicMC.22,ShamsizadehZG24} are confined to the G\"odel structure or, more generally, Heyting algebras, whereas the methods developed in \cite{PEEVA20084152,CSIP.10,SCI.14,SCB.18,DEMENDIVILGRAU2024109108,FSS-D-25-00029} apply to arbitrary residuated lattices. 
A common feature of many minimization and reduction techniques is the use of quotient constructions, where equivalent states are identified and merged to obtain a smaller automaton.

Description logics are closely related to modal logics. In particular, interpretations, individuals, concepts, and roles in description logics correspond, respectively, to Kripke models, states, formulas, and reachability relations in modal logics (atomic and complex roles correspond to atomic and complex programs, respectively, in propositional dynamic logic). Similarly, fuzzy description logics (FDLs) are closely related to fuzzy modal logics. For example, the FDL $\fALCreg$ is a variant of fuzzy propositional dynamic logic. The works most closely related to the present article are \cite{minimization-by-fBS,DBLP:journals/fss/Nguyen24,nguyen2026approximateminimizationinterpretationsfuzzy}, which concern the minimization of fuzzy interpretations in FDLs. We mention below the aspects of these works that are most relevant to the present article.
\begin{itemize}
\item In~\cite{minimization-by-fBS}, the authors studied the problem of minimizing fuzzy interpretations in FDLs under the G\"odel semantics. The considered class of FDLs consists of the FDL $\fALCreg$ and its extensions with nominals, inverse roles, and/or the universal role. Given a fuzzy interpretation $\mI$, they defined the quotient interpretation $\mIE$ with respect to the greatest fuzzy auto-bisimulation $E$ of $\mI$ and proved that it is a minimal reduction of $\mI$ that preserves fuzzy terminological axioms.\footnote{Terminological axioms, also called TBox axioms, can be represented as general concept inclusions (GCIs).} Furthermore, if the underlying FDL includes the universal role, then $\mIE$ is also a minimal reduction of $\mI$ that is bisimilar to $\mI$, as well as a minimal reduction of $\mI$ that preserves fuzzy concept assertions.

\item In~\cite{DBLP:journals/fss/Nguyen24}, we studied the problem of minimizing fuzzy interpretations in FDLs that include the Baaz projection operator.\footnote{The Baaz projection operator~$\triangle$ is defined by $\triangle x = 1$ if $x = 1$, and $\triangle x = 0$ otherwise (i.e., for $0 \leq x < 1$).} In the presence of this operator, greatest crisp auto-bisimulations, rather than greatest fuzzy auto-bisimulations, characterize the logical indistinguishability of individuals, independently of the underlying algebra of fuzzy truth values~\cite{DBLP:journals/tfs/NguyenN23}. Given a fuzzy interpretation $\mI$, we defined the quotient interpretation $\mIE$ with respect to the greatest crisp auto-bisimulation $E$ of $\mI$ and proved that it is a minimal reduction of $\mI$ that preserves fuzzy terminological axioms. Furthermore, if the underlying FDL includes the universal role or $\mI$ is connected, then $\mIE$ is also a minimal reduction of $\mI$ that is bisimilar to $\mI$, as well as a minimal reduction of $\mI$ that preserves fuzzy concept assertions. If the universal role is disallowed and $\mI$ is not connected, unreachable individuals can be deleted first.

\item In~\cite{nguyen2026approximateminimizationinterpretationsfuzzy}, we provided an algorithm for minimizing a finite fuzzy interpretation while preserving fuzzy concept assertions in FDLs without the universal role and the Baaz projection operator, under the G\"odel semantics. The underlying FDL is a restriction or an extension of $\fALCreg$, optionally enriched with nominals and/or inverse roles. The algorithm supports approximate preservation: it minimizes an input fuzzy interpretation $\mI$ while preserving fuzzy concept assertions up to a degree in $(0,1]$.
\end{itemize}

Works on fuzzy and crisp bisimulations between fuzzy structures include~\cite{Fan15,EleftheriouKN12,aml/MartiM18,fuin/Diaconescu20,fss/WuD16,DBLP:journals/fss/WuCBD18,DBLP:journals/ijar/WuCHC18,DBLP:journals/tfs/QiaoZF23,fss/QiaoZP23,ai/FanL14,IgnjatovicCS15}. We refer the reader to the related work sections of~\cite{FBSML,DBLP:journals/tfs/NguyenN23,nguyen2026approximateminimizationinterpretationsfuzzy} for further information about these works. Some other works that are less closely related to the present article are discussed in the related work section of~\cite{nguyen2026approximateminimizationinterpretationsfuzzy}.

\subsection{Motivation}

To the best of our knowledge, the problem of minimizing fuzzy interpretations while preserving terminological axioms or concept assertions in FDLs without the Baaz projection operator, over residuated lattices other than the G\"odel structure, has not been studied. Equivalently, in fuzzy modal logics, the corresponding problem of minimizing fuzzy Kripke models while preserving global or local assertions, in the absence of the Baaz projection operator and over residuated lattices other than the G\"odel structure, has not been studied either. This provides the main motivation for the present work.

Although the Baaz projection operator increases the expressive power of the logic, it may in many situations be excessively discriminating, in a way reminiscent of overfitting in machine learning. More precisely, in the presence of the Baaz projection operator, the similarity between two states $x$ and $x'$, defined as the infimum of $\varphi^\mM(x) \fequiv \varphi^\mM(x')$ over all formulas $\varphi$ of the considered fuzzy modal logic in a given fuzzy Kripke model~$\mM$, is always crisp, taking only the values~$0$ and~$1$~\cite{DBLP:journals/tfs/NguyenN23}. Thus, the presence of the Baaz projection operator prevents logical similarity between states from taking intermediate degrees. In contrast, in fuzzy modal logics without the Baaz projection operator, logical similarity can take intermediate degrees, providing a natural setting for studying minimization that takes graded similarities between states into account.

The universal program constructor has a different effect. Its presence allows formulas to refer, in a single modal step, to arbitrary states of the model, and therefore affects the relationship between global and local assertions. In particular, when the universal program constructor is available, the quotient construction simultaneously yields a minimal reduction preserving global assertions and a minimal reduction preserving local assertions, whereas in its absence, additional techniques are needed to construct a minimal reduction that preserves local assertions. This distinction motivates considering both settings in the present work.

\subsection{Our contributions}

The main objective of this work is to minimize a fuzzy Kripke model over a residuated lattice while preserving its global or local assertions. The underlying fuzzy modal logic is either fuzzy propositional dynamic logic or a restriction thereof, optionally enriched with nominals, converse modalities, and/or the universal program constructor, but without the Baaz projection operator.

We prove that, whenever the multiplication operator of the underlying residuated lattice is right-continuous, the quotient model $\mME$ of a finite fuzzy Kripke model $\mM$ with respect to its greatest fuzzy auto-bisimulation $E$ is a minimal reduction of $\mM$ that preserves global assertions. Furthermore, we prove that, if the underlying fuzzy modal logic includes the universal program constructor, then the quotient model is also a minimal reduction of $\mM$ that preserves local assertions, as well as a minimal fuzzy Kripke model bisimilar to~$\mM$. 
These results generalize the results of~\cite{minimization-by-fBS}, which were established only for the G\"odel structure.

We also present the first algorithm (Algorithm~\ref{alg1}) that, given a finite fuzzy Kripke model $\mM$ and its greatest fuzzy auto-bisimulation $E$, constructs a minimal reduction of $\mM$ preserving local assertions, under the assumptions that the underlying fuzzy modal logic does not include the universal program constructor and that the multiplication operator of the underlying residuated lattice is right-continuous and strictly monotone with respect to~1. (The product and {\L}ukasiewicz structures satisfy these properties.) Under the same assumptions, the model returned by the algorithm is also a minimal fuzzy Kripke model bisimilar to~$\mM$.

Assuming that the greatest fuzzy auto-bisimulation $E$ of the input fuzzy Kripke model $\mM$ is given, the quotient model $\mME$ can be constructed in time $O(mn^2)$, while Algorithm~\ref{alg1} runs in time $O(mn)$, where $n$ is the number of states of $\mM$ and $m$ is the number of non-zero transitions in~$\mM$.

Finally, we develop a Python package that implements the construction of the quotient fuzzy Kripke model $\mME$ under the product, G\"odel, and {\L}ukasiewicz semantics for a given finite fuzzy Kripke model $\mM$ and its greatest fuzzy auto-bisimulation~$E$. The package also implements Algorithm~\ref{alg1} for the product and {\L}ukasiewicz structures. It also includes a previously developed implementation for computing the greatest fuzzy auto-bisimulation of a given finite fuzzy Kripke model $\mM$ under the product, G\"odel, or {\L}ukasiewicz semantics.

\subsection{The structure of this work}

The remainder of this article is organized as follows.
Section~\ref{section: prel} presents the preliminaries needed for the subsequent development.
Sections~\ref{sec: min-ga} and~\ref{sec: min-la} are devoted to minimizing fuzzy Kripke models while preserving global and local assertions, respectively.
In Section~\ref{section: experiments}, we describe our Python package and report the results of performance tests.
Section~\ref{sec: conc} contains a discussion and conclusions.
The appendix contains additional examples. 

%===============================================================================

\section{Preliminaries}
\label{section: prel}

In this section, we recall the definitions of residuated lattices, fuzzy sets, and fuzzy relations; present the syntax and semantics of the fuzzy modal logics considered in this work; and then recall fuzzy bisimulations between fuzzy Kripke models together with some related theorems. The presentation is largely adapted from~\cite{FBSML}.

\subsection{Residuated lattices and fuzzy sets}

A {\em residuated lattice} \cite{Hajek1998,Belohlavek2002} is an algebra $\mL = \tuple{L, \leq, \fand, \fto, 0, 1}$ such that%\footnote{See also \url{https://en.wikipedia.org/wiki/Residuated_lattice}} 
\begin{itemize}
\item $\tuple{L, \leq, 0, 1}$ is a bounded lattice with the least element 0 and the greatest element 1,
\item $\tuple{L, \fand, 1}$ is a commutative monoid, 
\item for every $a, b, c \in L$,
\begin{equation}
a \fand b \leq c \ \ \textrm{iff}\ \ a \leq (b \fto c). \label{fop: GDJSK 00} 
\end{equation}
\end{itemize}

The operation $\fand$ is called the {\em multiplication}, and $\fto$ its {\em residuum}. Given a residuated lattice \mbox{$\mL = \tuple{L, \leq, \fand, \fto, 0, 1}$}, let $\land$ and $\lor$ denote the {\em meet} and {\em join} operators associated with the lattice. We write $x \fequiv y$ to denote \mbox{$(x \fto y) \land (y \fto x)$}. We assume that $\land$ and $\fand$ bind stronger than $\lor$, which in turn binds stronger than~$\fto$ and~$\fequiv$. 

\begin{table}[!t]
\centering
\caption{Residuated lattices on $[0,1]$ with $\fand$ being a t-norm and $\fto$ equal to its residuum.\label{tbl.TNormImplExm}}
\begin{tabular}{@{}p{7em}p{17em}p{13em}@{}}
        \hline
        \textbf{Name} & \textbf{$a \fand b$} & \textbf{$a \fto b$} \\
        \hline
        G\"odel & $\min(a,b)$ & $b$ if $a > b$ else 1 \\[1ex]
        Product & $ab$ & $b/a$ if $a > b$ else 1 \\[1ex]
        %Hamacher & $\displaystyle\frac{ab}{a+b-ab}$ if $a \neq 0$ or $b \neq 0$ else 0 & $\displaystyle\frac{ab}{a-b+ab}$ if $a > b$ else 1 \\[2ex]
        {\L}ukasiewicz & $\max(a+b-1, 0)$ & $\min(b - a + 1, 1)$ \\
        \hline
\end{tabular}
\end{table}

A residuated lattice $\mL = \tuple{L, \leq, \fand, \fto, 0, 1}$ is {\em linear} (respectively, {\em complete}) if the bounded lattice \mbox{$\tuple{L, \leq, 0, 1}$} is linear (respectively, complete). If~$\fand$ is the same as~$\land$, then $\mL$ is called a Heyting algebra.  
The most well-known linear and complete residuated lattices with $L = [0,1]$, where $\fand$ is a t-norm and $\fto$ is its corresponding residuum, are listed in Table~\ref{tbl.TNormImplExm}. Note that the G\"odel structure is a Heyting algebra. 

From now on, let $\mL = \tuple{L, \leq, \fand, \fto, 0, 1}$ be an arbitrary linear and complete residuated lattice. 

Note that the operation~$\fand$ is monotone with respect to both arguments, whereas the operation~$\fto$ is antitone with respect to the first argument and monotone with respect to the second one. That is, if $a \leq a'$ and $b \leq b'$, then \mbox{$a \fand b \leq a' \fand b'$} and \mbox{$(a' \fto b) \leq (a \fto b')$}. For the purposes of this work, we further say that $\fand$ is {\em strictly monotone w.r.t.~$1$} if, for every $0 < a$ and $b < 1$, it holds that $a \fand b < a$ (i.e., $a \fand b < a \fand 1$). Observe that the product and {\L}ukasiewicz t-norms are strictly monotone w.r.t.~$1$, whereas the G\"odel t-norm is not.

Since $\mL$ is assumed to be complete, we have 
\begin{equation}\label{eq: JHFAH}
a \fand \bigvee\!B = \bigvee_{b \in B} (a \fand b),
\end{equation}
for all $a \in L$ and $B \subseteq L$. This property states that $\fand$ is {\em left-continuous}. 
The operation $\fand$ is said to be {\em right-continuous} if
\[ a \fand \bigwedge\!B = \bigwedge_{b \in B} (a \fand b), \]
for all $a \in L$ and $B \subseteq L$. Clearly, all the G\"odel, {\L}ukasiewicz and product t-norms (when $L$ is the unit interval $[0,1]$) are right-continuous. 

Given a set $X$, a function $f: X \to L$ is a {\em fuzzy set}, called a {\em fuzzy subset} of $X$. 
For $x \in X$, the value $f(x)$ represents the degree of membership of $x$ in the fuzzy subset. 
Given $\{x_1,\ldots,x_n\} \subseteq X$ and $\{a_1,\ldots,a_n\} \subseteq L$, we denote by $\{x_1:a_1$, \ldots, $x_n:a_n\}$ the fuzzy subset $f$ of $X$ defined by $f(x_i) = a_i$ for $1 \leq i \leq n$ and $f(x) = 0$ for all $x \in X \setminus \{x_1,\ldots,x_n\}$. 
Given fuzzy subsets $f$ and $g$ of $X$, we write $f \leq g$ to denote that $f(x) \leq g(x)$ for all $x \in X$. If $f \leq g$, then we say that $g$ is greater than or equal to~$f$.

A {\em fuzzy relation} between $X$ and $Y$ is a fuzzy subset of $X \times Y$. A fuzzy relation on~$X$ is a fuzzy relation between $X$ and itself. 
Given $R: X \times Y \to L$, the {\em converse} \mbox{$R^- : Y \times X \to L$} of $R$ is defined by $R^-(y,x) = R(x,y)$, for all $\tuple{y,x} \in Y \times X$. 
The {\em composition} of fuzzy relations \mbox{$R_1: X \times Y \to L$} and \mbox{$R_2: Y \times Z \to L$} is the fuzzy relation \mbox{$R_1 \circ R_2 : X \times Z \to L$} such that 
\[ (R_1 \circ R_2)(x,z) = \bigvee_{y \in Y} (R_1(x,y) \fand R_2(y,z)), \] 
for all $\tuple{x,z} \in X \times Z$. 
Given a family $\mR$ of fuzzy relations between $X$ and $Y$, we define $\bigvee\!\mR$ to be the fuzzy relation between $X$ and $Y$ such that $(\bigvee\!\mR)(x,y) = \bigvee_{R \in \mR} R(x,y)$, for all $x \in X$ and $y \in Y$. 
The {\em identity relation} on $X$ is the fuzzy relation $id_X: X \times X \to L$ such that $id_X(x,x) = 1$ and $id_X(x,y) = 0$ whenever $x \neq y$.

A fuzzy relation $R$ on $X$ is called a {\em fuzzy equivalence relation} if it is
\begin{itemize}
\item {\em reflexive}: $id_X \leq R$, 
\item {\em symmetric}: $R = R^-$,
\item {\em transitive}: $R \circ R \leq R$. 
\end{itemize}

Given a fuzzy equivalence relation $R$ on $X$ and \mbox{$x \in X$}, the {\em fuzzy equivalence class} of $R$ determined by $x$ is defined to be the fuzzy subset $R_x$ of $X$ specified by $R_x(y) = R(x,y)$, for $y \in X$. 

\begin{lemma}[see, e.g., \cite{CiricIDB12}]\label{lemma: JFSJA}
Let $R$ be a fuzzy equivalence relation on $X$. For every $x, y \in X$, the following conditions are equivalent:
\begin{itemize}
\item $R(x,y) = 1$,
\item $R_x = R_y$, 
\item $R(x,z) = R(y,z)$ for all $z \in X$,
\item $R(z,x) = R(z,y)$ for all $z \in X$.
\end{itemize}
\end{lemma}

\subsection{Fuzzy modal logics}

Let $\SA$ denote a set of {\em actions}, also called {\em atomic programs}, and let $\SP$ denote a set of {\em propositions}, also called {\em atomic formulas}. In addition to $\SA$ and $\SP$, let the signature also contain a set $\Ss$ of {\em named states} (which may, for example, consist solely of the initial state $\tau$) and a set $\So$ of {\em nominals}, where $\So \subseteq \Ss$. 

We consider a class of fuzzy modal logics \fKP parameterized by $\Phi \subseteq \{\cup, \circ, *, \lor, \Box, C, U, O\}$, where the symbols $C$, $U$, and $O$ stand for converse programs, the universal program, and nominals, respectively. 
{\em Programs} and {\em formulas} of \fKP over a signature $\tuple{\SA,\SP,\Ss,\So}$ and a residuated lattice $\mL$ are defined inductively as follows:
\begin{itemize}
\item every $\sigma \in \SA$ is a program of \fKP, 
\item every $a \in L$ is a formula of \fKP, 
\item every $p \in \SP$ is a formula of \fKP,
\item if $\alpha$ and $\beta$ are programs of \fKP, $\varphi$ and $\psi$ are formulas of \fKP, and $a \in L$, then:
    \begin{itemize}
    \item $\varphi \land \psi$, $a \to \varphi$, $\varphi \to a$, and $\tuple{\alpha}\varphi$ are formulas of \fKP, 
    \item if $\mL$ is a Heyting algebra, then $\varphi \to \psi$ is a formula of \fKP,
    \item if $\cup \in \Phi$, then $\alpha \cup \beta$ is a program of \fKP,
    \item if $\circ \in \Phi$, then $\alpha \circ \beta$ is a program of \fKP,
    \item if $* \in \Phi$, then $\alpha^*$ is a program of \fKP,
    \item if $\mL$ is a Heyting algebra, then $\varphi?$ is a program of \fKP,
    \item if $\lor \in \Phi$, then $\varphi \lor \psi$ is a formula of \fKP,
    \item if $\Box \in \Phi$, then $[\alpha]\varphi$ is a formula of \fKP,
    \item if $C \in \Phi$, then $\alpha^-$ is a program of \fKP, called the {\em converse} of $\alpha$, 
    \item if $U \in \Phi$, then $U$ is a program of \fKP, called the {\em universal program}, 
    \item if $O \in \Phi$, then every $s \in \So$ is a formula of \fKP, called a {\em nominal}. 
    \end{itemize} 
\end{itemize}

In the case where $\Phi = \{\cup, \circ, *, \lor, \Box\}$, \fKP is a kind of fuzzy propositional dynamic logic (cf.\ \cite{FBSML,FDSML}). 
For the reason why the test operator ($\varphi?$) and the full implication \mbox{$(\varphi \to \psi)$} are included in \fKP only when $\mL$ is a Heyting algebra, we refer the reader to \cite[Remark~4.5]{FBSML}. 

We use letters such as 
$\sigma$ to denote actions from $\SA$, 
$p$ and $q$ to denote propositions from $\SP$, 
$s$ to denote a state name from $\Ss$, 
$a$ and $b$ to denote values from $L$, 
$\varphi$ and $\psi$ to denote formulas, 
$\alpha$ and $\beta$ to denote programs.

Given a finite set of formulas $\Psi = \{\psi_1,\ldots,\psi_n\}$ with $n \geq 0$, we denote
\begin{eqnarray*}
\textstyle\bigwedge\!\Psi & = & \psi_1 \land\ldots\land \psi_n \land 1, \\
\textstyle\bigotimes\!\Psi & = & \psi_1 \fand\cdots\fand \psi_n \fand 1.
\end{eqnarray*}
Thus, for $\Psi = \emptyset$, we have $\bigwedge\!\Psi = \bigotimes\!\Psi = 1$.

\begin{table}[t]
\caption{Interpretation of complex programs and formulas\label{table: interpretation}}
\fbox{
\parbox{0.97\textwidth}{
\begin{eqnarray*}
U^\mM(x,y) & \!=\! & 1 \\
(\alpha^-)^\mM(x,y) & \!=\! & \alpha^\mM(y,x) \\
(\alpha \cup \beta)^\mM(x,y) & \!=\! & \alpha^\mM(x,y) \lor \beta^\mM(x,y) \\
(\alpha \circ \beta)^\mM(x,y) & \!=\! & \bigvee\{\alpha^\mM(x,z) \fand \beta^\mM(z,y) \mid z \in \DeltaM \} \\
(\alpha^*)^\mM(x,y) & \!=\! & \bigvee \{\textstyle\bigotimes\{\alpha^\mM(x_i,x_{i+1}) \mid 0 \leq i < n\} \mid \\
	&& \qquad n \geq 0,\ x_0,\ldots,x_n \in \DeltaM, x_0 = x,\ x_n = y\} \\
(\varphi?)^\mM(x,y) & \!=\! & \textrm{($\varphi^\mM(x)$ if $x = y$ else 0)} \\
a^\mM(x) & \!=\! & a \\
s^\mM(x) & \!=\! & \textrm{(1 if $x = s^\mM$ else 0)} \\
(\varphi \lor \psi)^\mM(x) & \!=\! & \varphi^\mM(x) \lor \psi^\mM(x) \\
(\varphi \land \psi)^\mM(x) & \!=\! & \varphi^\mM(x) \land \psi^\mM(x) \\
(\varphi \to \psi)^\mM(x) & \!=\! & (\varphi^\mM(x) \fto \psi^\mM(x)) \\
(\tuple{\alpha}\varphi)^\mM(x) & \!=\! & \bigvee \{\alpha^\mM(x,y) \fand \varphi^\mM(y) \mid y \in \DeltaM\} \\
([\alpha]\varphi)^\mM(x) & \!=\! & \bigwedge \{\alpha^\mM(x,y) \fto \varphi^\mM(y) \mid y \in \DeltaM\}
\end{eqnarray*}}
}
\end{table}

A {\em fuzzy Kripke model} over a signature $\tuple{\SA,\SP,\Ss,\So}$ and a residuated lattice $\mL$ is a pair $\mM = \tuple{\DeltaM, \cdot^\mM}$, where $\DeltaM$ is a non-empty set of {\em states}, called the {\em domain}, and $\cdot^\mM$ is the {\em interpretation function} that maps each $p \in \SP$ to a fuzzy subset $p^\mM$ of $\DeltaM$, each $\sigma \in \SA$ to a fuzzy relation $\sigma^\mM$ on $\DeltaM$, and each $s \in \Ss$ to a state in~$\DeltaM$. The function~$\cdot^\mM$ is extended to complex programs and formulas as shown in Table~\ref{table: interpretation}. 
A fuzzy Kripke model $\mM$ is {\em finite} if $\DeltaM$ and the signature $\tuple{\SA,\SP,\Ss,\So}$ are finite. 

Every program of \fKP can be converted into converse normal form, in which the converse operator applies only to actions. Without loss of generality, we assume that all programs of \fKP are in converse normal form (when $C \in \Phi$). 

A {\em basic program} w.r.t.~$\Phi$ is an action $a \in \SA$, its converse $a^-$ when $C \in \Phi$, or the universal program $U$ when $U \in \Phi$. 
A {\em basic formula} w.r.t.~$\Phi$ is a proposition $p \in \SP$, or a nominal $s \in \So$ when $O \in \Phi$. 
A fuzzy Kripke model $\mM$ is {\em $\Phi$-image-finite} if, for every basic program $\alpha$ w.r.t.~$\Phi$ and every $x \in \DeltaM$, the set $\{y \in \DeltaM \mid \alpha^\mM(x,y) > 0\}$ is finite. In particular, if $U \in \Phi$, then $\mM$ is $\Phi$-image-finite only if it is finite. 

We say that a formula $\varphi$ is a {\em global assertion} of $\mM$ if $\varphi^\mM(x) = 1$ for all $x \in \DeltaM$. 

%===============================================================================

\subsection{Fuzzy bisimulations between fuzzy Kripke models}
\label{section: fbir}

In this subsection, we define fuzzy bisimulations in \fKP and present their properties.

\begin{Definition}[{cf.~\cite[Definition~3.1]{FSS2020}}]\label{def: HDJAK}\em
Given fuzzy Kripke models $\mM$ and $\mMp$, a fuzzy relation $Z$ between $\DeltaM$ and $\DeltaMp$ is called a {\em fuzzy $\Phi$-bisimulation} between $\mM$ and $\mMp$ if the following conditions hold for every basic formula $\varphi$ w.r.t.~$\Phi$, every basic program $\alpha$ w.r.t.~$\Phi$, and all possible values of the free variables:
\begin{eqnarray}
&& Z(x,x') \leq (\varphi^\mM(x) \fequiv \varphi^\mMp(x')) \label{eq: FB1} \\
&& \E y' \in \DeltaMp\ (Z(x,x') \fand \alpha^\mM(x,y) \leq \alpha^\mMp(x',y') \fand Z(y,y')) \label{eq: FB2} \\
&& \E y \in \DeltaM\ (Z(x,x') \fand \alpha^\mMp(x',y') \leq \alpha^\mM(x,y) \fand Z(y,y')). \label{eq: FB3}
\end{eqnarray}
We say that $\mM$ and $\mMp$ are {\em $\Phi$-bisimilar} if there exists a fuzzy $\Phi$-bisimulation $Z$ between them such that $Z(s^\mM, s^\mMp) = 1$ for all $s \in \Ss$. A fuzzy $\Phi$-bisimulation between $\mM$ and itself is called a {\em fuzzy $\Phi$-auto-bisimulation} of~$\mM$.
\myend
\end{Definition}

This definition differs from \cite[Definition~3.1]{FSS2020} in several respects:
\begin{itemize}
\item It is formulated for fuzzy Kripke models over residuated lattices, rather than for fuzzy interpretations in fuzzy description logics under the specific G\"odel semantics.
\item Owing to the introduction of basic formulas and the extension of the notion of basic programs (by incorporating the universal program), the definition is simplified.
\end{itemize}

Note that the notion of a $\Phi$-bisimulation is independent of whether the constructors $\cup$, $\circ$, $*$, $\lor$, and $\Box$ belong to~$\Phi$. It depends only on whether the features $C$, $U$, and $O$ are included in~$\Phi$. 
% Therefore, the notion of an $\emptyset$-bisimulation (respectively, an $\emptyset$-auto-bisimulation) coincides with that of a $\Phi$-bisimulation (respectively, a $\Phi$-auto-bisimulation) for $\Phi \subseteq \{\cup, \circ, *, \lor, \Box\}$. 

\begin{proposition}[{cf.~\cite[Proposition~3.6]{FBSML}}] \label{prop: JHDJH}
Let $\mM$, $\mMp$ and $\mMdp$ be $\Phi$-image-finite fuzzy Kripke models.
\begin{enumerate}
\item The identity relation on $\DeltaM$ is a fuzzy $\Phi$-auto-bisimulation of~$\mM$.
\item If $Z$ is a fuzzy $\Phi$-bisimulation between $\mM$ and $\mMp$, then $Z^-$ is a fuzzy $\Phi$-bisimulation between $\mMp$ and $\mM$.
\item If $Z$ is a fuzzy $\Phi$-bisimulation between $\mM$ and $\mMp$, and $Z'$ is a fuzzy $\Phi$-bisimulation between $\mMp$ and $\mMdp$, then $Z \circ Z'$ is a fuzzy $\Phi$-bisimulation between $\mM$ and $\mMdp$.
\item If $\mZ$ is a set of fuzzy $\Phi$-bisimulations between $\mM$ and $\mMp$, then $\bigvee\!\mZ$ is also a fuzzy $\Phi$-bisimulation between $\mM$ and $\mMp$.
\end{enumerate}
\end{proposition}

% Recall that the underlying residuated lattice is assumed to be linear and complete. 
The first two assertions are straightforward. The third assertion can be proved analogously to the corresponding assertion in \cite[Proposition~3.6]{FBSML}. The last assertion generalizes the corresponding assertion in \cite[Proposition~3.6]{FBSML} by dropping the requirement that $\mZ$ be finite (and by dealing with $\Phi$-bisimulations instead of bisimulations).

\begin{proof}
In view of the above observations, we prove only the last assertion. Let $\mZ$ be a set of fuzzy $\Phi$-bisimulations between $\mM$ and $\mMp$. We prove that $\bigvee\!\mZ$ is a fuzzy $\Phi$-bisimulation between $\mM$ and $\mMp$. It is easily seen that the condition~\eqref{eq: FB1} with $Z$ replaced by $\bigvee\!\mZ$ holds. 

Consider the condition~\eqref{eq: FB2} with $Z$ replaced by $\bigvee\!\mZ$. Let $\alpha$ be a basic program w.r.t.~$\Phi$, and let $x, y \in \DeltaM$ and $x' \in \DeltaMp$. By~\eqref{eq: JHFAH}, we have
\[
   \big(\bigvee\!\mZ\big)(x,x') \fand \alpha^\mM(x,y) = \big(\!\bigvee_{Z \in \mZ}\! Z(x,x')\big) \fand \alpha^\mM(x,y) = \bigvee_{Z \in \mZ} (Z(x,x') \fand \alpha^\mM(x,y)).
\]
For each $Z \in \mZ$, there exists $y'_Z \in \DeltaMp$ such that 
\[ 
    Z(x,x') \fand \alpha^\mM(x,y) \leq \alpha^\mMp(x', y'_Z) \fand Z(y, y'_Z).
\]
Since $\mL$ is linear and $\mMp$ is $\Phi$-image-finite, there exists $y' \in \DeltaMp$ such that 
\[ 
    \bigvee_{Z \in \mZ} (\alpha^\mMp(x', y'_Z) \fand Z(y, y'_Z)) \leq \bigvee_{Z \in \mZ} (\alpha^\mMp(x', y') \fand Z(y, y')),
\]
Thus, 
\[
         \big(\bigvee\!\mZ\big)(x,x') \fand \alpha^\mM(x,y) 
    \leq \bigvee_{Z \in \mZ} (\alpha^\mMp(x', y') \fand Z(y, y')) 
    = \alpha^\mMp(x', y') \fand \bigvee_{Z \in \mZ}\! Z(y, y').
\]
Therefore, 
\[
         \big(\bigvee\!\mZ\big)(x,x') \fand \alpha^\mM(x,y) 
    \leq \alpha^\mMp(x', y') \fand \big(\bigvee\!\mZ\big)(y, y'),
\]
which implies that the condition~\eqref{eq: FB2} with $Z$ replaced by $\bigvee\!\mZ$ holds.

The condition~\eqref{eq: FB3} with $Z$ replaced by $\bigvee\!\mZ$ can be proved analogously.
\myend
\end{proof}

\begin{corollary}\label{cor: NSLAJ}
The $\Phi$-bisimilarity relation between $\Phi$-image-finite fuzzy Kripke models is an equivalence relation.
\end{corollary}

\begin{proof}
This corollary follows quite straightforwardly from Proposition~\ref{prop: JHDJH}. For instance, we show that the $\Phi$-bisimilarity relation is transitive. 
Let $\mM$, $\mMp$ and $\mMdp$ be $\Phi$-image-finite fuzzy Kripke models such that $\mM$ is $\Phi$-bisimilar to $\mMp$, which in turn is $\Phi$-bisimilar to $\mMdp$. Let $Z$ be a fuzzy $\Phi$-bisimulation between $\mM$ and~$\mMp$ such that $Z(s^\mM, s^\mMp) = 1$ for all $s \in \Ss$, and let $Z'$ be a fuzzy $\Phi$-bisimulation between $\mMp$ and~$\mMdp$ such that $Z'(s^\mMp, s^\mMdp) = 1$ for all $s \in \Ss$. We have $(Z \circ Z')(s^\mM, s^\mMdp) = 1$ for all $s \in \Ss$. By Proposition~\ref{prop: JHDJH}, $Z \circ Z'$ is a fuzzy $\Phi$-bisimulation between $\mM$ and $\mMdp$. Therefore, $\mM$ and $\mMdp$ are $\Phi$-bisimilar.
\myend
\end{proof}

\begin{corollary}\label{cor: JFBXJ}
There exists the greatest fuzzy $\Phi$-bisimulation between any two $\Phi$-image-finite fuzzy Kripke models.
\end{corollary}

This corollary follows immediately from the fourth assertion of Proposition~\ref{prop: JHDJH}. 

\begin{corollary}[{cf.~\cite[Corollaries~3.7 and~5.3]{FBSML}}]\label{cor: HGEKS} 
Let $\mM$ be a $\Phi$-image-finite fuzzy Kripke model. Then the greatest fuzzy $\Phi$-auto-bisimulation of $\mM$ exists and is a fuzzy equivalence relation.
\end{corollary}

This corollary follows immediately from Proposition~\ref{prop: JHDJH}. Compared with \cite[Corollary~3.7]{FBSML}, it additionally asserts the existence of the greatest fuzzy $\Phi$-auto-bisimulation of $\mM$. Moreover, it extends \cite[Corollary~5.3]{FBSML} by not requiring~$\fand$ to be right-continuous.

According to the following proposition, the universal program can be omitted when computing the greatest fuzzy $\Phi$-auto-bisimulation of a finite fuzzy Kripke model.

\begin{proposition}\label{prop: HFLMS}
Let $\mM$ be a $\Phi$-image-finite fuzzy Kripke model. Then the greatest fuzzy \mbox{$(\Phi \setminus \{U\})$}-auto-bisimulation of $\mM$ is also the greatest fuzzy $\Phi$-auto-bisimulation of~$\mM$.
\end{proposition}

\begin{proof}
Let $Z$ be the greatest fuzzy $(\Phi \setminus \{U\})$-auto-bisimulation of $\mM$. Then $Z$ is reflexive and satisfies the conditions~\eqref{eq: FB2} and~\eqref{eq: FB3} with $\alpha = U$. In particular, when $\mMp = \mM$ and $\alpha = U$, the conditions~\eqref{eq: FB2} and~\eqref{eq: FB3} simplify to \mbox{$\E y'\ Z(x,x') \leq Z(y,y')$} and \mbox{$\E y\ Z(x,x') \leq Z(y,y')$}, respectively. For the former, we can take $y' = y$, while for the latter, we can take $y = y'$. Hence, $Z$ is a fuzzy $(\Phi \cup \{U\})$-auto-bisimulation of $\mM$. Consequently, since $Z$ is the greatest fuzzy $(\Phi \setminus \{U\})$-auto-bisimulation of $\mM$, it is also the greatest fuzzy $\Phi$-auto-bisimulation of~$\mM$. 
\myend
\end{proof}

A formula $\varphi$ is said to be {\em invariant} under fuzzy $\Phi$-bisimulations between finite fuzzy Kripke models if, 
for every finite fuzzy Kripke models $\mM$ and $\mMp$ and for every fuzzy $\Phi$-bisimulation $Z$ between $\mM$ and $ \mMp$, \mbox{$Z(x,x') \leq (\varphi^\mM(x) \fequiv \varphi^\mMp(x'))$} for all $x \in \DeltaM$ and $x' \in \DeltaMp$. 

\begin{theorem}[{cf.~\cite[Theorem~4.1]{FBSML}}] \label{theorem: invariance}
All formulas of \fKP are invariant under fuzzy $\Phi$-bisimulations between finite fuzzy Kripke models.
\end{theorem}

Unlike \cite[Theorem~4.1]{FBSML}, this theorem restricts attention to finite models (rather than considering witnessed models) and concerns $\Phi$-bisimulations instead of bisimulations. Its proof is analogous to that of \cite[Lemma~4.2]{FBSML}, and the extensions required to handle nominals, converse programs, and the universal program are straightforward.

The following theorem states the Hennessy--Milner property of fuzzy $\Phi$-bisimulations. It is a counterpart of \cite[Theorem~5.2]{FBSML} and can be proved analogously.

\begin{theorem}[{cf.~\cite[Theorem~5.2]{FBSML}}] \label{theorem: HM}
Suppose $\fand$ is right-continuous. Let $\mM$ and $\mMp$ be finite fuzzy Kripke models. Then the fuzzy relation $Z$ on $\DeltaM$ defined by 
\[ Z(x,x') = \bigwedge\{\varphi^\mM(x) \fequiv \varphi^\mMp(x') \mid \textrm{$\varphi$ is a formula of \fKP}\} \] 
is the greatest fuzzy $\Phi$-bisimulation between $\mM$ and~$\mMp$.	
\end{theorem}

In general, like \cite[Lemma~4.2 and Theorem~5.2]{FBSML}, Theorems~\ref{theorem: invariance} and~\ref{theorem: HM} can be generalized to infinite fuzzy Kripke models that are ``witnessed'' and ``modally saturated''. Here, we restrict attention to finite fuzzy Kripke models for simplicity.

%===============================================================================

\section{Minimizing fuzzy Kripke models while preserving global assertions}
\label{sec: min-ga}

Given a finite fuzzy Kripke model $\mM$, we call a fuzzy Kripke model $\mMp$ a {\em minimal reduction of $\mM$ preserving global assertions in \fKP} if:
\begin{itemize}
\item for every formula $\varphi$ of \fKP, $\varphi$ is a global assertion of $\mMp$ iff it is a global assertion of $\mM$;
\item if $\mMdp$ is a fuzzy Kripke model such that, for every formula $\varphi$ of \fKP, $\varphi$ is a global assertion of $\mMdp$ iff it is a global assertion of $\mM$, then $|\DeltaMdp| \geq |\DeltaMp|$.
\end{itemize}

In~\cite{minimization-by-fBS}, together with a colleague, we established results on minimizing finite fuzzy interpretations in fuzzy description logics while preserving fuzzy general concept inclusions. These results were formulated only for the G\"odel semantics, that is, for the case where $\fand$ is the G\"odel t-norm. In this section, we show that analogous results hold for any linear and complete residuated lattice where $\fand$ is right-continuous, when reformulated in terms of minimizing fuzzy Kripke models while preserving global assertions. This demonstrates that the minimization approach of~\cite{minimization-by-fBS} is not specific to the G\"odel semantics but applies uniformly across a broad class of fuzzy logics. The proofs of Lemma~\ref{lemma: HDAMA} and Theorem~\ref{theorem: HSJAO} presented here are adapted from those in~\cite[Section~4]{minimization-by-fBS}. The following definition is central to the minimization studied in this section.

\begin{Definition}[{cf.~\cite[Definition~3]{minimization-by-fBS}}] \label{def: HDJOS}\em
Let $\mM$ be a finite fuzzy Kripke model and let $E$ be its greatest fuzzy $\Phi$-auto-bisimulation. 
The {\em quotient fuzzy Kripke model} $\mME$ is specified as follows:
\begin{itemize}
\item $\DeltaME = \{E_x \mid x \in \DeltaM\}$, 
\item $s^{\mME} = E_{s^\mM}$ for $s \in \Ss$,
\item $p^{\mME}(E_x) = p^\mM(x)$ for $p \in \SP$ and \mbox{$x \in \DeltaM$}, 
\item $\sigma^{\mME}(E_x, E_y) = (E \circ \sigma^\mM \circ E)(x,y)$ for $\sigma \in \SA$ and $x,y \in \DeltaM$.
\myend
\end{itemize}
\end{Definition}

\begin{proposition}\label{prop: JHRLM}
The quotient fuzzy Kripke model $\mME$ of a finite fuzzy Kripke model $\mM$ with respect to its greatest fuzzy $\Phi$-auto-bisimulation~$E$ is well defined.
\end{proposition}

\begin{proof}
We need to prove that, for every $x,x',y,y' \in \DeltaM$, $p \in \SP$, and $\sigma \in \SA$, 
\begin{enumerate}[(a)]
\item\label{assertion: HGGWL 1} if $E_x = E_{x'}$, then $p^\mM(x) = p^\mM(x')$;
\item\label{assertion: HGGWL 2} if $E_x = E_{x'}$ and $E_y = E_{y'}$, then $(E \circ \sigma^\mM \circ E)(x,y) = (E \circ \sigma^\mM \circ E)(x',y')$.
\end{enumerate}

Consider the assertion~\eqref{assertion: HGGWL 1} and assume that $E_x = E_{x'}$. By Lemma~\ref{lemma: JFSJA}, we have $E(x,x') = 1$. By Theorem~\ref{theorem: invariance}, it follows that $(p^\mM(x) \fequiv p^\mM(x')) = 1$, which implies $p^\mM(x) = p^\mM(x')$.\footnote{By~\cite[Lemma~2.1]{FBSML}, $a \leq b$ iff $(a \fto b) = 1$.\label{footnote: HGJHS}} 

Consider the assertion~\eqref{assertion: HGGWL 2} and assume that $E_x = E_{x'}$ and $E_y = E_{y'}$. By Lemma~\ref{lemma: JFSJA}, we have $E(x,x') = E(y,y') = 1$. Since $E$ is a fuzzy equivalence relation, we have $E \circ E \leq E$. Consequently, 
\begin{eqnarray*}
&       (E \circ \sigma^\mM \circ E)(x,y) 
   \geq (E \circ E \circ \sigma^\mM \circ E \circ E)(x,y) \geq \\
&  \geq E(x,x') \fand (E \circ \sigma^\mM \circ E)(x',y') \fand E(y',y) 
   \geq (E \circ \sigma^\mM \circ E)(x',y').
\end{eqnarray*}
Similarly, $(E \circ \sigma^\mM \circ E)(x',y') \geq (E \circ \sigma^\mM \circ E)(x,y)$. Therefore, $(E \circ \sigma^\mM \circ E)(x,y) = (E \circ \sigma^\mM \circ E)(x',y')$, which completes the proof.
\myend
\end{proof}

\newcommand{\MforHFJSK}{
	\node (M) {$\mM$};
	\node (tau) [draw, node distance=1.0cm, below of=M] {$v_0\,(\tau)$};
	\node (v1)  [draw, node distance=2.8cm, right of=tau] {$v_1$};
	\node (v3) [draw, node distance=2.0cm, below of=tau] {$v_3\!:\!p_{\,0.8}$};
	\node (v2) [draw, node distance=2.8cm, left of=v3] {$v_2\!:\!p_{\,0.9}$};
	\node (v4) [draw, node distance=2.0cm, below of=v1] {$v_4\!:\!p_{\,0.8}$};
	\node (v5) [draw, node distance=2.0cm, below of=v3] {$v_5\!:\!p_{\,0.9}$};
	\node (v6) [draw, node distance=2.0cm, below of=v4] {$v_6\!:\!p_{\,0.9}$};
	\draw (tau) to node [left,xshift=-6pt]{\footnotesize{1}} (v2);	
	\draw (tau) to node [left]{\footnotesize{0.8}} (v3);	
	\draw (tau) to node [left,xshift=-6pt]{\footnotesize{0.9}} (v4);	
	\draw (v1) to node [left]{\footnotesize{0.2}} (v4);	
	\draw (v3) to node [left]{\footnotesize{0.7}} (v5);	
	\draw (v3) to node [left,xshift=-6pt]{\footnotesize{0.8}} (v6);	
	\draw (v4) to node [left]{\footnotesize{0.9}} (v6);	
    \draw (v5) edge[bend left=10] node[above]{\footnotesize{0.6}} (v6);
    \draw (v6) edge[bend left=10] node[below]{\footnotesize{0.8}} (v5);
    \draw (v2) edge[out=-130,in=-50,looseness=8] node[below]{\footnotesize{0.8}} (v2);
    \draw (v5) edge[out=-130,in=-50,looseness=8] node[below]{\footnotesize{0.8}} (v5);
}

\newcommand{\MPforHFJSK}{
	\node (M) {$\mME$\ for the product semantics:};
	\node (tau) [draw, node distance=1.4cm, below of=M] {$E_{v_0}\,(\tau)$};
	\node (v1)  [draw, node distance=2.8cm, right of=tau] {$E_{v_1}$};
	\node (v3) [draw, node distance=3.0cm, below of=tau] {$E_{v_3}\!:\!p_{\,0.8}$};
	\node (v3b) [node distance=1.8cm, below of=v3] {};
	\node (v2) [draw, node distance=2.8cm, left of=v3] {$E_{v_2}\!:\!p_{\,0.9}$};
	\node (v4) [draw, node distance=3.0cm, below of=v1] {$E_{v_4}\!:\!p_{\,0.8}$};
	\draw (tau) to node[left,pos=0.3,xshift=-1]{\footnotesize{1}} (v2);	
	\draw (tau) to node[left,pos=0.3,xshift=2]{$\tfrac{8}{9}$} (v3);	
	\draw (tau) to node[right,pos=0.65]{\footnotesize{0.9}} (v4);	
	\draw (v1) to node[right,xshift=-2]{\footnotesize{0.2}} (v4);	
	\draw (v1) to node[left,pos=0.6,yshift=5]{$\tfrac{8}{45}$} (v2);	
	\draw (v1) to node[left,pos=0.7,yshift=2]{$\tfrac{8}{45}$} (v3);	
	\draw (v3) edge[bend right=10] node[above,pos=0.35,yshift=-2]{\footnotesize{0.8}} (v2);	
	\draw (v2) edge[bend right=10] node[below,yshift=1]{$\tfrac{32}{45}$} (v3);	
	\draw (v3) edge[bend left=10] node[above,yshift=-2]{$\tfrac{32}{45}$} (v4);	
	\draw (v4) edge[bend left=10] node[below,yshift=1]{\footnotesize{0.8}} (v3);	
    \draw (v4) edge[bend left=48] node[above,yshift=-2]{\footnotesize{0.9}} (v2);
    \draw (v2) edge[bend right=60] node[below,yshift=1]{$\tfrac{32}{45}$} (v4);
    \draw (v2) edge[out=-60,in=-120,looseness=6] node[below,yshift=1]{\footnotesize{0.8}} (v2);
    \draw (v3) edge[out=-120,in=-60,looseness=6] node[left,xshift=-5,yshift=6]{$\tfrac{32}{45}$} (v3);
    \draw (v4) edge[out=-60,in=-120,looseness=6] node[below,yshift=1]{\footnotesize{0.8}} (v4);
}

\newcommand{\MLforHFJSK}{
	\node (M) {$\mME$\ for the \L{}ukasiewicz semantics:};
	\node (tau) [draw, node distance=1.4cm, below of=M] {$E_{v_0}\,(\tau)$};
	\node (v1)  [draw, node distance=2.8cm, right of=tau] {$E_{v_1}$};
	\node (v3) [draw, node distance=3.0cm, below of=tau] {$E_{v_3}\!:\!p_{\,0.8}$};
	\node (v3b) [node distance=2.2cm, below of=v3] {};
	\node (v2) [draw, node distance=2.8cm, left of=v3] {$E_{v_2}\!:\!p_{\,0.9}$};
	\node (v4) [draw, node distance=3.0cm, below of=v1] {$E_{v_4}\!:\!p_{\,0.8}$};
	\draw (tau) to node[above,yshift=-2]{\footnotesize{0.1}} (v1);	
	\draw (tau) to node[left,pos=0.3,xshift=-1]{\footnotesize{1}} (v2);	
	\draw (tau) to node[left,pos=0.3,xshift=2]{\footnotesize{0.9}} (v3);	
	\draw (tau) to node[right,pos=0.65]{\footnotesize{0.9}} (v4);	
	\draw (v1) to node[right,xshift=-2]{\footnotesize{0.2}} (v4);	
	\draw (v1) to node[left,pos=0.6,yshift=2]{\footnotesize{0.1}} (v2);	
	\draw (v1) to node[left,pos=0.7]{\footnotesize{0.1}} (v3);	
	\draw (v3) edge[bend right=10] node[above,pos=0.35,yshift=-2]{\footnotesize{0.8}} (v2);	
	\draw (v2) edge[bend right=10] node[below,yshift=1]{\footnotesize{0.7}} (v3);	
	\draw (v3) edge[bend left=10] node[above,yshift=-2]{\footnotesize{0.7}} (v4);	
	\draw (v4) edge[bend left=10] node[below,yshift=1]{\footnotesize{0.8}} (v3);	
    \draw (v4) edge[bend left=48] node[above,yshift=-2]{\footnotesize{0.9}} (v2);
    \draw (v2) edge[bend right=60] node[below,yshift=1]{\footnotesize{0.7}} (v4);
    \draw (v2) edge[out=-60,in=-120,looseness=6] node[below,yshift=1]{\footnotesize{0.8}} (v2);
    \draw (v3) edge[out=-120,in=-60,looseness=6] node[left,xshift=-5,yshift=7]{\footnotesize{0.7}} (v3);
    \draw (v4) edge[out=-60,in=-120,looseness=6] node[below,yshift=1]{\footnotesize{0.8}} (v4);
    \draw (tau) edge[out=120,in=60,looseness=6] node[left,xshift=-5,yshift=-8]{\footnotesize{0.1}} (tau);
}

\newcommand{\EPforHFJSK}{\begin{tabular}{|*{8}{>{\centering\arraybackslash}m{0.525cm}|}}
\hline
\multicolumn{8}{|c|}{The greatest fuzzy $\emptyset$-auto-bisimulation} \\
\multicolumn{8}{|c|}{of $\mM$ under the product semantics} \\
\hline
$E$   & $v_0$ & $v_1$ & $v_2$ & $v_3$ & $v_4$ & $v_5$ & $v_6$ \\ \hline
$v_0$ &   1   & $\scriptstyle\frac{8}{45}$ &   0   &   0   &   0   &   0   &   0   \\ \hline
$v_1$ & $\scriptstyle\frac{8}{45}$ &   1   &   0   &   0   &   0   &   0   &   0   \\ \hline
$v_2$ &   0   &   0   &   1   & $\scriptstyle\frac{8}{9}$ & $\scriptstyle\frac{8}{9}$ &   1   &   1   \\ \hline
$v_3$ &   0   &   0   & $\scriptstyle\frac{8}{9}$ &   1   & $\scriptstyle\frac{8}{9}$ & $\scriptstyle\frac{8}{9}$ & $\scriptstyle\frac{8}{9}$ \\ \hline
$v_4$ &   0   &   0   & $\scriptstyle\frac{8}{9}$ & $\scriptstyle\frac{8}{9}$ &   1   & $\scriptstyle\frac{8}{9}$ & $\scriptstyle\frac{8}{9}$ \\ \hline
$v_5$ &   0   &   0   &   1   & $\scriptstyle\frac{8}{9}$ & $\scriptstyle\frac{8}{9}$ &   1   &   1   \\ \hline
$v_6$ &   0   &   0   &   1   & $\scriptstyle\frac{8}{9}$ & $\scriptstyle\frac{8}{9}$ &   1   &   1   \\ \hline
\end{tabular}}

\newcommand{\ELforHFJSK}{\begin{tabular}{|*{8}{>{\centering\arraybackslash}m{0.525cm}|}}
\hline
\multicolumn{8}{|c|}{The greatest fuzzy $\emptyset$-auto-bisimulation} \\
\multicolumn{8}{|c|}{of $\mM$ under the \L{}ukasiewicz semantics} \\
\hline
$E$   & $v_0$ & $v_1$ & $v_2$ & $v_3$ & $v_4$ & $v_5$ & $v_6$ \\ \hline
$v_0$ &   1   & {\small 0.1} & {\small 0.1} & {\small 0.2} & {\small 0.2} & {\small 0.1} & {\small 0.1} \\ \hline
$v_1$ & {\small 0.1} &   1   & {\small 0.1} & {\small 0.2} & {\small 0.2} & {\small 0.1} & {\small 0.1} \\ \hline
$v_2$ & {\small 0.1} & {\small 0.1} &   1   & {\small 0.9} & {\small 0.9} &   1   &   1   \\ \hline
$v_3$ & {\small 0.2} & {\small 0.2} & {\small 0.9} &   1   & {\small 0.9} & {\small 0.9} & {\small 0.9} \\ \hline
$v_4$ & {\small 0.2} & {\small 0.2} & {\small 0.9} & {\small 0.9} &   1   & {\small 0.9} & {\small 0.9} \\ \hline
$v_5$ & {\small 0.1} & {\small 0.1} &   1   & {\small 0.9} & {\small 0.9} &   1   &   1   \\ \hline
$v_6$ & {\small 0.1} & {\small 0.1} &   1   & {\small 0.9} & {\small 0.9} &   1   &   1   \\ \hline
\end{tabular}}

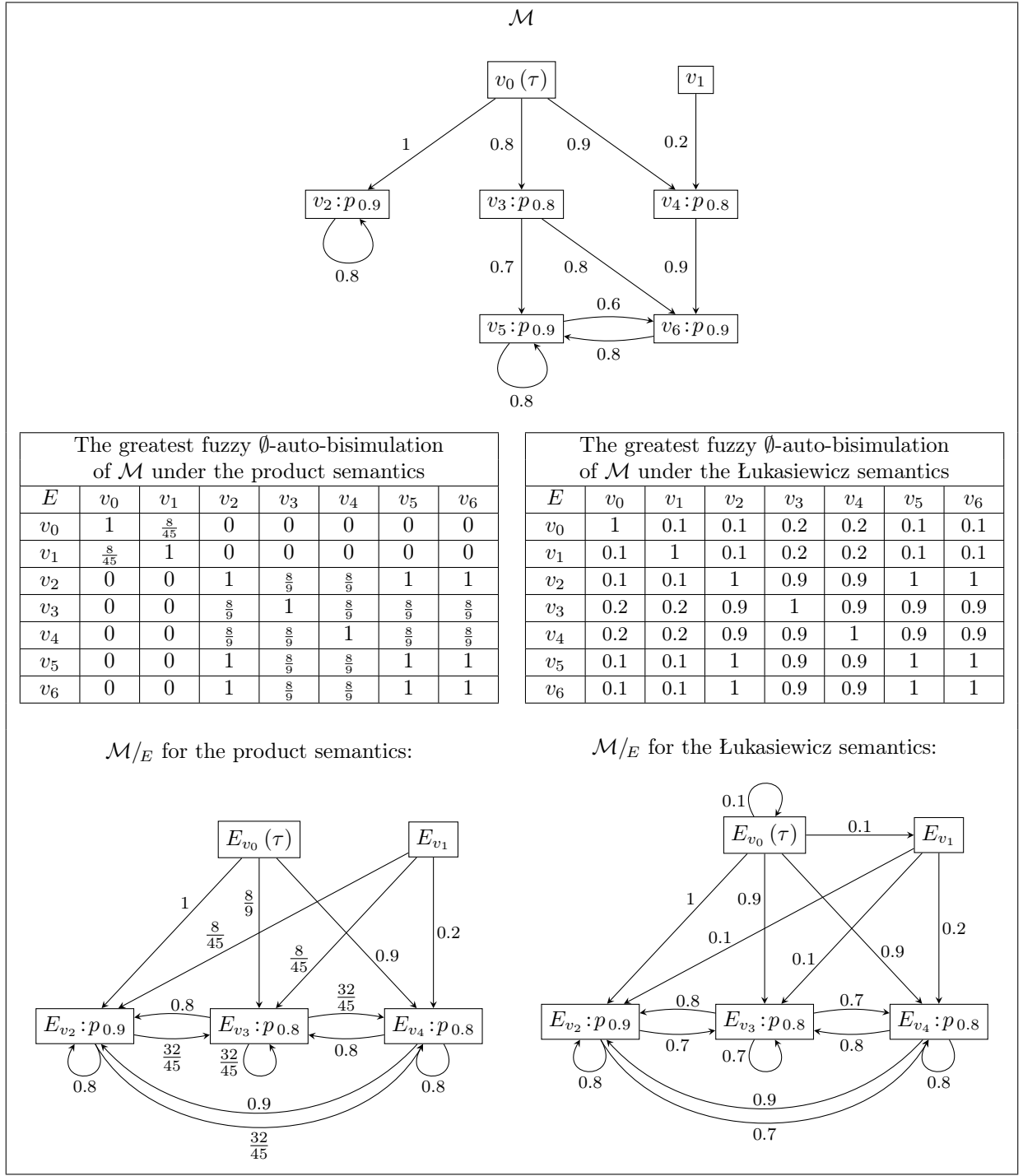
\begin{figure}[th!]
\begin{center}
\begin{tabular}{|cc|}
\hline
\multicolumn{2}{|c|}{
\begin{tikzpicture}[->,>=stealth,auto,black]
\MforHFJSK
\end{tikzpicture}
} \\[0.5em]
\EPforHFJSK
&
\ELforHFJSK
\\
& \\
%\hline
\begin{tikzpicture}[->,>=stealth,auto,black]
\MPforHFJSK
\end{tikzpicture}
& 
\begin{tikzpicture}[->,>=stealth,auto,black]
\MLforHFJSK
\end{tikzpicture}
\\
\hline
\end{tabular}

\caption{A supplement for Example~\ref{example: HFJSK}.\label{fig: HFJSK}}
\end{center}
\end{figure}

We provide two examples below to illustrate the construction of a quotient fuzzy Kripke model.

\begin{example}\label{example: HFJSK}
Let $\SA = \{\sigma\}$, $\SP = \{p\}$, $\Ss = \{\tau\}$, $\So = \emptyset$, and $\Phi \subseteq \{\cup, \circ, *, \lor, \Box\}$. Let $\mM$ be the fuzzy Kripke model illustrated at the top of Figure~\ref{fig: HFJSK} and specified as follows:
\begin{itemize}
\item $\DeltaM = \{v_0,\ldots,v_6\}$, $\tau^\mM = v_0$, 
\item $p^\mM = \{v_2\!:\!0.9$, $v_3\!:\!0.8$, $v_4\!:\!0.8$, $v_5\!:\!0.9$, $v_6\!:\!0.9\}$, 
\item $\sigma^\mM = \{\tuple{v_0,v_2}\!:\!1$, $\tuple{v_0,v_3}\!:\!0.8$, $\tuple{v_0,v_4}\!:\!0.9$, $\tuple{v_1,v_4}\!:\!0.2$, $\tuple{v_2,v_2}\!:\!0.8$, $\tuple{v_3,v_5}\!:\!0.7$, $\tuple{v_3,v_6}\!:\!0.8$, $\tuple{v_4,v_6}\!:\!0.9$, $\tuple{v_5,v_5}\!:\!0.8$, $\tuple{v_5,v_6}\!:\!0.6$, $\tuple{v_6,v_5}\!:\!0.8\}$.
\end{itemize}
The greatest fuzzy $\Phi$-auto-bisimulations of $\mM$ under the product and \L{}ukasiewicz semantics (i.e., when the underlying residuated lattice is the product or \L{}ukasiewicz structure) are specified in the middle of Figure~\ref{fig: HFJSK}. They can be verified using our implementation~\cite{MinFKM-prog}, described in Section~\ref{section: experiments}. For each of these semantics, letting $E$ denote the greatest fuzzy $\Phi$-auto-bisimulations of $\mM$, the corresponding quotient fuzzy Kripke model $\mME$ is illustrated at the bottom of Figure~\ref{fig: HFJSK}. In both cases, we have $E_{v_2} = E_{v_5} = E_{v_6}$ and $|\DeltaME| = 5$, whereas $|\DeltaM| = 7$.
\myend
\end{example}

In contrast to the previous example, the following one presents a case in which the quotient fuzzy Kripke models with respect to the greatest fuzzy $\Phi$-auto-bisimulations under the product and \L{}ukasiewicz semantics have different sizes.

\newcommand{\MforHGDHJ}{
	\node (M) {$\mM$};
	\node (v0) [draw, node distance=1.0cm, below of=M] {$v_0\,(\tau)$};
	\node (bv0) [node distance=2.4cm, below of=v0] {};
	\node (v1)  [draw, node distance=1.0cm, left of=bv0] {$v_1$};
	\node (v2)  [draw, node distance=1.0cm, right of=bv0] {$v_2$};
	\node (bv1) [node distance=2.4cm, below of=v1] {};
	\node (bv2) [node distance=2.4cm, below of=v2] {};
	\node (v3)  [draw, node distance=1.0cm, left of=bv1] {$v_3\!:\!p_{\,0.5}$};
	\node (v4)  [draw, node distance=1.0cm, right of=bv1] {$v_4\!:\!p_{\,1.0}$};
	\node (v5)  [draw, node distance=1.0cm, right of=bv2] {$v_5\!:\!p_{\,1.0}$};
	\draw (v0) to node [left,pos=0.60]{\footnotesize{0.5}} (v1);	
	\draw (v0) to node [right,pos=0.60]{\footnotesize{0.7}} (v2);	
	\draw (v1) to node [left,pos=0.30]{\footnotesize{0.3}} (v3);	
	\draw (v1) to node [left]{\footnotesize{0.6}} (v4);	
	\draw (v2) to node [right,pos=0.30]{\footnotesize{0.6}} (v5);	
	\draw[dashed] (v4) to node [right,pos=0.225]{\footnotesize{0.8}} (v0);	
    \draw[dashed] (v3) edge[bend left=35] node[left]{\footnotesize{0.5}} (v0);
    \draw[dashed] (v5) edge[bend right=35] node[right]{\footnotesize{0.8}} (v0);
}

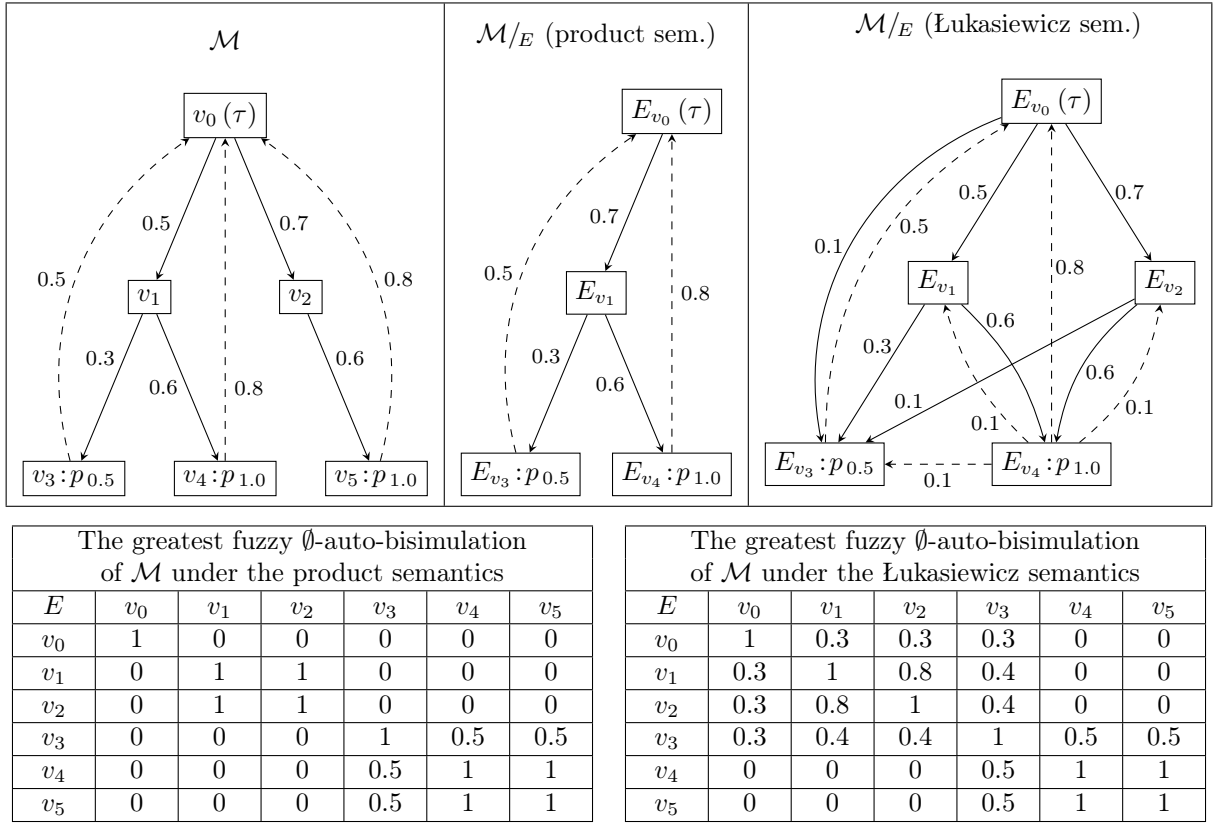
\begin{figure}[t]
\begin{center}
\begin{tabular}{|c|c|c|}
\hline
\begin{tikzpicture}[->,>=stealth,auto,black]
\MforHGDHJ
\end{tikzpicture}
&
\begin{tikzpicture}[->,>=stealth,auto,black]
	\node (M) {};
	\node (lM) [node distance=1.0cm, left of=M] {$\mME$ (product sem.)};
	\node (v0) [draw, node distance=1.0cm, below of=M] {$E_{v_0}\,(\tau)$};
	\node (bv0) [node distance=2.4cm, below of=v0] {};
	\node (v1)  [draw, node distance=1.0cm, left of=bv0] {$E_{v_1}$};
	\node (bv1) [node distance=2.4cm, below of=v1] {};
	\node (v3)  [draw, node distance=1.0cm, left of=bv1] {$E_{v_3}\!:\!p_{\,0.5}$};
	\node (v4)  [draw, node distance=1.0cm, right of=bv1] {$E_{v_4}\!:\!p_{\,1.0}$};
	\draw (v0) to node [left,pos=0.60]{\footnotesize{0.7}} (v1);	
	\draw (v1) to node [left,pos=0.30]{\footnotesize{0.3}} (v3);	
	\draw (v1) to node [left]{\footnotesize{0.6}} (v4);	
	\draw[dashed] (v4) to node [right]{\footnotesize{0.8}} (v0);	
    \draw[dashed] (v3) edge[bend left=35] node[left]{\footnotesize{0.5}} (v0);
\end{tikzpicture}
&
\begin{tikzpicture}[->,>=stealth,auto,black]
	\node (M) {};
	\node (lM) [node distance=0.7cm, left of=M] {$\mME$ (\L{}ukasiewicz sem.)};
	\node (v0) [draw, node distance=1.0cm, below of=M] {$E_{v_0}\,(\tau)$};
	\node (bv0) [node distance=2.4cm, below of=v0] {};
	\node (v1)  [draw, node distance=1.5cm, left of=bv0] {$E_{v_1}$};
	\node (v2)  [draw, node distance=1.5cm, right of=bv0] {$E_{v_2}$};
	\node (bv1) [node distance=2.4cm, below of=v1] {};
	\node (bv2) [node distance=2.4cm, below of=v2] {};
	\node (v3)  [draw, node distance=1.5cm, left of=bv1] {$E_{v_3}\!:\!p_{\,0.5}$};
	\node (v4)  [draw, node distance=1.5cm, right of=bv1] {$E_{v_4}\!:\!p_{\,1.0}$};
	\draw (v0) to node[left,xshift=1]{\footnotesize{0.5}} (v1);	
	\draw (v0) to node[right,xshift=-1]{\footnotesize{0.7}} (v2);	
    \draw (v0) edge[bend right=40] node[left,xshift=1]{\footnotesize{0.1}} (v3);
	\draw (v1) to node[left,pos=0.30,xshift=1]{\footnotesize{0.3}} (v3);	
	\draw (v1) edge[bend left=15] node[right,pos=0.12]{\footnotesize{0.6}} (v4);	
	\draw (v2) to node[left,pos=0.75,yshift=2]{\footnotesize{0.1}} (v3);	
	\draw (v2) edge[bend right=20] node[right,pos=0.55,xshift=-1]{\footnotesize{0.6}} (v4);	
    \draw[dashed] (v3) edge[bend left=30] node[right,pos=0.6]{\footnotesize{0.5}} (v0);
	\draw[dashed] (v4) to node[right,pos=0.55,xshift=-2]{\footnotesize{0.8}} (v0);
    \draw[dashed] (v4) edge[bend left=15] node[left,pos=0.15]{\footnotesize{0.1}} (v1);
    \draw[dashed] (v4) edge[bend right=20] node[right,pos=0.3]{\footnotesize{0.1}} (v2);
	\draw[dashed] (v4) to node[below]{\footnotesize{0.1}} (v3);	
\end{tikzpicture}
\\
\hline
\end{tabular}

\medskip

\begin{tabular}{cc}
\begin{tabular}{|*{7}{>{\centering\arraybackslash}m{0.66cm}|}}
\hline
\multicolumn{7}{|c|}{The greatest fuzzy $\emptyset$-auto-bisimulation} \\
\multicolumn{7}{|c|}{of $\mM$ under the product semantics} \\
\hline
$E$   & $v_0$ & $v_1$ & $v_2$ & $v_3$ & $v_4$ & $v_5$ \\ \hline
$v_0$ &   1   &   0   &   0   &   0   &   0   &   0   \\ \hline
$v_1$ &   0   &   1   &   1   &   0   &   0   &   0   \\ \hline
$v_2$ &   0   &   1   &   1   &   0   &   0   &   0   \\ \hline
$v_3$ &   0   &   0   &   0   &   1   &  0.5  &  0.5  \\ \hline
$v_4$ &   0   &   0   &   0   &  0.5  &   1   &   1   \\ \hline
$v_5$ &   0   &   0   &   0   &  0.5  &   1   &   1   \\ \hline
\end{tabular}
&
\begin{tabular}{|*{7}{>{\centering\arraybackslash}m{0.66cm}|}}
\hline
\multicolumn{7}{|c|}{The greatest fuzzy $\emptyset$-auto-bisimulation} \\
\multicolumn{7}{|c|}{of $\mM$ under the \L{}ukasiewicz semantics} \\
\hline
$E$   & $v_0$ & $v_1$ & $v_2$ & $v_3$ & $v_4$ & $v_5$ \\ \hline
$v_0$ &   1   &  0.3  &  0.3  &  0.3  &   0   &   0   \\ \hline
$v_1$ &  0.3  &   1   &  0.8  &  0.4  &   0   &   0   \\ \hline
$v_2$ &  0.3  &  0.8  &   1   &  0.4  &   0   &   0   \\ \hline
$v_3$ &  0.3  &  0.4  &  0.4  &   1   &  0.5  &  0.5  \\ \hline
$v_4$ &   0   &   0   &   0   &  0.5  &   1   &   1   \\ \hline
$v_5$ &   0   &   0   &   0   &  0.5  &   1   &   1   \\ \hline
\end{tabular}
\end{tabular}
\caption{A supplement for Example~\ref{example: HGDHJ}.\label{fig: HGDHJ}}
\end{center}
\end{figure}

\begin{example}\label{example: HGDHJ}
Let $\SA = \{\sigma,\varrho\}$, $\SP = \{p\}$, $\Ss = \{\tau\}$, $\So = \emptyset$, and $\Phi \subseteq \{\cup, \circ, *, \lor, \Box\}$. Let $\mM$ be the fuzzy Kripke model illustrated in Figure~\ref{fig: HGDHJ} and specified as follows:
\begin{itemize}
\item $\DeltaM = \{v_0,\ldots,v_5\}$, $\tau^\mM = v_0$, 
\item $p^\mM = \{v_3\!:\!0.5$, $v_4\!:\!1$, $v_5\!:\!1\}$, 
\item $\sigma^\mM = \{\tuple{v_0,v_1}\!:\!0.5$, $\tuple{v_0,v_2}\!:\!0.7$, $\tuple{v_1,v_3}\!:\!0.3$, $\tuple{v_1,v_4}\!:\!0.6$, $\tuple{v_2,v_5}\!:\!0.6\}$,
\item $\varrho^\mM = \{\tuple{v_3,v_0}\!:\!0.5$, $\tuple{v_4,v_0}\!:\!0.8$, $\tuple{v_5,v_0}\!:\!0.8\}$.
\end{itemize}
The greatest fuzzy $\Phi$-auto-bisimulations of $\mM$ under the product and \L{}ukasiewicz semantics are specified in Figure~\ref{fig: HGDHJ}. They can be verified using our implementation~\cite{MinFKM-prog}, described in Section~\ref{section: experiments}. For each of these semantics, letting $E$ denote the greatest fuzzy $\Phi$-auto-bisimulations of $\mM$, the corresponding quotient fuzzy Kripke model $\mME$ is illustrated in Figure~\ref{fig: HGDHJ}. Note that $|\DeltaME| = 4$ for the product semantics, whereas $|\DeltaME| = 5$ for the \L{}ukasiewicz semantics. 
\myend
\end{example}

The following proposition addresses the complexity of constructing the quotient fuzzy Kripke model $\mME$ when the greatest fuzzy $\Phi$-auto-bisimulation $E$ is given explicitly.

\begin{proposition}\label{prop: JFKSM}
Given a finite fuzzy Kripke model $\mM$ and its greatest fuzzy $\Phi$-auto-bisimulation~$E$, the quotient fuzzy Kripke model $\mME$ can be constructed in time $O(m n^2)$, where $n = |\Delta^\mM|$ and $m = |\{\tuple{x,\sigma,y} \in \DeltaM \times \SA \times \DeltaM : \sigma^\mM(x,y) > 0\}|$ (the number of non-zero transitions in $\mM$), assuming that the signature $\tuple{\SA,\SP,\Ss,\So}$ is fixed (i.e., has constant size).
\end{proposition}

\begin{proof}
Assuming that the elements of $\Delta^\mM$ are enumerated, each fuzzy equivalence class $E_x$ in $\DeltaMp$ can be represented by the corresponding $x$ with the smallest index. Thus, $\DeltaMp$ can be constructed in time $O(n^2)$, and $p^\mMp$, for $p \in \SP$, can be constructed in time $O(n)$.\footnote{Using the union-find technique, $\DeltaMp$ can be constructed in time $O(n\,\alpha(n))$, where $\alpha$ is the inverse Ackermann function.} The total time needed to construct the fuzzy relations $\sigma^\mMp$, for all $\sigma \in \SA$, is $O(m n^2)$. Therefore, $\mME$ can be constructed in time $O(m n^2)$.
\myend
\end{proof}

We will need the following lemma.

\begin{lemma}\label{lemma: JHLNA}
Let $\mM$ be a finite fuzzy Kripke model and let $E$ be its greatest fuzzy $\Phi$-auto-bisimulation. Then, 
\begin{enumerate}
\item if $O \in \Phi$, then $s^{\mME}(E_x) = s^\mM(x)$ for all $s \in \So$ and \mbox{$x \in \DeltaM$}, 
\item if $C \in \Phi$, then $(\sigma^{-})^{\mME}(E_x, E_y) = (E \circ (\sigma^{-})^\mM \circ E)(x,y)$ for all $\sigma \in \SA$ and $x,y \in \DeltaM$,
\item if $U \in \Phi$, then $U^{\mME}(E_x, E_y) = (E \circ U^\mM \circ E)(x,y)$ for all $x,y \in \DeltaM$.
\end{enumerate}
\end{lemma}

\begin{proof}
Consider the first assertion and assume that $O \in \Phi$. By definition and Lemma~\ref{lemma: JFSJA}, treating $s \in \So$ as a formula, we have 
\[ s^{\mME}(E_x) = (1 \textrm{ if } E_x = s^{\mME} \textrm{ else } 0) = (1 \textrm{ if } E_x = E_{s^\mM} \textrm{ else } 0) 
   = (1 \textrm{ if } E(x, s^\mM) = 1 \textrm{ else } 0). \] 
If $x = s^\mM$, then $E(x, s^\mM) = 1$, and therefore $s^{\mME}(E_x) = 1 = s^\mM(x)$. Otherwise, by Theorem~\ref{theorem: invariance}, 
\[ E(x, s^\mM) \leq (s^\mM(x) \fequiv s^\mM(s^\mM)) = s^\mM(x) = 0, \] 
and hence $s^{\mME}(E_x) = 0 = s^\mM(x)$.

Consider the second assertion and assume that $C \in \Phi$. We have 
\[ (\sigma^{-})^{\mME}(E_x, E_y) =  \sigma^{\mME}(E_y, E_x) = (E \circ \sigma^\mM \circ E)(y,x) = (E \circ \sigma^\mM \circ E)^-(x,y) = \]
\[ = (E^- \circ (\sigma^-)^\mM \circ E^-)(x,y) = (E \circ (\sigma^-)^\mM \circ E)(x,y). \]

Consider the third assertion and assume that $U \in \Phi$. We have 
\[ U^{\mME}(E_x, E_y) = 1 = U^\mM(x,y) \leq (E \circ U^\mM \circ E)(x,y). \]
Hence, $U^{\mME}(E_x, E_y) = (E \circ U^\mM \circ E)(x,y)$. 
\myend
\end{proof}

We now present a key lemma for the main result of this section.

\begin{lemma}\label{lemma: HDAMA}
Let $\mM$ be a finite fuzzy Kripke model and let $E$ be its greatest fuzzy $\Phi$-auto-bisimulation. Let \mbox{$Z: \DeltaM \times \DeltaME \to L$} be specified by \mbox{$Z(x,E_{x_2}) = E(x,x_2)$}. Then $Z$ is a fuzzy $\Phi$-bisimulation between $\mM$ and $\mME$. 
\end{lemma}

\begin{proof}
We need to prove the conditions~\eqref{eq: FB1}--\eqref{eq: FB3} for $\mMp = \mME$. 

Consider the condition~\eqref{eq: FB1} with $x' = E_{x_2}$. Using Theorem~\ref{theorem: invariance}, we have 
\[ Z(x, E_{x_2}) = E(x,x_2) \leq (\varphi^\mM(x) \fequiv \varphi^\mM(x_2)) = (\varphi^\mM(x) \fequiv \varphi^{\mME}(E_{x_2})), \]
for every basic formula $\varphi$ w.r.t.~$\Phi$, since $\varphi^\mM(x_2) = \varphi^{\mME}(E_{x_2})$. This latter equality follows from Definition~\ref{def: HDJOS} when $\varphi$ is some $p \in \SP$, and from Lemma~\ref{lemma: JHLNA} when $O \in \Phi$ and $\varphi$ is some $s \in \So$. 

Consider the condition~\eqref{eq: FB2}, where $\alpha$ is a basic program w.r.t.~$\Phi$. Let $x' = E_{x_2}$. Since $E$ is the greatest fuzzy $\Phi$-auto-bisimulation of $\mM$, there exists $y_2 \in \DeltaM$ such that 
\[ E(x,x_2) \fand \alpha^\mM(x,y) \leq \alpha^\mM(x_2,y_2) \fand E(y,y_2). \] 
Since $E(x_2,x_2) = E(y_2, y_2) = 1$, it follows that 
\begin{equation}\label{eq: JHKSJ}
E(x,x_2) \fand \alpha^\mM(x,y) \leq (E \circ \alpha^\mM \circ E)(x_2,y_2) \fand E(y,y_2).
\end{equation}
By Definition~\ref{def: HDJOS} and Lemma~\ref{lemma: JHLNA}, we have $(E \circ \alpha^\mM \circ E)(x_2,y_2) = \alpha^{\mME}(E_{x_2},E_{y_2})$. This together with~\eqref{eq: JHKSJ} gives
\[ Z(x,E_{x_2}) \fand \alpha^\mM(x,y) \leq \alpha^{\mME}(E_{x_2},E_{y_2}) \fand Z(y,E_{y_2}). \]
So, we can choose $y' = E_{y_2}$ to obtain 
\[ Z(x,x') \fand \alpha^\mM(x,y) \leq \alpha^{\mME}(x', y') \fand Z(y,y'). \]
		
Consider the condition~\eqref{eq: FB3}, where $\alpha$ is a basic program w.r.t.~$\Phi$. Let $x' = E_{x_2}$ and $y' = E_{y_2}$ with $x_2, y_2 \in \DeltaM$. 
By Definition~\ref{def: HDJOS} and Lemma~\ref{lemma: JHLNA}, there exist $x_3, y_3 \in \DeltaM$ such that 
\[ \alpha^{\mME}(E_{x_2}, E_{y_2}) = E(x_2,x_3) \fand \alpha^\mM(x_3,y_3) \fand E(y_3, y_2). \]
Since $E$ is the greatest fuzzy $\Phi$-auto-bisimulation of $\mM$, there exists $y \in \DeltaM$ such that 
\[ E(x,x_3) \fand \alpha^\mM(x_3,y_3) \leq \alpha^\mM(x,y) \fand E(y,y_3). \]
Hence, we have 
\begin{eqnarray*}
& Z(x,E_{x_2}) \fand \alpha^{\mME}(E_{x_2},E_{y_2}) 
  = E(x,x_2) \fand \alpha^{\mME}(E_{x_2},E_{y_2}) = \\
& = E(x,x_2) \fand E(x_2,x_3) \fand \alpha^\mM(x_3,y_3) \fand E(y_3, y_2) 
  \leq E(x,x_3) \fand \alpha^\mM(x_3,y_3) \fand E(y_3, y_2) \leq \\
& \leq \alpha^\mM(x,y) \fand E(y,y_3) \fand E(y_3, y_2)
  \leq \alpha^\mM(x,y) \fand E(y,y_2)
  = \alpha^\mM(x,y) \fand Z(y,E_{y_2}).
\end{eqnarray*}
In short, we have 
\[ Z(x,x') \fand \alpha^{\mME}(x',y') \leq \alpha^\mM(x,y) \fand Z(y,y'), \]
which completes the proof.
\myend
\end{proof}

\begin{corollary}\label{cor: JFWKA}
Let $\mM$ be a finite fuzzy Kripke model and let $E$ be its greatest fuzzy $\Phi$-auto-bisimulation.
Then $\mM$ and $\mME$ are $\Phi$-bisimilar.
\end{corollary}

\begin{proof}
Let $Z$ be defined as in Lemma~\ref{lemma: HDAMA}. By this lemma, $Z$ is a fuzzy $\Phi$-bisimulation between $\mM$ and $\mME$. 
For every $s \in \Ss$, we have $Z(s^\mM, s^\mME) = Z(s^\mM, E_{s^\mM}) = E(s^\mM, s^\mM) = 1$. Therefore, $\mM$ and $\mME$ are $\Phi$-bisimilar.
\myend
\end{proof}

\begin{corollary}\label{cor: JFWKA2}
Let $\mM$ be a finite fuzzy Kripke model and let $E$ be its greatest fuzzy $\Phi$-auto-bisimulation.
Then a formula $\varphi$ of \fKP is a global assertion of $\mM$ iff it is a global assertion of~$\mME$.
\end{corollary}

\begin{proof}
Let $Z$ be defined as in Lemma~\ref{lemma: HDAMA}. By this lemma, $Z$ is a fuzzy $\Phi$-bisimulation between $\mM$ and $\mME$. 
For every $x \in \DeltaM$, we have $Z(x, E_x) = E(x,x) = 1$. That is, for every $x \in \DeltaM$, there exists $x' \in \mME$ such that $Z(x,x') = 1$; and conversely, for every $x' \in \mME$, there exists $x \in \DeltaM$ such that $Z(x,x') = 1$.   

Suppose that $\varphi$ is a global assertion of $\mM$. Consider an arbitrary $x' \in \DeltaME$. There exists $x \in \DeltaM$ such that $Z(x,x') = 1$. By Theorem~\ref{theorem: invariance}, it follows that $\varphi^\mME(x') = \varphi^\mM(x) = 1$. Therefore, $\varphi$ is also a global assertion of $\mME$. 

Conversely, suppose that $\varphi$ is a global assertion of $\mME$. Consider an arbitrary $x \in \DeltaM$. There exists $x' \in \DeltaME$ such that $Z(x,x') = 1$. By Theorem~\ref{theorem: invariance}, it follows that $\varphi^\mM(x) = \varphi^\mME(x') = 1$. Therefore, $\varphi$ is also a global assertion of $\mM$. 
\myend
\end{proof}

We give below the main result of this section.

\begin{theorem}\label{theorem: HSJAO}
Suppose $\fand$ is right-continuous and $\lor \in \Phi$. Let $\mM$ be a finite fuzzy Kripke model and let $E$ be its greatest fuzzy $\Phi$-auto-bisimulation. Then $\mME$ is a minimal reduction of $\mM$ preserving global assertions in \fKP.
\end{theorem}

\begin{proof}
By Corollary~\ref{cor: JFWKA2}, it suffices to show that, if $\mMp\,$ is a fuzzy Kripke model such that, for any formula $\varphi$ of \fKP, $\varphi$ is a global assertion of $\mMp\,$ iff it is a global assertion of $\mM$, then $|\DeltaMp| \geq |\DeltaME|$. 

Since $\mM$ is finite, $\mME$ is also finite. Let $n = |\DeltaME|$ and $\DeltaME = \{u_1,\ldots,u_n\}$, where $u_i = E_{x_i}$ with $x_i \in \DeltaM$, for \mbox{$1 \leq i \leq n$}. Let $Z$ be defined as in Lemma~\ref{lemma: HDAMA}. Thus, $Z(x_i, u_i) = Z(x_i, E_{x_i}) = E(x_i,x_i) = 1$, for \mbox{$1 \leq i \leq n$}. Consider any different indices \mbox{$1 \leq i,j\leq n$}. Since $u_i \neq u_j$, which means $E_{x_i} \neq E_{x_j}$, by Lemma~\ref{lemma: JFSJA}, $E(x_i, x_j) < 1$. By Theorem~\ref{theorem: HM}, there exists a formula $\varphi_{i,j}$ of \fKP such that $\varphi_{i,j}^\mM(x_i) \neq \varphi_{i,j}^\mM(x_j)$. Since $Z(x_i, u_i) = Z(x_j, u_j) = 1$, by Lemma~\ref{lemma: HDAMA} and Theorem~\ref{theorem: invariance}, we must have $\varphi_{i,j}^\mM(x_i) = \varphi_{i,j}^\mME(u_i)$ and $\varphi_{i,j}^\mM(x_j) = \varphi_{i,j}^\mME(u_j)$. Therefore, $\varphi_{i,j}^\mME(u_i) \neq \varphi_{i,j}^\mME(u_j)$. Let \mbox{$\psi_{i,j} = (\varphi_{i,j} \to \varphi_{i,j}^{\mME}(u_i))$} if \mbox{$\varphi_{i,j}^{\mME}(u_i) < \varphi_{i,j}^{\mME}(u_j)$}, and \mbox{$\psi_{i,j} = (\varphi_{i,j}^{\mME}(u_i) \to \varphi_{i,j})$} otherwise (i.e., when \mbox{$\varphi_{i,j}^{\mME}(u_i) > \varphi_{i,j}^{\mME}(u_j)$}). We have $\psi_{i,j}^{\mME}(u_i) = 1$ and $\psi_{i,j}^{\mME}(u_j) < 1$.\footref{footnote: HGJHS} Let $\psi_i = \bigwedge_{k \neq i} \psi_{i,k}$. We have $\psi_i^{\mME}(u_i) = 1$ and $\psi_i^{\mME}(u_j) < 1$. Let $\xi = \bigvee_{1 \leq i \leq n} \psi_i$ and $\xi_i = \bigvee_{k \neq i} \psi_k$. 
Thus, $\xi^\mME(u_i) = 1$ and $\xi_i^\mME(u_i) < 1$, for all $1 \leq i \leq n$. That is, $\xi$ is a global assertion of $\mME$, but $\xi_i$ is not, for any $1 \leq i \leq n$. Note that $\xi$ and $\xi_i$, for $1 \leq i \leq n$, are formulas of \fKP. 

Suppose that $\mMp$ is a fuzzy Kripke model such that, for any formula $\varphi$ of \fKP, $\varphi$ is a global assertion of $\mMp\,$ iff it is a global assertion of $\mM$. Since $\mME$ is a reduction of $\mM$ that preserves global assertions in \fKP, we have that $\xi$ is a global assertion of $\mMp$, but $\xi_i$ is not, for any $1 \leq i \leq n$. Hence, for every $1 \leq i \leq n$, there exists $y_i \in \DeltaMp$ such that $\psi_i^\mMp(y_i) = 1$ and $\psi_j^\mMp(y_i) < 1$ for all $j \neq i$, $1 \leq j \leq n$. This implies that $|\DeltaMp| \geq n = |\DeltaME|$, which completes the proof.
\myend
\end{proof}

The following corollary directly follows from the above theorem and Proposition~\ref{prop: JFKSM}.

\begin{corollary}\label{cor: JHKSO}
Given a finite fuzzy Kripke model $\mM$ and its greatest fuzzy $\Phi$-auto-bisimulation~$E$, a minimal reduction of $\mM$ preserving global assertions in \fKP can be constructed in time $O(m n^2)$, where $n = |\Delta^\mM|$ and $m = |\{\tuple{x,\sigma,y} \in \DeltaM \times \SA \times \DeltaM : \sigma^\mM(x,y) > 0\}|$ (the number of non-zero transitions in $\mM$), assuming that the signature $\tuple{\SA,\SP,\Ss,\So}$ is fixed (i.e., has constant size).
\end{corollary}

%===============================================================================
\section{Minimizing fuzzy Kripke models while preserving local assertions}
\label{sec: min-la}

In this section, we assume that $\Ss$ is not empty. 
Given a finite fuzzy Kripke model $\mM$, we call a fuzzy Kripke model $\mMp$ a {\em minimal reduction of $\mM$ preserving local assertions in \fKP} if:
\begin{itemize}
\item $\varphi^\mMp(s^\mMp) = \varphi^\mM(s^\mM)$ for every formula $\varphi$ of \fKP and every $s \in \Ss$;
\item if $\mMdp$ is a fuzzy Kripke model such that $\varphi^\mMdp(s^\mMdp) = \varphi^\mM(s^\mM)$ for every formula $\varphi$ of \fKP and every $s \in \Ss$, then $|\DeltaMdp| \geq |\DeltaMp|$.
\end{itemize}

A fuzzy Kripke model $\mathcal{M}'$ is called a {\em minimal fuzzy Kripke model $\Phi$-bisimilar to $\mathcal{M}$} if it is $\Phi$-bisimilar to $\mathcal{M}$ and has minimal domain cardinality among all such models.

\subsection{Case $U \in \Phi$}

In this subsection, we consider the problem of minimizing a finite fuzzy Kripke model while preserving local assertions in the presence of the universal program (i.e., when $U \in \Phi$). The case without the universal program is studied in the next subsection.

Given fuzzy Kripke models $\mM$ and $\mMp$, for $x \in \DeltaM$ and $x' \in \DeltaMp$, we write $x \sim_\Phi x'$ to denote that there exists a fuzzy $\Phi$-bisimulation $Z$ between $\mM$ and $\mMp$ such that $Z(x,x') = 1$. For such $x$ and $x'$, we also write $x \equiv_\Phi x'$ to denote that $\varphi^\mM(x) = \varphi^\mMp(x')$ for all formulas $\varphi$ of \fKP. 

The following theorem is a generalization of the second assertion of \cite[Theorem~5]{minimization-by-fBS}, which was established only for the G\"odel structure.

\begin{theorem}\label{theorem: JFKSN}
Suppose $\fand$ is right-continuous and $U \in \Phi$. 
Let $\mM$ be a finite fuzzy Kripke model and let $E$ be its greatest fuzzy $(\Phi \setminus \{U\})$-auto-bisimulation. Then,  
\begin{enumerate}[(a)]
\item\label{ass: JFKSN 1} $\mME$  is a minimal fuzzy Kripke model $\Phi$-bisimilar to $\mM$,
\item\label{ass: JFKSN 2} $\mME$ is also a minimal reduction of $\mM$ preserving local assertions in~\fKP.
\end{enumerate}
\end{theorem}

\begin{proof}
By Proposition~\ref{prop: HFLMS}, $E$ is the greatest fuzzy $\Phi$-auto-bisimulation of~$\mM$. 
By Corollary~\ref{cor: JFWKA}, $\mM$ and $\mME$ are $\Phi$-bisimilar. 
By Theorem~\ref{theorem: invariance}, it follows that $\varphi^\mME(s^\mME) = \varphi^\mM(s^\mM)$ for every $s \in \Ss$ and every formula $\varphi$ of \fKP.

Since $\mM$ is finite, $\mME$ is also finite. Let $\DeltaME = \{u_1,\ldots,u_n\}$, where $u_1,\ldots,u_n$ are pairwise distinct. For each $1 \leq i \leq n$, there exists $x_i \in \DeltaM$ such that $u_i = E_{x_i}$. By Lemma~\ref{lemma: HDAMA}, we have $x_i \sim_\Phi u_i$ for all $1 \leq i \leq n$. Then, by Theorem~\ref{theorem: invariance}, it follows that $x_i \equiv_\Phi u_i$ for all $1 \leq i \leq n$. Let $i$ and $j$ be distinct indices in $1..n$. Since $u_i \neq u_j$, we have $E_{x_i} \neq E_{x_j}$, and thus, by Lemma~\ref{lemma: JFSJA}, $E(x_i, x_j) < 1$. By Theorem~\ref{theorem: HM}, it follows that $x_i \not\equiv_\Phi x_j$, and hence $u_i \not\equiv_\Phi u_j$. 

Consider the assertion~\eqref{ass: JFKSN 1}. Let $\mMp$ be a fuzzy Kripke model $\Phi$-bisimilar to $\mM$. We need to show that $|\DeltaMp| \geq n$. Without loss of generality, assume that $\mMp$ is finite. By Corollary~\ref{cor: NSLAJ}, $\mMp$ and $\mME$ are $\Phi$-bisimilar, as they are $\Phi$-bisimilar to $\mM$. Let $Z$ be the greatest fuzzy $\Phi$-bisimulation between $\mMp$ and $\mME$, and let $s \in \Ss$. We have $Z(s^\mMp, s^\mME) = 1$. For each $1 \leq i \leq n$, by the condition~\eqref{eq: FB3} with $\alpha = U$, there exists $y_i \in \DeltaMp$ such that $Z(s^\mMp, s^\mME) \leq Z(y_i,u_i)$. Thus, $Z(y_i, u_i) = 1$ for all $1 \leq i \leq n$. By Theorem~\ref{theorem: invariance}, it follows that $y_i \equiv_\Phi u_i$ for all $1 \leq i \leq n$.
For distinct indices $i$ and $j$ in $1..n$, since $u_i \not\equiv_\Phi u_j$, it follows that $y_i \not\equiv_\Phi y_j$. Therefore, $|\DeltaMp| \geq n$. 

Consider the assertion~\eqref{ass: JFKSN 2}. Let $\mMp$ be a fuzzy Kripke model such that $\varphi^\mMp(s^\mMp) = \varphi^\mM(s^\mM)$ for every formula $\varphi$ of \fKP and every $s \in \Ss$. We need to show that $|\DeltaMp| \geq n$. Without loss of generality, assume that $\mMp$ is finite. By Theorem~\ref{theorem: HM}, $\mMp$ is $\Phi$-bisimilar to $\mM$. As proved in the above paragraph, we have $|\DeltaMp| \geq n$, which completes the proof.
\myend
\end{proof}

\subsection{Case $U \notin \Phi$}

In this subsection, we present Algorithm~\ref{alg1}, which, given a finite fuzzy Kripke model $\mM$ and its greatest fuzzy $\Phi$-auto-bisimulation $E$, constructs a minimal reduction of $\mM$ that preserves local assertions in \fKP, under the assumption that $U \notin \Phi$ and $\fand$ is right-continuous and strictly monotone w.r.t.~1. Moreover, we show that the model returned by Algorithm~\ref{alg1} is also a minimal fuzzy Kripke model $\Phi$-bisimilar to~$\mM$.

The reduction, denoted by $\mMp$, has states drawn from $\mM$ with the same contents as in $\mM$. That is, for every $x \in \DeltaMp \subseteq \DeltaM$ and every $p \in \SP$, we have $p^\mMp(x) = p^\mM(x)$. The algorithm initializes $\mMp$ as follows: for each $s \in \Ss$, if $s^\mM$ is equivalent (via~$E$) to a named state already added to $\DeltaMp$, then $s^\mMp$ is set to that state; otherwise, $s^\mM$ is added to $\DeltaMp$, and $s^\mMp$ is set to it. The set $\unprocessed$ is then initialized with the states in $\DeltaMp$. 
As the main loop, while $\unprocessed$ is not empty, the algorithm extracts a state $x$ from the set and processes it as follows. For each basic program $\alpha$ w.r.t.~$\Phi$, the $\alpha$-successors of $x$ in $\mM$ -- that is, the states $y \in \Delta^\mM$ with $\alpha^\mM(x,y) > 0$ -- are considered in decreasing order with respect to $\alpha^\mM(x,y)$. 
For each such $\alpha$-successor $y$ of $x$ in~$\mM$, the algorithm proceeds as follows:
\begin{itemize}
\item If $x$ already has an $\alpha$-successor $y'$ in $\mMp$ with $\alpha^\mMp(x,y') \fand E(y',y) \geq \alpha^\mM(x,y)$, then $y$ is ignored, as $y'$ can serve the role of $y$. 
\item Otherwise, if there exists $y' \in \DeltaMp$ such that $E(y,y') = 1$, then $y'$ is used as the corresponding $\alpha$-successor of $x$ in $\mMp$, setting $\alpha^\mMp(x,y') = \alpha^\mM(x,y)$ (in the case where $\alpha = \sigma^-$, this is understood as \mbox{$\sigma^\mMp(y',x) := \sigma^\mM(y,x)$}). 

\item Otherwise, $y$ is added to $\DeltaMp$ and to $\unprocessed$, with $\alpha^\mMp(x,y)$ set to $\alpha^\mM(x,y)$.
\end{itemize}

\begin{algorithm}[t]
\caption{Minimizing a fuzzy Kripke model while preserving local assertions in \fKP \label{alg1}}
\Input{a finite fuzzy Kripke model $\mM$ and its greatest fuzzy $\Phi$-auto-bisimulation $E$.}
\Output{a minimal reduction of $\mM$ preserving local assertions in \fKP, under the assumption that $U \notin \Phi$ and $\fand$ is right-continuous and strictly monotone w.r.t.~1.}

\BlankLine

initialize $\mMp$ by setting $\DeltaMp = \emptyset$\label{step: alg1 1}\;
\ForEach{$s \in \Ss$\label{step: alg1 1b}}{
    \lIf{there exists $x \in \DeltaMp$ such that $E(s^\mM,x) = 1$}{$s^\mMp := x$\label{step: alg1 1c}}
    \lElse{add $s^\mM$ to $\DeltaMp$, set $s^\mMp = s^\mM$, and set $p^\mMp(s^\mM) = p^\mM(s^\mM)$ for all $p \in \SP$\label{step: alg1 1d}}
}
$\unprocessed := \DeltaMp$\label{step: alg1 2}\tcp*{as a hard copy}

\While{$\unprocessed \neq \emptyset$}{
    extract $x$ from $\unprocessed$\;
    \ForEach{basic program $\alpha$ w.r.t.~$\Phi$, and each $y \in \DeltaM$ with $\alpha^\mM(x,y) > 0$, in the decreasing order with respect to $\alpha^\mM(x,y)$ when restricted to~$\alpha$,\label{step: alg1 5}}{
        \If{there does not exist any $y' \in \DeltaMp$ such that $\alpha^\mMp(x,y') \fand E(y',y) \geq \alpha^\mM(x,y)$\label{step: alg1 6}}{
            \uIf{there exists $y' \in \DeltaMp$ such that $E(y,y') = 1$\label{step: alg1 7}}{
                $\alpha^\mMp(x,y') := \alpha^\mM(x,y)$\label{step: alg1 8}\tcp*{when $\alpha = \sigma^-$, update $\sigma^\mMp$ appropriately}
            }
            \Else{\label{step: alg1 9}
                add $y$ to both $\DeltaMp$ and $\unprocessed$\label{step: alg1 10}\;
                \lForEach{$p \in \SP$}{$p^\mMp(y) := p^\mM(y)$\label{step: alg1 11}}
                $\alpha^\mMp(x,y) := \alpha^\mM(x,y)$\label{step: alg1 12}\tcp*{when $\alpha = \sigma^-$, update $\sigma^\mMp$ appropriately}
            }
        }
    }
}
\Return{$\mMp$}\;
\end{algorithm}

It is easily seen that the following property always holds during the execution of Algorithm~\ref{alg1}:
\begin{equation}\label{eq: JHHAL}
\unprocessed \subseteq \DeltaMp \subseteq \DeltaM.
\end{equation}

We provide two examples below to illustrate the execution of Algorithm~\ref{alg1}. 
They can be verified using our implementation of this algorithm~\cite{MinFKM-prog}, described in Section~\ref{section: experiments}. 

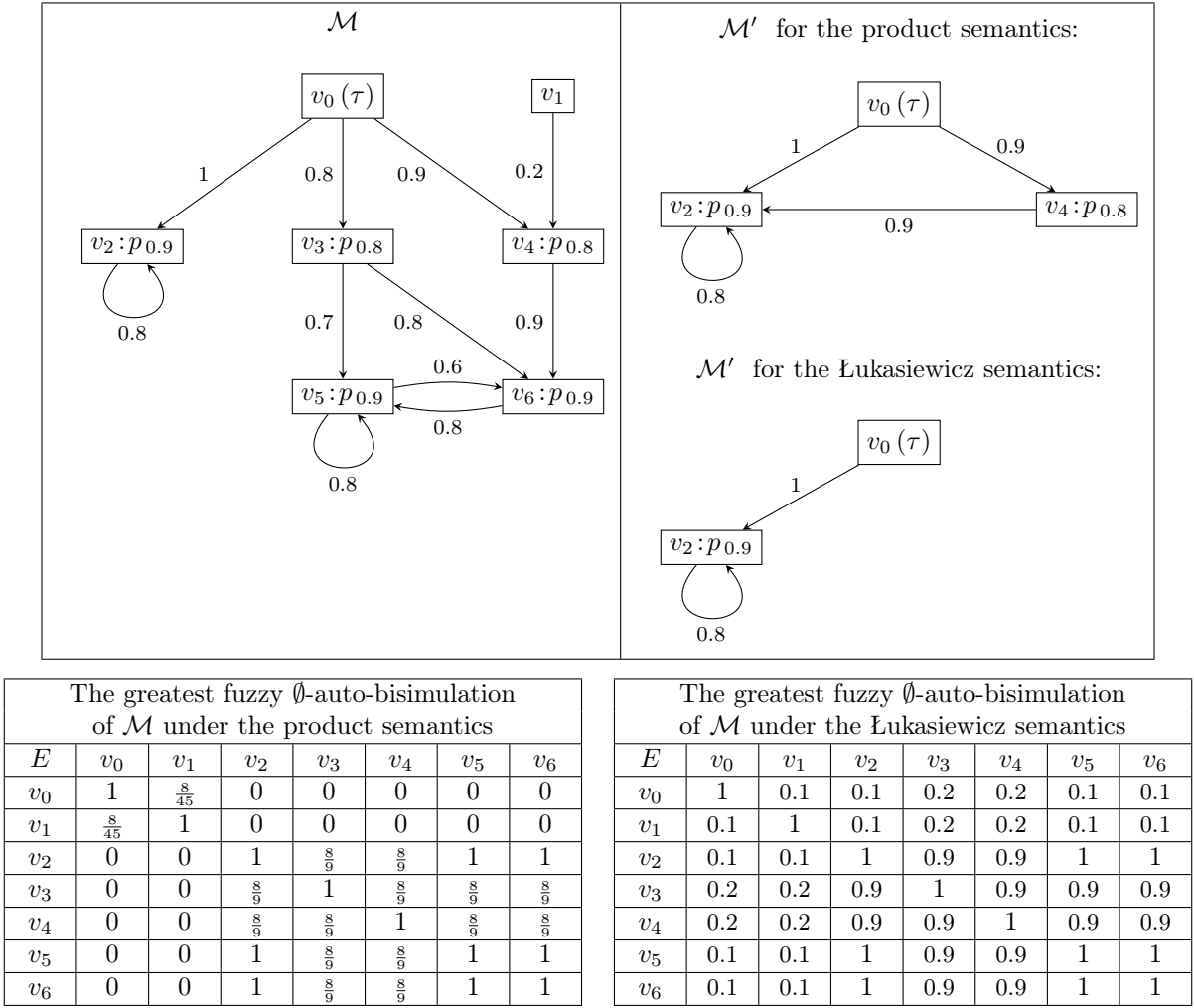
\begin{figure}[t]
\begin{center}
\begin{tabular}{|c|c|}
\hline
\begin{tikzpicture}[->,>=stealth,auto,black]
\MforHFJSK
\node (extra) [node distance=8.25cm, below of=M] {};
\end{tikzpicture}
& 
\begin{tikzpicture}[->,>=stealth,auto,black]
	\node (M) {$\mMp$\ \ for the product semantics:};
	\node (tau) [draw, node distance=1.0cm, below of=M] {$v_0\,(\tau)$};
	\node (v1)  [node distance=2.5cm, right of=tau] {};
	\node (v3) [node distance=1.4cm, below of=tau] {};
	\node (v2) [draw, node distance=2.5cm, left of=v3] {$v_2\!:\!p_{\,0.9}$};
	\node (v4) [draw, node distance=1.4cm, below of=v1] {$v_4\!:\!p_{\,0.8}$};
	\draw (tau) to node [left,pos=0.3,xshift=-5pt]{\footnotesize{1}} (v2);	
	\draw (tau) to node [right,pos=0.3,xshift=5pt]{\footnotesize{0.9}} (v4);	
	\draw (v4) to node [below]{\footnotesize{0.9}} (v2);
    \draw (v2) edge[out=-130,in=-50,looseness=8] node[below]{\footnotesize{0.8}} (v2);

	\node (lM) [node distance=4.5cm, below of=M] {$\mMp$\ \ for the \L{}ukasiewicz semantics:};
	\node (ltau) [draw, node distance=1.0cm, below of=lM] {$v_0\,(\tau)$};
	\node (lv1)  [node distance=2.5cm, right of=ltau] {};
	\node (lv3) [node distance=1.4cm, below of=ltau] {};
	\node (lv2) [draw, node distance=2.5cm, left of=lv3] {$v_2\!:\!p_{\,0.9}$};
	\draw (ltau) to node [left,pos=0.3,xshift=-5pt]{\footnotesize{1}} (lv2);	
    \draw (lv2) edge[out=-130,in=-50,looseness=8] node[below]{\footnotesize{0.8}} (lv2);
\end{tikzpicture}
\\
\hline
\end{tabular}

\medskip

\begin{tabular}{cc}
\EPforHFJSK
&
\ELforHFJSK
\end{tabular}
\caption{A supplement for Example~\ref{example: HFJSK2}.\label{fig: HFJSK2}}
\end{center}
\end{figure}

\begin{example}\label{example: HFJSK2}
Let $\SA = \{\sigma\}$, $\SP = \{p\}$, $\Ss = \{\tau\}$, $\So = \emptyset$, and $\Phi \subseteq \{\cup, \circ, *, \lor, \Box\}$. We reconsider the fuzzy Kripke model $\mM$ specified in Example~\ref{example: HFJSK} and illustrated in Figure~\ref{fig: HFJSK}. Recall that the greatest fuzzy $\Phi$-auto-bisimulations of $\mM$ under the product and {\L}ukasiewicz semantics (one for each semantics) are given in Figure~\ref{fig: HFJSK}. In this example, we examine the execution of Algorithm~\ref{alg1} for $\mM$ with respect to these bisimulations under the corresponding semantics. The resulting fuzzy Kripke models $\mMp$, one for the product semantics and one for the {\L}ukasiewicz semantics, are shown in Figure~\ref{fig: HFJSK2}. To facilitate verification of the execution of Algorithm~\ref{alg1} for $\mM$, this figure also includes the illustration of $\mM$ together with the corresponding greatest fuzzy $\Phi$-auto-bisimulations.

First, consider the execution of Algorithm~\ref{alg1} for $\mM$ and its greatest fuzzy $\Phi$-auto-bisimulation $E$ under the product semantics. After executing the statements \ref{step: alg1 1}--\ref{step: alg1 2}, we have $\DeltaMp = \{v_0\}$ and $\unprocessed = \{v_0\}$. In the first iteration of the ``while'' loop, $v_0$ is extracted from the set $\unprocessed$ and processed. Since $\sigma^\mM(v_0, v_2) > \sigma^\mM(v_0, v_4) > \sigma^\mM(v_0, v_3) > 0$, the states $v_2$, $v_4$, and $v_3$ are processed in subsequent iterations of the (inner) ``foreach'' loop. As a result, $v_2$ and $v_4$ are added to $\DeltaMp$ and to $\unprocessed$ by the statement~\ref{step: alg1 10}, with $\sigma^\mMp(v_0, v_2)$ and $\sigma^\mMp(v_0, v_4)$ set to 1 and 0.9, respectively, by the statement~\ref{step: alg1 12}, while $v_3$ is not included in $\DeltaMp$. The reason is that 
\[ \sigma^\mMp(v_0, v_2) \fand E(v_2, v_4) = \frac{8}{9} < 0.9 = \sigma^\mM(v_0, v_4), \]
\[ \sigma^\mMp(v_0, v_2) \fand E(v_2, v_3) = \frac{8}{9} > 0.8 = \sigma^\mM(v_0, v_3). \]
In the second iteration of the ``while'' loop, $v_2$ is extracted from $\unprocessed$ and processed. As a result, $\sigma^\mMp(v_2,v_2)$ is set to 0.8 by the statement~\ref{step: alg1 8} (where $y' = y = v_2$). 
In the third iteration of the ``while'' loop, $v_4$ is extracted from $\unprocessed$ and processed. As a result, $\sigma^\mMp(v_4,v_2)$ is set to 0.9 by the statement~\ref{step: alg1 8} (where $y = v_6$ and $y' = v_2$). The resulting fuzzy Kripke model $\mMp$ has three states: $v_0$, $v_2$, and $v_4$.

Now, consider the execution of Algorithm~\ref{alg1} for $\mM$ and its greatest fuzzy $\Phi$-auto-bisimulation $E$ under the \L{}ukasiewicz semantics. This case differs from that under the product semantics in that $v_4$ is not included in $\DeltaMp$. The reason is that 
\[ \sigma^\mMp(v_0, v_2) \fand E(v_2, v_4) = 0.9 = \sigma^\mM(v_0, v_4). \]
Consequently, the resulting fuzzy Kripke model $\mMp$ has only two states: $v_0$ and $v_2$.
\myend
\end{example}

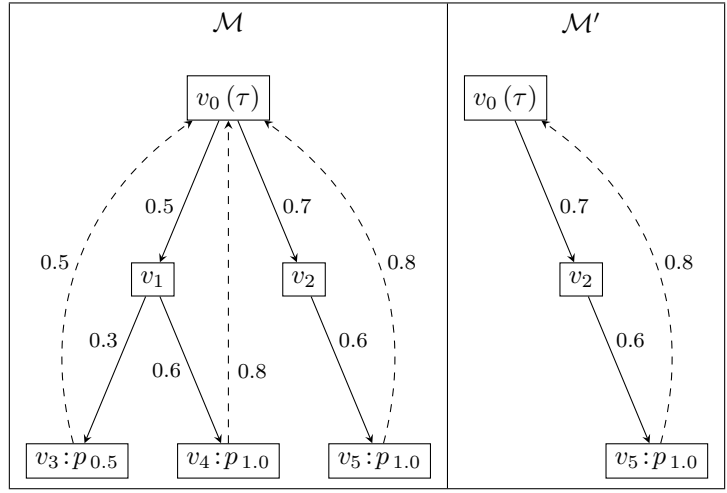
\begin{figure}[t]
\begin{center}
\begin{tabular}{|c|c|}
\hline
\begin{tikzpicture}[->,>=stealth,auto,black]
\MforHGDHJ
\end{tikzpicture}
&
\begin{tikzpicture}[->,>=stealth,auto,black]
	\node (M) {};
	\node (Mp)  [node distance=1.0cm, right of=M] {$\mMp$};
	\node (v0) [draw, node distance=1.0cm, below of=M] {$v_0\,(\tau)$};
	\node (bv0) [node distance=2.4cm, below of=v0] {};
	\node (v2)  [draw, node distance=1.0cm, right of=bv0] {$v_2$};
	\node (bv2) [node distance=2.4cm, below of=v2] {};
	\node (v5)  [draw, node distance=1.0cm, right of=bv2] {$v_5\!:\!p_{\,1.0}$};
	\draw (v0) to node [right,pos=0.60]{\footnotesize{0.7}} (v2);	
	\draw (v2) to node [right,pos=0.30]{\footnotesize{0.6}} (v5);	
    \draw[dashed] (v5) edge[bend right=35] node[right]{\footnotesize{0.8}} (v0);
\end{tikzpicture}
\\
\hline
\end{tabular}
\caption{An illustration for Example~\ref{example: HGDHJ2}.\label{fig: HGDHJ2}}
\end{center}
\end{figure}

\begin{example}\label{example: HGDHJ2}
Let $\SA = \{\sigma,\varrho\}$, $\SP = \{p\}$, $\Ss = \{\tau\}$, $\So = \emptyset$, and $\Phi \subseteq \{\cup, \circ, *, \lor, \Box\}$.  We reconsider the fuzzy Kripke model $\mM$ specified in Example~\ref{example: HGDHJ} and illustrated in Figure~\ref{fig: HGDHJ}. Recall that the greatest fuzzy $\Phi$-auto-bisimulations of $\mM$ under the product and Łukasiewicz semantics (one for each semantics) are given in Figure~\ref{fig: HGDHJ}. In this example, we examine the execution of Algorithm~\ref{alg1} for $\mM$ with respect to these bisimulations under the corresponding semantics. In both cases, the algorithm returns the same fuzzy Kripke model $\mMp$, which is illustrated in Figure~\ref{fig: HGDHJ2}. For ease of reference, this figure also includes the illustration of $\mM$.

After executing the statements \ref{step: alg1 1}--\ref{step: alg1 2}, we have $\DeltaMp = \{v_0\}$ and $\unprocessed = \{v_0\}$. In the first iteration of the ``while'' loop, $v_0$ is extracted from the set $\unprocessed$ and processed. Since $\sigma^\mM(v_0, v_2) > \sigma^\mM(v_0, v_1) > 0$, the states $v_2$ and $v_1$ are processed in subsequent iterations of the (inner) ``foreach'' loop. As a result, $v_2$ is added to $\DeltaMp$ and to $\unprocessed$ by the statement~\ref{step: alg1 10}, with $\sigma^\mMp(v_0, v_2)$ set to 0.7 by the statement~\ref{step: alg1 12}, while $v_1$ is not included in $\DeltaMp$ because 
\[ \sigma^\mMp(v_0, v_2) \fand E(v_2, v_1) \geq 0.5 = \sigma^\mM(v_0, v_1). \]
In fact, $\sigma^\mMp(v_0, v_2) \fand E(v_2, v_1)$ is equal to $0.7 \fand 1 = 0.7$ under the product semantics, and to $0.7 \fand 0.8 = 0.5$ under the \L{}ukasiewicz semantics. 
In the second iteration of the ``while'' loop, $v_2$ is extracted from $\unprocessed$ and processed. As a result, $v_5$ is added to $\DeltaMp$ and to $\unprocessed$ by the statement~\ref{step: alg1 10}, with $\sigma^\mMp(v_2, v_5)$ set to 0.6 by the statement~\ref{step: alg1 12}. 
In the third iteration of the ``while'' loop, $v_5$ is extracted from $\unprocessed$ and processed. As a result, $\varrho^\mMp(v_5, v_0)$ is set to 0.8 by the statement~\ref{step: alg1 8} (where $y' = y = v_0$). The resulting fuzzy Kripke model $\mMp$ has three states: $v_0$, $v_2$, and $v_5$.
\myend
\end{example}

The following lemma establishes a property of the fuzzy Kripke model returned by Algorithm~\ref{alg1}.

\begin{lemma}\label{lemma: KFGSB}
Suppose $U \notin \Phi$ and $\fand$ is strictly monotone w.r.t.~1. 
Let $\mMp$ be the fuzzy Kripke model returned by Algorithm~\ref{alg1} when applied to a finite fuzzy Kripke model~$\mM$ and its greatest fuzzy $\Phi$-auto-bisimulation~$E$. Then, for every basic program $\alpha$ w.r.t.~$\Phi$, and every $x,y \in \DeltaMp$ with $\alpha^\mMp(x,y) > 0$ set by the statement~\ref{step: alg1 12}, there does not exist any $y' \in \DeltaMp$ such that $y' \neq y$ and  \mbox{$\alpha^\mMp(x,y') \fand E(y',y) \geq \alpha^\mMp(x,y)$}.
\end{lemma}

\begin{proof}
Due to the condition of the ``if'' statement~\ref{step: alg1 6}, when $\alpha^\mMp(x,y)$ is assigned a positive value, there exists no $y' \in \DeltaMp$ such that $y' \neq y$ and \mbox{$\alpha^\mMp(x,y') \fand E(y',y) \geq \alpha^\mMp(x,y)$}. Furthermore, due to the decreasing order imposed by the ``foreach'' statement~\ref{step: alg1 5}, together with the condition of the ``if'' statement~\ref{step: alg1 6} and the assumption that $\fand$ is strictly monotone w.r.t.~1, when $\alpha^\mMp(x,y')$ is assigned a positive value (either by the statement~\ref{step: alg1 8}, or by the statement~\ref{step: alg1 12} where~$y$ denotes~$y'$), there exists no $y \in \DeltaMp$ such that $y \neq y'$ and \mbox{$\alpha^\mMp(x,y') \fand E(y',y) \geq \alpha^\mMp(x,y)$}. 
In fact, by the decreasing order imposed by the ``foreach'' statement~\ref{step: alg1 5} and the assumption that $\fand$ is strictly monotone w.r.t.~1, the inequality \mbox{$\alpha^\mMp(x,y') \fand E(y',y) \geq \alpha^\mMp(x,y)$} can hold only if $E(y',y) = 1$. However, due to the ``if'' statement~\ref{step: alg1 7}, there do not exist distinct $y$ and $y'$ in $\DeltaMp$ such that $E(y',y) = 1$. 
\myend
\end{proof}

The following lemma shows that  the fuzzy Kripke model returned by Algorithm~\ref{alg1} is $\Phi$-bisimilar to the input model and therefore, by Theorem~\ref{theorem: invariance}, preserves its local assertions (in \fKP).

\begin{lemma}\label{lemma: KJHZL}
Suppose $U \notin \Phi$. Let $\mMp$ be the fuzzy Kripke model returned by Algorithm~\ref{alg1} when applied to a finite fuzzy Kripke model~$\mM$ and its greatest fuzzy $\Phi$-auto-bisimulation $E$. Then $\mMp$ is $\Phi$-bisimilar to~$\mM$.
\end{lemma}

\begin{proof}
Let $Z = E|_{\DeltaM \times \DeltaMp}$. For every $s \in \Ss$, either $s^\mMp = x$ for some $x \in \DeltaMp$ with \mbox{$E(s^\mM,x) = 1$}, or $s^\mMp = s^\mM$ (and $Z(s^\mM, s^\mMp) = E(s^\mM, s^\mM) = 1$). Hence $Z(s^\mM, s^\mMp) = 1$ for all $s \in \Ss$. Thus, it suffices to prove that $Z$ is a fuzzy $\Phi$-bisimulation between~$\mM$ and~$\mMp$. 

Consider the condition~\eqref{eq: FB1} in the case $\varphi = p$, where $p \in \SP$. By the statements~\ref{step: alg1 1d} and~\ref{step: alg1 11}, we have $p^\mMp(x') = p^\mM(x')$ for all $p \in \SP$ and $x' \in \DeltaMp$. Consequently, since $E$ is a fuzzy $\Phi$-auto-bisimulation of~$\mM$, for every $p \in \SP$, $x \in \DeltaM$, and $x' \in \DeltaMp$, we have
\[ Z(x,x') = E(x,x') \leq (p^\mM(x) \fequiv p^\mM(x')) = (p^\mM(x) \fequiv p^\mMp(x')). \]

Consider the condition~\eqref{eq: FB1} in the case where $O \in \Phi$ and $\varphi = s$ with $s \in \So$. Since $E$ is a fuzzy $\Phi$-auto-bisimulation of~$\mM$, for every $s \in \So$ and $x \in \DeltaM$, it follows that $E(s^\mM,x) = 1$ only if $x = s^\mM$. Therefore, by the statements~\ref{step: alg1 1b}--\ref{step: alg1 1d}, we have $s^\mMp = s^\mM$ for all $s \in \So$. Consequently, since $E$ is a fuzzy $\Phi$-auto-bisimulation of~$\mM$, for every $s \in \So$, $x \in \DeltaM$, and $x' \in \DeltaMp$, we have
\[ Z(x,x') = E(x,x') \leq (s^\mM(x) \fequiv s^\mM(x')) = (s^\mM(x) \fequiv s^\mMp(x')). \]

Consider the condition~\eqref{eq: FB2} in the following form, where $x_0, y_0 \in \DeltaM$, $x \in \DeltaMp$ and $\alpha$ is a basic program w.r.t.~$\Phi$:
\[ \E y' \in \DeltaMp\ (Z(x_0,x) \fand \alpha^\mM(x_0,y_0) \leq \alpha^\mMp(x,y') \fand Z(y_0,y')). \]
It is equivalent to 
\begin{equation}\label{eq: JJSKW}
\E y' \in \DeltaMp\ (E(x_0,x) \fand \alpha^\mM(x_0,y_0) \leq \alpha^\mMp(x,y') \fand E(y_0,y')).
\end{equation}
Without loss of generality, assume that $0 < E(x_0,x) \fand \alpha^\mM(x_0,y_0)$. Since $E$ is a fuzzy $\Phi$-auto-bisimulation of $\mM$, there exists $y \in \DeltaM$ such that 
\begin{equation}\label{eq: JHDJW}
E(x_0,x) \fand \alpha^\mM(x_0,y_0) \leq \alpha^\mM(x,y) \fand E(y_0,y).
\end{equation}
We have $\alpha^\mM(x,y) > 0$. Consider the iteration of the ``while'' loop in which $x$ is extracted from $\unprocessed$. 
If the conditions of the ``if'' statements~\ref{step: alg1 6} and~\ref{step: alg1 7} hold, then $E(y,y') = 1$ and after executing the statement~\ref{step: alg1 8} we have $\alpha^\mMp(x,y') = \alpha^\mM(x,y)$. In this case, \eqref{eq: JJSKW} holds because, by~\eqref{eq: JHDJW} and the fact that $E$ is a fuzzy equivalence relation, we have
\begin{eqnarray*}
& E(x_0,x) \fand \alpha^\mM(x_0,y_0) \leq \alpha^\mM(x,y) \fand E(y_0,y) = \\
& = \alpha^\mMp(x,y') \fand E(y_0,y) \fand E(y,y') \leq \alpha^\mMp(x,y') \fand E(y_0,y').
\end{eqnarray*}
Consider the case where the condition of the ``if'' statement~\ref{step: alg1 6} holds, but the condition of the ``if'' statement~\ref{step: alg1 7} does not. After executing the body of the ``else'' statement~\ref{step: alg1 9}, \eqref{eq: JJSKW} follows from~\eqref{eq: JHDJW} by taking $y' = y$.
Now, consider the case where the condition of the ``if'' statements~\ref{step: alg1 6} does not hold. Thus, there exists $y' \in \DeltaMp$ such that $\alpha^\mMp(x,y') \fand E(y',y) \geq \alpha^\mM(x,y)$. By~\eqref{eq: JHDJW} and the fact that $E$ is a fuzzy equivalence relation, it follows that 
\[ E(x_0,x) \fand \alpha^\mM(x_0,y_0) \leq \alpha^\mMp(x,y') \fand E(y',y) \fand E(y_0,y) \leq \alpha^\mMp(x,y') \fand E(y_0,y'). \]
This implies~\eqref{eq: JJSKW} and completes the proof of the condition~\eqref{eq: FB2}. 

Consider the condition~\eqref{eq: FB3} in the following form, where $x_0 \in \DeltaM$, $x,y_1 \in \DeltaMp$ and $\alpha$ is a basic program w.r.t.~$\Phi$:
\[ \E y_0 \in \DeltaM\ (Z(x_0,x) \fand \alpha^\mMp(x,y_1) \leq \alpha^\mM(x_0,y_0) \fand Z(y_0,y_1)). \]
It is equivalent to 
\begin{equation}\label{eq: OITND}
\E y_0 \in \DeltaM\ (E(x_0,x) \fand \alpha^\mMp(x,y_1) \leq \alpha^\mM(x_0,y_0) \fand E(y_0,y_1)).
\end{equation}
Without loss of generality, assume that $0 < E(x_0,x) \fand \alpha^\mMp(x,y_1)$. 
Consider the iteration of the ``while'' loop in which $x$ is extracted from $\unprocessed$. The positive value of $\alpha^\mMp(x,y_1)$ is set in this iteration either by the statement~\ref{step: alg1 8} with $y' = y_1$ or by the statement~\ref{step: alg1 12} with $y = y_1$. For the former case, let $y \in \DeltaM$ be the involved element. 
Since $E$ is a fuzzy $\Phi$-auto-bisimulation of~$\mM$, there exists $y_0 \in \DeltaM$ such that 
\begin{equation}\label{eq: JRIOA}
E(x_0,x) \fand \alpha^\mM(x,y) \leq \alpha^\mM(x_0,y_0) \fand E(y_0,y).
\end{equation}

Consider the case where $\alpha^\mMp(x,y_1)$ is set by the statement~\ref{step: alg1 8} with $y' = y_1$. We have \mbox{$E(y,y') = 1$} and $\alpha^\mMp(x,y') = \alpha^\mM(x,y)$. 
Consequently, by using also~\eqref{eq: JRIOA} and the fact that $E$ is a fuzzy equivalence relation, we have the following, which implies~\eqref{eq: OITND}:
\begin{eqnarray*}
&    E(x_0,x) \fand \alpha^\mMp(x,y_1) 
   = E(x_0,x) \fand \alpha^\mMp(x,y')
   = E(x_0,x) \fand \alpha^\mM(x,y) 
   \leq \alpha^\mM(x_0,y_0) \fand E(y_0,y) = \\
&
   = \alpha^\mM(x_0,y_0) \fand E(y_0,y) \fand E(y,y')
   \leq \alpha^\mM(x_0,y_0) \fand E(y_0,y')
   = \alpha^\mM(x_0,y_0) \fand E(y_0,y_1).
\end{eqnarray*}

Consider the remaining case, where $\alpha^\mMp(x,y_1)$ is set by the statement~\ref{step: alg1 12} with $y = y_1$. We have $\alpha^\mMp(x,y_1) = \alpha^\mMp(x,y) = \alpha^\mM(x,y)$. By using~\eqref{eq: JRIOA} and the fact that $E$ is a fuzzy equivalence relation, it follows that:
\[   E(x_0,x) \fand \alpha^\mMp(x,y_1) 
   = E(x_0,x) \fand \alpha^\mM(x,y) 
   \leq \alpha^\mM(x_0,y_0) \fand E(y_0,y) 
   = \alpha^\mM(x_0,y_0) \fand E(y_0,y_1).
\]
Hence, \eqref{eq: OITND} holds. This completes the proof of the condition~\eqref{eq: FB3}, as well as the proof of the lemma.
\myend
\end{proof}

In combination with the previous lemma, the following lemma shows that the model returned by Algorithm~\ref{alg1} is a minimal fuzzy Kripke model $\Phi$-bisimilar to the input model.

\begin{lemma}\label{lemma: JORJH}
Suppose $U \notin \Phi$ and $\fand$ is right-continuous and strictly monotone w.r.t.~1. Let $\mMp$ be the fuzzy Kripke model returned by Algorithm~\ref{alg1} when applied to a finite fuzzy Kripke model~$\mM$ and its greatest fuzzy $\Phi$-auto-bisimulation $E$. Let $\mMdp$ be any fuzzy Kripke model $\Phi$-bisimilar to~$\mM$. Then $|\DeltaMdp| \geq |\DeltaMp|$.
\end{lemma}

\begin{proof}
Without loss of generality, assume that $\mMdp$ is finite. 
By Lemma~\ref{lemma: KJHZL} and Corollary~\ref{cor: NSLAJ}, $\mMp$ and $\mMdp$ are $\Phi$-bisimilar. 
Let $Z$ be the greatest fuzzy $\Phi$-bisimulation between $\mMp$ and $\mMdp$. We have $Z(s^\mMp, s^\mMdp) = 1$ for all $s \in \Ss$. 

We first prove that, for every $y \in \DeltaMp$, there exists $y'' \in \DeltaMdp$ such that \mbox{$Z(y,y'') = 1$}. 
It is done by induction on the moment when $y$ is added to $\DeltaMp$. 
Consider the case $y \neq s^\mMp$ for all $s \in \Ss$ and the context where $y$ is added to $\DeltaMp$ by the statement~\ref{step: alg1 10}. 
By~\eqref{eq: JHHAL}, $x \in \DeltaMp$. By the inductive assumption, there exists $x'' \in \DeltaMdp$ such that $Z(x,x'') = 1$. 
Since $Z$ is a fuzzy $\Phi$-bisimulation between $\mMp$ and $\mMdp$, there exists $y'' \in \DeltaMdp$ such that 
\[
   Z(x,x'') \fand \alpha^\mMp(x,y) \leq \alpha^\mMdp(x'',y'') \fand Z(y,y'').
\]
For a similar reason, there also exists $y' \in \DeltaMp$ such that 
\[
   Z(x,x'') \fand \alpha^\mMdp(x'',y'') \leq \alpha^\mMp(x,y') \fand Z(y',y'').
\]
Since $Z(x,x'') = 1$, it follows that 
\begin{equation}\label{eq: JHRIA}
\alpha^\mMp(x,y) \leq \alpha^\mMp(x,y') \fand Z(y',y'') \fand Z(y,y'').
\end{equation}
It is known that (see, e.g., \cite[Lemma~2.1]{FBSML}), for every $a,b,c \in L$, 
\[ (a \fequiv b) \fand (b \fequiv c) \leq (a \fequiv c). \]
Using this and Theorem~\ref{theorem: HM}, we can derive that 
\begin{equation}\label{eq: JHHGS}
Z(y',y'') \fand Z(y,y'') \leq \bigwedge\{\varphi^\mMp(y') \fequiv \varphi^\mMp(y) \mid \textrm{$\varphi$ is a formula of \fKP}\}.
\end{equation}
As shown in the proof of Lemma~\ref{lemma: KJHZL}, $E|_{\DeltaM \times \DeltaMp}$ is a fuzzy $\Phi$-bisimulation between $\mM$ and $\mMp$. By Theorem~\ref{theorem: invariance}, it follows that $\varphi^\mM(y) = \varphi^\mMp(y)$ and $\varphi^\mM(y') = \varphi^\mMp(y')$, for any formula $\varphi$ of \fKP. Hence, \eqref{eq: JHHGS} can be rewritten to 
\[ Z(y',y'') \fand Z(y,y'') \leq \bigwedge\{\varphi^\mM(y') \fequiv \varphi^\mM(y) \mid \textrm{$\varphi$ is a formula of \fKP}\}. \]
By Theorem~\ref{theorem: HM}, it follows that 
\[ Z(y',y'') \fand Z(y,y'') \leq E(y',y). \]
From this and~\eqref{eq: JHRIA} it follows that 
\[ \alpha^\mMp(x,y) \leq \alpha^\mMp(x,y') \fand E(y',y). \]
Consequently, by Lemma~\ref{lemma: KFGSB}, we must have $y = y'$. 
By~\eqref{eq: JHRIA} and the assumption that $\fand$ is strictly monotone w.r.t.~1, we further derive that $Z(y,y'') = 1$. 

Let $|\DeltaMp| = n$ and $\DeltaMp = \{y_1,\ldots,y_n\}$. Consider any different indices \mbox{$1 \leq i, j \leq n$}. By the conditions of the ``if'' statements~\ref{step: alg1 1c} and~\ref{step: alg1 7}, we have $E(y_i, y_j) < 1$. Consequently, by Theorem~\ref{theorem: HM}, there exists a formula $\varphi_{i,j}$ of \fKP such that $\varphi_{i,j}^\mM(y_i) \neq \varphi_{i,j}^\mM(y_j)$. Using a similar reasoning as done before, we have $\varphi_{i,j}^\mM(y_i) = \varphi_{i,j}^\mMp(y_i)$ and $\varphi_{i,j}^\mM(y_j) = \varphi_{i,j}^\mMp(y_j)$. Hence, $\varphi_{i,j}^\mMp(y_i) \neq \varphi_{i,j}^\mMp(y_j)$. 

There exist $y''_1,\ldots,y''_n \in \DeltaMdp$ such that $Z(y_i, y''_i) = 1$ for all $1 \leq i \leq n$. Consider any different indices \mbox{$1 \leq i, j \leq n$}. Since $Z$ is a fuzzy $\Phi$-bisimulation between $\mMp$ and $\mMdp$, by Theorem~\ref{theorem: invariance}, we have $\varphi_{i,j}^\mMp(y_i) = \varphi_{i,j}^\mMdp(y''_i)$ and $\varphi_{i,j}^\mMp(y_j) = \varphi_{i,j}^\mMdp(y''_j)$. Hence, $\varphi_{i,j}^\mMdp(y''_i) \neq \varphi_{i,j}^\mMdp(y''_j)$. This implies that $y''_i \neq y''_j$. Therefore, $|\DeltaMdp| \geq n$, which completes the proof.
\myend
\end{proof}

We give below the main result of this subsection.

\begin{theorem}\label{theorem: LFHZM}
Suppose $U \notin \Phi$ and $\fand$ is right-continuous and strictly monotone w.r.t.~1. Let $\mM$ be a finite fuzzy Kripke model and let $E$ be its greatest fuzzy $\Phi$-auto-bisimulation. Then the fuzzy Kripke model $\mMp$ returned by Algorithm~\ref{alg1} when applied to~$\mM$ and~$E$ is a minimal reduction of $\mM$ preserving local assertions in \fKP. Furthermore, $\mMp$ is also a minimal fuzzy Kripke model $\Phi$-bisimilar to~$\mM$. 
\end{theorem}

\begin{proof}
By Lemma~\ref{lemma: KJHZL} and Theorem~\ref{theorem: invariance}, $\varphi^\mMp(s^\mMp) = \varphi^\mM(s^\mM)$ for every formula $\varphi$ of \fKP and every $s \in \Ss$. Let $\mMdp$ be a finite fuzzy Kripke model such that $\varphi^\mMdp(s^\mMdp) = \varphi^\mM(s^\mM)$ for every formula $\varphi$ of \fKP and every $s \in \Ss$. By Theorem~\ref{theorem: HM}, $\mMdp$ and $\mM$ are $\Phi$-bisimilar. By Lemma~\ref{lemma: JORJH}, it follows that $|\DeltaMdp| \geq |\DeltaMp|$. Therefore, $\mMp$ is a minimal reduction of $\mM$ preserving local assertions in \fKP. 
By Lemmas~\ref{lemma: KJHZL} and~\ref{lemma: JORJH}, $\mMp$ is also a minimal fuzzy Kripke model $\Phi$-bisimilar to~$\mM$.
\myend
\end{proof}

The following proposition addresses the complexity of Algorithm~\ref{alg1} when the greatest fuzzy $\Phi$-auto-bisimulation of the input fuzzy Kripke model is given explicitly.

\begin{proposition}\label{prop: JHEKA}
Given a finite fuzzy Kripke model $\mM$ and its greatest fuzzy $\Phi$-auto-bisimulation~$E$ as input, Algorithm~\ref{alg1} runs in time $O(mn)$, where $n = |\Delta^\mM|$ and $m = |\{\tuple{x,\sigma,y} \in \DeltaM \times \SA \times \DeltaM : \sigma^\mM(x,y) > 0\}|$ (the number of non-zero transitions in $\mM$), assuming that $m \geq 1$ and that the signature $\tuple{\SA,\SP,\Ss,\So}$ is fixed (i.e., has constant size).
\end{proposition}

\begin{proof}
The initialization (the statements \ref{step: alg1 1}--\ref{step: alg1 2}) can be performed in time $O(1)$. 
The ``foreach'' loop~\ref{step: alg1 5} executes at most $m$ iterations in total, each of which takes time $O(n)$. Therefore, 
the ``while'' loop runs in time $O(n + mn)$. Thus, the overall running time of the algorithm is $O(mn)$, assuming $m \geq 1$.
\myend
\end{proof}

%===============================================================================
\section{Implementation and performance tests}
\label{section: experiments}

We have developed a Python package~\cite{MinFKM-prog} that implements the construction of the quotient fuzzy Kripke model $\mME$ under the product, G\"odel, and {\L}ukasiewicz semantics for a given finite fuzzy Kripke model $\mM$ and its greatest fuzzy $\Phi$-auto-bisimulation~$E$. The package also implements Algorithm~\ref{alg1} for the case where the underlying residuated lattice is the product or {\L}ukasiewicz structure. In addition, it includes a previously developed implementation for computing the greatest fuzzy $\Phi$-auto-bisimulation of a given finite fuzzy Kripke model $\mM$ under the product, G\"odel, or {\L}ukasiewicz semantics~\cite{FDSML,CompMinimaxG}.

In the following, let $\mM$ be a finite fuzzy Kripke model, and let $E$ be its greatest fuzzy $\Phi$-auto-bisimulation. Recall that, by Theorem~\ref{theorem: HSJAO}, $\mME$ is a minimal reduction of $\mM$ that preserves global assertions in \fKP, provided that $\fand$ is right-continuous and $\lor \in \Phi$. Moreover, by Theorem~\ref{theorem: JFKSN}, if $U \in \Phi$, then $\mME$ is both a minimal fuzzy Kripke model that is $\Phi$-bisimilar to $\mM$ and a minimal reduction of $\mM$ preserving local assertions in~\fKP. 
Furthermore, by Theorem~\ref{theorem: LFHZM}, if $U \notin \Phi$, then the fuzzy Kripke model returned by Algorithm~\ref{alg1} when applied to $\mM$ and $E$ is both a minimal fuzzy Kripke model that is $\Phi$-bisimilar to $\mM$ and a minimal reduction of $\mM$ preserving local assertions in~\fKP.

The module {\em MinFKMPL.py} of~\cite{MinFKM-prog} provides the following functions:
\begin{itemize}
\item \verb|computeQuotientFuzzyKripkeModel|$(\mM, \mathit{lattice}, E)$ computes and returns the quotient fuzzy Kripke model $\mME$, together with the execution time (in seconds), using the residuated lattice specified by the second parameter.

\item \verb|computeQuotientFuzzyKripkeModel_|$(\mM, \mathit{lattice})$ first computes the greatest fuzzy $\Phi$-auto-bisimulation $E$ of $\mM$ and then calls \verb|computeQuotientFuzzyKripkeModel| to construct $\mME$. The second parameter must specify the product, G\"odel, or {\L}ukasiewicz structure. The information about whether $C \in \Phi$ and/or $O \in \Phi$ is encoded in the input model $\mM$. The function returns a triple $(\mME, t_1, t_2)$, where $t_1$ is the time required to compute $E$ and $t_2$ is the additional time required to construct $\mME$, both measured in seconds.

\item \verb|minPreservingLocalAssertions_as_in_the_paper|$(\mM, \mathit{lattice}, E)$ implements Algorithm~\ref{alg1} using the residuated lattice specified by the second parameter. The information about whether $C \in \Phi$ and/or $O \in \Phi$ is encoded in the input model $\mM$. In addition to the reduced model, the function returns the execution time (in seconds).

\item \verb|minPreservingLocalAssertions_as_in_the_paper_|$(\mM, \mathit{lattice})$ first computes the greatest fuzzy $\Phi$-auto-bisimulation $E$ of $\mM$ and then calls the function described in the previous list item. The second parameter must specify either the product structure or the {\L}ukasiewicz structure. The function returns a triple $(\mMp, t_1, t_2)$, where $\mMp$ is the reduced model, $t_1$ is the time required to compute $E$, and $t_2$ is the additional time required to compute $\mMp$, both measured in seconds.

\item \verb|minPreservingLocalAssertions| and \verb|minPreservingLocalAssertions_| are improved versions of the two functions above, respectively.
\end{itemize}
The code under \verb|if __name__ == '__main__'| in the module {\em MinFKMPL.py} provides examples illustrating the use of the aforementioned functions. Executing the module {\em MinFKMPL.py} (i.e., executing that code) produces the reductions described in Examples~\ref{example: HFJSK}, \ref{example: HGDHJ}, \ref{example: HFJSK2}, and~\ref{example: HGDHJ2}.

\begin{table}[t]
\caption{Running times of computing the greatest fuzzy $\Phi$-auto-bisimulation $E$ of a randomly generated fuzzy Kripke model $\mM$ (columns $t_1$), constructing the quotient fuzzy Kripke model $\mME$ (columns $t_2$), and executing Algorithm~\ref{alg1} (columns $t_3$). To generate $\mM$, we use the parameters listed in columns 2--11 and set $\So = \Ss$. Here, $C?$, $O?$, and $A?$ indicate whether $C \in \Phi$, whether $O \in \Phi$, and whether $\mM$ is acyclic, respectively. The meanings of the remaining parameters are described in the text accompanying this table.\label{table: JHKSP}}

\smallskip

\centering
{\footnotesize
\begin{tabular}{|c||c|c|c|r|r|r|r|r|r|c||r|c|c||r|r|c|}
\hline
&&&&&&&&&&& \multicolumn{6}{c|}{Running time (ms)}\\
\cline{12-17}
$\!\!$Test$\!\!$ & $\!C?\!$ & $\!O?\!$ & $n$ & $m$ & $\!\!|\Ss|\!\!$ & $p$ & $l$ & $\!\!|\SP|\!\!$ & $\!\!|\SA|\!\!$ & $\!A?\!$ & \multicolumn{3}{c||}{product sem.} & \multicolumn{3}{c|}{{\L}ukasiewicz sem.} \\
\cline{12-17}
&&&&&&&&&&& $t_1$ & $t_2$ & $t_3$ & $t_1$ & $t_2$ & $t_3$ \\
\hline
\hline
1 & {\bf 0} & {\bf 0} & {\bf 100} & {\bf 600} & {\bf 6} & {\bf 300} & {\bf 10} & {\bf 10} & {\bf 10} & {\bf 0} & {\bf 3116} & {\bf 1} & {\bf 1} & {\bf 3703} & {\bf 179} & {\bf 2} \\ \hline
2 & {\bf 1} & 0 & 100 & 600 & 6 & 300 & 10 & 10 & 10 & 0 & 7700 & 3 & 3 & 8951 & 27 & 3 \\ \hline
3 & 0 & {\bf 1} & 100 & 600 & 6 & 300 & 10 & 10 & 10 & 0 & 3210 & 1 & 1 & 3642 & 145 & 1 \\ \hline
4 & 0 & 0 & 100 & {\bf 200} & 6 & 300 & 10 & 10 & 10 & 0 & 2673 & 1 & 0 & 2732 & 277 & 1 \\ \hline
5 & 0 & 0 & 100 & {\bf 1000} & 6 & 300 & 10 & 10 & 10 & 0 & 3605 & 2 & 2 & 4937 & 75 & 2 \\ \hline
6 & 0 & 0 & 100 & 600 & {\bf 2} & 300 & 10 & 10 & 10 & 0 & 3171 & 1 & 1 & 3691 & 179 & 2 \\ \hline
7 & 0 & 0 & 100 & 600 & {\bf 10} & 300 & 10 & 10 & 10 & 0 & 3274 & 1 & 1 & 3801 & 195 & 2 \\ \hline
8 & 0 & 0 & 100 & 600 & 6 & {\bf 100} & 10 & 10 & 10 & 0 & 3637 & 1 & 1 & 4349 & 299 & 2 \\ \hline
9 & 0 & 0 & 100 & 600 & 6 & {\bf 500} & 10 & 10 & 10 & 0 & 3122 & 1 & 1 & 3573 & 135 & 1 \\ \hline
10 & 0 & 0 & 100 & 600 & 6 & 300 & {\bf 1} & 10 & 10 & 0 & 2463 & 1 & 1 & 2504 & 2 & 1 \\ \hline
11 & 0 & 0 & 100 & 600 & 6 & 300 & {\bf 100} & 10 & 10 & 0 & 3323 & 2 & 1 & 4497 & 3024 & 3 \\ \hline
12 & 0 & 0 & 100 & 600 & 6 & 300 & 10 & {\bf 5} & 10 & 0 & 3343 & 2 & 1 & 4103 & 261 & 2 \\ \hline
13 & 0 & 0 & 100 & 600 & 6 & 300 & 10 & {\bf 15} & 10 & 0 & 3184 & 2 & 1 & 3699 & 173 & 1 \\ \hline
14 & 0 & 0 & 100 & 600 & 6 & 300 & 10 & 10 & {\bf 5} & 0 & 1417 & 1 & 1 & 2133 & 179 & 1 \\ \hline
15 & 0 & 0 & 100 & 600 & 6 & 300 & 10 & 10 & {\bf 15} & 0 & 5105 & 2 & 1 & 5421 & 156 & 1 \\ \hline
16 & 0 & 0 & 100 & 600 & 6 & 300 & 10 & 10 & 10 & {\bf 1} & 3206 & 1 & 1 & 3643 & 222 & 1 \\ \hline
\end{tabular}
}
\end{table}

The module {\em experiments3.py} of the package~\cite{MinFKM-prog} contains our performance tests for the functions provided by the module {\em MinFKMPL.py}. The tests use randomly generated fuzzy Kripke models $\mM$, whose generation is controlled by the following parameters:
\begin{itemize}
\item {\em withC} (denoted by $C?$ in Table~\ref{table: JHKSP}) specifies whether $C \in \Phi$;
\item {\em withO} (denoted by $O?$ in Table~\ref{table: JHKSP}) specifies whether $O \in \Phi$;
\item $n$ denotes the number of states in $\mM$ (i.e., $|\DeltaM|$);
\item $m$ denotes the number of non-zero transitions in $\mM$ (i.e., $m = |\{\tuple{x,\sigma,y} \in \DeltaM \times \SA \times \DeltaM :$ $\sigma^\mM(x,y) > 0\}|$);
\item $p$ denotes the number of non-zero proposition instances in $\mM$ (i.e., $|\{\tuple{p,x} \in \SP \times \DeltaM :$ $p^\mM(x) > 0\}|$);
\item $l$ denotes the number of distinct positive fuzzy values that may occur in $\mM$;
\item {\em acyclic} (denoted by $A?$ in Table~\ref{table: JHKSP}) specifies whether the graph induced by the non-zero transitions of $\mM$ is required to be acyclic;
\item $|\SP|$, $|\SA|$, and $|\Ss|$ are additional parameters, under the assumption that $\So = \Ss$ whenever $O \in \Phi$.
\end{itemize}

All experiments were carried out on a virtual machine running Debian GNU/Linux~12 (Bookworm), using a standard, non-privileged user account. The virtual machine provided access to 64 logical Intel Broadwell CPU cores clocked at approximately 2.2~GHz and 184~GiB of RAM. The account operated under the default resource limits imposed by the system.

\begin{figure}[th!]
\includegraphics[scale=0.55]{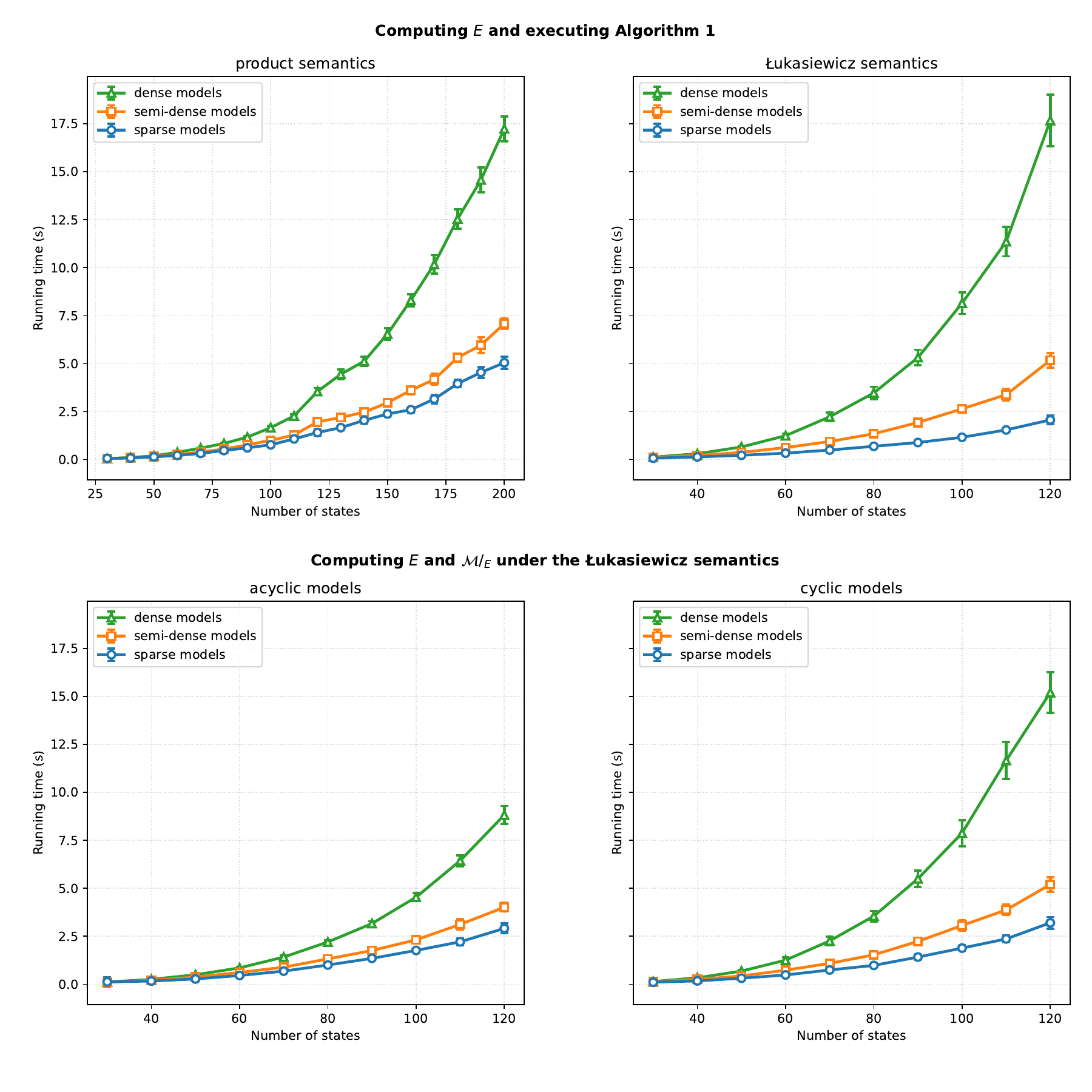}
\caption{The combined running time of computing the greatest fuzzy $\Phi$-auto-bisimulation $E$ of a randomly generated fuzzy Kripke model $\mM$ and executing Algorithm~\ref{alg1} (upper subplots), and the combined running time of computing $E$ and the quotient model $\mME$ under the {\L}ukasiewicz semantics (lower subplots). The model $\mM$ is generated using the parameters $\mathit{withC} = \mathit{withO} = \mathit{False}$, $p = n$, $l = 10$, $|\SP| = |\SA| = 3$, $|\Ss| = [\log n]$, $\So = \Ss$, and, for the upper subplots, $\mathit{acyclic} = \mathit{False}$, where $n$ denotes the number of states of $\mM$. Sparse, semi-dense, and dense models correspond to $m = 3n$, $m = [n\log n]$, and $m = [n^2/5]$, respectively.\label{fig: JEKSN}}
\end{figure}

Table~\ref{table: JHKSP} presents the results of 16 performance tests. It reports the running times of computing the greatest fuzzy $\Phi$-auto-bisimulation $E$ of a randomly generated fuzzy Kripke model $\mM$ (columns~$t_1$), constructing the quotient fuzzy Kripke model $\mME$ given~$E$ (columns~$t_2$), and executing Algorithm~\ref{alg1} (columns~$t_3$), under the product and {\L}ukasiewicz semantics. Each test was repeated 30 times, and the values of $t_1$, $t_2$, and $t_3$ are the rounded averages over these repetitions. Each of Tests~2--16 differs from Test~1 in that exactly one parameter used to generate $\mM$, except~$n$, is modified. The table leads to the following observations:
\begin{itemize}
\item The running times show that $t_3$ (the execution time of Algorithm~\ref{alg1}) is much smaller than $t_1$ (the time required to compute the greatest fuzzy $\Phi$-auto-bisimulation~$E$ of~$\mM$). Likewise, under the product semantics, $t_2$ (the time required to construct $\mME$ given~$E$) is much smaller than $t_1$. Under the {\L}ukasiewicz semantics, however, $t_2$ is considerably larger than $t_3$, although it remains smaller than~$t_1$.

\item The values of $t_1$ further show that computing $E$ is consistently faster under the product semantics than under the {\L}ukasiewicz semantics.

\item Test~2 shows that computing $E$ takes longer when converse modalities are included (i.e., when $C \in \Phi$). This is because the number of basic programs doubles, and the graph induced by the instances of basic programs contains substantially more cycles.

\item Test~3 suggests that the presence of nominals (i.e., when $O \in \Phi$) generally has little effect on the time required to compute~$E$.

\item Tests~4 and~5 demonstrate that increasing $m$ leads to a noticeable increase in the time required to compute~$E$.

\item Tests~8 and~9 indicate that increasing $p$ generally does not increase the time required to compute~$E$.

\item Tests~10 and~11 show that changing $l$ noticeably affects $t_1$ and, under the {\L}ukasiewicz semantics, has an even more pronounced effect on $t_2$. The unusually large value of $t_2$ in Test~11 under the {\L}ukasiewicz semantics is explained by the fact that the corresponding relation $E$ is much denser than in the other tests or under the product semantics. Recall that, by Corollary~\ref{cor: JHKSO}, $t_2$ is bounded by $O(mn^2)$. Comparing $t_1$ and $t_2$ also yields an interesting observation: although the implementation for computing $E$ when {\em acyclic = False} may have exponential running time in the worst case, the implemented optimizations substantially accelerate the computation, so that $t_1$ in Test~11 under the {\L}ukasiewicz semantics is only about one and a half times larger than~$t_2$.

\item Tests~12 and~13 show that increasing $|\SP|$ does not noticeably affect the time required to compute~$E$.

\item Tests~14 and~15 demonstrate that increasing $|\SA|$ leads to a noticeable increase in the time required to compute~$E$.

\item Finally, Test~16 shows that allowing cycles does not noticeably increase the time required to compute~$E$ under the product semantics, whereas it leads to a slight increase under the {\L}ukasiewicz semantics.
\end{itemize}

Figure~\ref{fig: JEKSN} presents the results of additional performance tests. The upper subplots show the combined running time of computing the greatest fuzzy $\Phi$-auto-bisimulation $E$ of a randomly generated fuzzy Kripke model $\mM$ and executing Algorithm~\ref{alg1}, whereas the lower subplots show the combined running time of computing $E$ and the quotient model $\mME$ under the {\L}ukasiewicz semantics. In these experiments, the main varying parameter is $n$, the number of states of~$\mM$, while $m$ is chosen as $3n$, $[n\log n]$, or $[n^2/5]$, corresponding to sparse, semi-dense, and dense models, respectively. The remaining parameters are as specified in the caption of Figure~\ref{fig: JEKSN}. The results confirm the observations drawn from Table~\ref{table: JHKSP}: the combined time required to compute $E$ and execute Algorithm~\ref{alg1} is consistently smaller under the product semantics than under the {\L}ukasiewicz semantics; increasing $m$ leads to a noticeable increase in the time required to compute~$E$; and allowing cycles results in a noticeable increase in the running time under the {\L}ukasiewicz semantics.

We do not present subplots showing the combined running time of computing the greatest fuzzy $\Phi$-auto-bisimulation $E$ of a randomly generated fuzzy Kripke model $\mM$ and constructing the quotient model $\mME$ under the product semantics, since, according to the experimental results from Table~\ref{table: JHKSP}, this running time is consistently smaller than that under the {\L}ukasiewicz semantics. Moreover, there is little difference in the running time between the cases with and without cycles under the product semantics.

To replicate the performance tests, run the modules {\em experiments3.py}, {\em process\_results3a.py}, and {\em process\_results3bc.py} from the package~\cite{MinFKM-prog}. The module {\em experiments3.py} generates the files {\em results3a.txt}, {\em results3b.txt}, and {\em results3c.txt}. The module {\em process\_results3a.py} processes {\em results3a.txt} to produce the data reported in Table~\ref{table: JHKSP}. The module {\em process\_results3bc.py} processes {\em results3b.txt} and {\em results3c.txt} to produce the plots shown in Figure~\ref{fig: JEKSN}.

%===============================================================================

\section{Discussion and conclusions}
\label{sec: conc}

In this work, we have studied the minimization of finite fuzzy Kripke models over residuated lattices while preserving global or local assertions in fuzzy modal logics~\fKP, which are parameterized by the underlying residuated lattice $\mL$ and a set $\Phi \subseteq \{\cup, \circ, *, \lor, \Box, C, U, O\}$ of optional features. The logic~\fKP with $\Phi = \{\cup, \circ, *, \lor, \Box\}$ corresponds to fuzzy propositional dynamic logic, whereas the additional features $C$, $U$, and $O$ stand for converse programs, the universal program, and nominals, respectively.

We have proved that, given a finite fuzzy Kripke model $\mM$, if the multiplication operator $\fand$ of the underlying residuated lattice is right-continuous and $\lor \in \Phi$, then the quotient $\mME$ induced by the greatest fuzzy $\Phi$-auto-bisimulation $E$ of $\mM$ is a minimal reduction of $\mM$ preserving global assertions in~\fKP (Theorem~\ref{theorem: HSJAO}). When $\fand$ is right-continuous and $U \in \Phi$, $\mME$ is also a minimal reduction of $\mM$ preserving local assertions, as well as a minimal fuzzy Kripke model $\Phi$-bisimilar to~$\mM$ (Theorem~\ref{theorem: JFKSN}).

In the case where $U \notin \Phi$ and $\fand$ is right-continuous and strictly monotone with respect to~1 (which holds in the product and {\L}ukasiewicz structures), we have provided the first algorithm (Algorithm~\ref{alg1}) that, given a finite fuzzy Kripke model $\mM$ and its greatest fuzzy $\Phi$-auto-bisimulation, constructs a minimal reduction of $\mM$ preserving local assertions in~\fKP. The resulting model is also a minimal fuzzy Kripke model $\Phi$-bisimilar to~$\mM$ (Theorem~\ref{theorem: LFHZM}).

Our results extend previous work on the minimization of fuzzy interpretations under the G\"odel semantics~\cite{minimization-by-fBS,nguyen2026approximateminimizationinterpretationsfuzzy} to a broader class of residuated lattices and to the setting of fuzzy Kripke models. The quotient construction $\mME$ and the techniques used to prove Theorems~\ref{theorem: HSJAO} and~\ref{theorem: JFKSN} generalize those developed in~\cite{minimization-by-fBS}, which considered fuzzy interpretations under the G\"odel semantics. In contrast, Algorithm~\ref{alg1} and the techniques used to prove Theorem~\ref{theorem: LFHZM} are novel and are not derived from the techniques developed in~\cite{minimization-by-fBS,nguyen2026approximateminimizationinterpretationsfuzzy}.

From an algorithmic perspective, assuming that the greatest fuzzy auto-bisimulation is given, the quotient construction can be performed in time $O(mn^2)$, while Algorithm~\ref{alg1} runs in time $O(mn)$, where $n$ is the number of states and $m$ is the number of non-zero transitions. The main computational bottleneck is therefore the computation of the greatest fuzzy auto-bisimulation $E$ of the input fuzzy Kripke model. For the G\"odel structure, $E$ can be computed efficiently, in time \mbox{$O(m\cdot\log{l}\cdot\log{n} + n^2)$}~\cite{DBLP:journals/isci/Nguyen23}.\footnote{The corresponding fuzzy partition can be computed in time \mbox{$O((m \log{l} + n)\log{n})$}~\cite{DBLP:journals/isci/Nguyen23}.} However, it is currently not known whether $E$ can be computed in polynomial time when the underlying residuated lattice is the product or {\L}ukasiewicz structure.

We have developed a Python package~\cite{MinFKM-prog} that implements Algorithm~\ref{alg1} for the product and {\L}ukasiewicz structures, as well as the construction of the quotient $\mME$ under the product, G\"odel, and {\L}ukasiewicz semantics for a given finite fuzzy Kripke model $\mM$ and its greatest fuzzy $\Phi$-auto-bisimulation~$E$. The package also includes a previously developed implementation for computing the greatest fuzzy $\Phi$-auto-bisimulation of a given finite fuzzy Kripke model under the product, G\"odel, and {\L}ukasiewicz semantics. We have also conducted performance tests and reported and discussed the obtained results.

As mentioned above, the time complexity of computing the greatest fuzzy $\Phi$-auto-bisimulation of a finite fuzzy Kripke model over the product or {\L}ukasiewicz structure remains an open problem: it is not known whether this problem can be solved in polynomial time or is NP-hard. Another problem worthy of study is the minimization of fuzzy Kripke models while preserving assertions in restricted classes of formulas, such as positive existential modal formulas or positive modal formulas. These classes are related to (fuzzy or crisp) simulations and directed simulations, respectively.

%===============================================================================

\biboptions{sort&compress}
\bibliography{BSfDL}

@article{BASAK2002223,
	title = {On quotient machines of a fuzzy automaton and the minimal machine},
	journal = {Fuzzy Sets Syst.},
	volume = {125},
	number = {2},
	pages = {223-229},
	year = {2002},
	doi = {10.1016/S0165-0114(01)00064-1},
	author = {N. C. Basak and A. Gupta},
}

@article{Belohlvek2009OnAM,
	title={On Approximate Minimization of Fuzzy Automata},
	author={R. B{\v e}lohl{\'a}vek and M. Krupka},
	journal={J. Multiple Valued Log. Soft Comput.},
	year={2009},
	volume={15},
	pages={125-135},
	doi={10.1142/9789812709677_0194}
}

@article{DEMENDIVILGRAU2024109108,
	title = {Polynomial crisp-minimization algorithm for fuzzy deterministic automata},
	journal = {Fuzzy Sets Syst.},
	volume = {495-496},
	pages = {109108},
	year = {2024},
	doi = {10.1016/j.fss.2024.109108},
	author = {A. G. {de Mendívil Grau} and F. Fariña and S. Stanimirović and I. Micić and J. R. G. {de Mendívil}},
}

@article{Halamish2015,
	author = {Halamish, S. and Kupferman, O.},
	title = {Minimizing Deterministic Lattice Automata},
	year = {2015},
	issue_date = {March 2015},
	publisher = {Association for Computing Machinery},
	address = {New York, NY, USA},
	volume = {16},
	number = {1},
         pages        = {1:1--1:21},
	doi = {10.1145/2631915},
	journal = {ACM Trans. Comput. Logic}
}

@ARTICLE{LITFS2015,
	author={Li, L. and Qiu, D.},
	journal={{IEEE} Trans. Fuzzy Syst.}, 
	title={On the State Minimization of Fuzzy Automata}, 
	year={2015},
	volume={23},
	number={2},
	pages={434-443},
	doi={10.1109/TFUZZ.2014.2315620}
}

@article{LI20071423,
	title = {Minimization of lattice finite automata and its application to the decomposition of lattice languages},
	journal = {Fuzzy Sets Syst.},
	volume = {158},
	number = {13},
	pages = {1423-1436},
	year = {2007},
	doi = {10.1016/j.fss.2007.03.003},
	author = {Y. Li and W. Pedrycz},
}

@article{ShamsizadehZG24,
  author       = {M. Shamsizadeh and
                  M. M. Zahedi and
                  M. J. A. Goki},
  title        = {Minimization of hesitant {L}-fuzzy automaton},
  journal = {Fuzzy Sets Syst.},
  volume       = {477},
  pages        = {108793},
  year         = {2024},
  _url          = {https://doi.org/10.1016/j.fss.2023.108793},
  doi          = {10.1016/J.FSS.2023.108793},
  _timestamp    = {Tue, 07 May 2024 20:21:38 +0200},
  _biburl       = {https://dblp.org/rec/journals/fss/ShamsizadehZG24.bib},
  _bibsource    = {dblp computer science bibliography, https://dblp.org}
}

@article{SCB.18,
    title = {Ranks of fuzzy matrices. {A}pplications in state reduction of fuzzy automata},
    journal = {Fuzzy Sets Syst.},
    volume = {333},
    pages = {124-139},
    year = {2018},
    doi = {10.1016/j.fss.2017.05.028},
    author = {A. Stamenkovi\'c and M. \'Ciri\'c and M. Ba\v{s}i\'c},
}

@article{SCI.14,
    title = {Reduction of fuzzy automata by means of fuzzy quasi-orders},
    journal = {Inf. Sci.},
    volume = {275},
    pages = {168-198},
    year = {2014},
    doi = {10.1016/j.ins.2014.02.028},
    author = {A. Stamenković and M. Ćirić and J. Ignjatović},
}

@article{Fan15,
  author    = {T.-F. Fan},
  title     = {Fuzzy Bisimulation for {G\"{o}del} Modal Logic},
  journal   = {{IEEE} Trans. Fuzzy Systems},
  doi          = {10.1109/TFUZZ.2015.2426724},
  volume    = {23},
  number    = {6},
  pages     = {2387--2396},
  year      = {2015}
}

@article{EleftheriouKN12,
  author    = {P.E. Eleftheriou and
               C.D. Koutras and
               C. Nomikos},
  title     = {Notions of Bisimulation for {Heyting}-Valued Modal Languages},
  journal   = {J. Log. Comput.},
  doi          = {10.1093/LOGCOM/EXQ005},
  volume    = {22},
  number    = {2},
  pages     = {213--235},
  year      = {2012}
}

@article{CiricIDB12,
  author    = {M. {\'C}iri{\'c} and
               J. Ignjatovi{\'c} and
               N. Damljanovi{\'c} and
               M. Ba\u{s}ic},
  title     = {Bisimulations for fuzzy automata},
  journal   = {Fuzzy Sets Syst.},
  doi          = {10.1016/J.FSS.2011.07.003},
  volume    = {186},
  number    = {1},
  pages     = {100--139},
  year      = {2012}
}

@article{ai/FanL14,
  author    = {T.{-}F. Fan and
               C.{-}J. Liau},
  title     = {Logical characterizations of regular equivalence in weighted social
               networks},
  journal   = {Artif. Intell.},
  volume    = {214},
  pages     = {66--88},
  year      = {2014},
  doi       = {10.1016/j.artint.2014.05.007}
}

@inproceedings{IgnjatovicCS15,
  author    = {J. Ignjatovi{\'c} and
               M. {\'C}iri{\'c} and
               I. Stankovi{\'c}},
  title     = {Bisimulations in fuzzy social network analysis},
  booktitle = {Proceedings of IFSA-EUSFLAT-15},
  pages={404--411},
  publisher = {Atlantis Press},
  year      = {2015}
}

@article{minimization-by-fBS,
  author    = {L. A. Nguyen and N. T. Nguyen},
  title     = {Minimizing interpretations in fuzzy description logics under the {G\"{o}del}
               semantics by using fuzzy bisimulations},
  journal   = { J. Intell. Fuzzy Syst.},
  volume    = {37},
  number    = {6},
  pages     = {7669--7678},
  doi       = {10.3233/JIFS-179371},
  year      = {2019}
}

@article{fss/WuD16,
  author    = {H. Wu and
               Y. Deng},
  title     = {Logical characterizations of simulation and bisimulation for fuzzy
               transition systems},
  journal   = {Fuzzy Sets Syst.},
  volume    = {301},
  pages     = {19--36},
  doi          = {10.1016/J.FSS.2015.09.012},
  year      = {2016}
}

@article{DBLP:journals/fss/WuCBD18,
  author    = {H. Wu and
               Y. Chen and
               T.{-}M. Bu and
               Y. Deng},
  title     = {Algorithmic and logical characterizations of bisimulations for non-deterministic
               fuzzy transition systems},
  journal   = {Fuzzy Sets Syst.},
  volume    = {333},
  pages     = {106--123},
  doi          = {10.1016/J.FSS.2017.02.008},
  year      = {2018}
}

@article{FSS2020,
  author    = {L. A. Nguyen and
               Q.-T. Ha and                
               N. T. Nguyen and
	       T.H.K. Nguyen and
               T.-L. Tran},
  title     = {Bisimulation and bisimilarity for fuzzy description logics under the {G\"odel} semantics},
  journal   = {Fuzzy Sets Syst.},
  volume    = {388},
  pages     = {146--178},
  doi       = {10.1016/j.fss.2019.08.004},
  year      = {2020}
}

@article{DBLP:journals/ijar/WuCHC18,
  author    = {H. Wu and
               T. Chen and
               T. Han and
               Y. Chen},
  title     = {Bisimulations for fuzzy transition systems revisited},
  journal   = {Int. J. Approx. Reason.},
  volume    = {99},
  pages     = {1--11},
  year      = {2018},
  doi       = {10.1016/j.ijar.2018.04.010}
}

@article{aml/MartiM18,
  author    = {M. Marti and
               G. Metcalfe},
  title     = {Expressivity in chain-based modal logics},
  journal   = {Arch. Math. Log.},
  volume    = {57},
  number    = {3-4},
  pages     = {361--380},
  year      = {2018},
  doi       = {10.1007/s00153-017-0573-4}
}

@article{fuin/Diaconescu20,
  author    = {D. Diaconescu},
  title     = {Modal Equivalence and Bisimilarity in Many-valued Modal Logics with Many-valued Accessibility Relations},
  journal   = {Fundam. Informaticae},
  volume    = {173},
  number    = {2-3},
  pages     = {177--189},
  year      = {2020},
  doi       = {10.3233/FI-2020-1920}
}

@article{StanimirovicMC.22,
  author="S. Stanimirovi{\'c} and I. Mici{\'c} and M. {\'C}iri{\'c}",
  title="Approximate Bisimulations for Fuzzy Automata over Complete {Heyting} Algebras", 
  journal = {{IEEE} Trans. Fuzzy Syst.}, 
  volume    = {30},
  number    = {2},
  pages     = {437--447},
  year      = {2022},
  _url       = {https://doi.org/10.1109/TFUZZ.2020.3039968},
  doi       = {10.1109/TFUZZ.2020.3039968},
  _timestamp = {Tue, 15 Feb 2022 10:48:26 +0100},
  _biburl    = {https://dblp.org/rec/journals/tfs/StanimirovicMC22.bib},
  _bibsource = {dblp computer science bibliography, https://dblp.org}   
}

@article{YANG202172,
	title = {Fuzzy $\epsilon$-approximate regular languages and minimal deterministic fuzzy automata $\epsilon$-accepting them},
	journal = {Fuzzy Sets Syst.},
	volume = {420},
	pages = {72-86},
	year = {2021},
	doi = {10.1016/j.fss.2020.08.010},
	_url = {https://www.sciencedirect.com/science/article/pii/S016501142030316X},
	author = {C. Yang and Y. Li}
}

@article{DBLP:journals/tfs/QiaoZF23,
  author       = {S. Qiao and
                  P. Zhu and
                  J. Feng},
  title        = {Fuzzy Bisimulations for Nondeterministic Fuzzy Transition Systems},
  journal      = {{IEEE} Trans. Fuzzy Syst.},
  volume       = {31},
  number       = {7},
  pages        = {2450--2463},
  year         = {2023},
  _url          = {https://doi.org/10.1109/TFUZZ.2022.3227400},
  doi          = {10.1109/TFUZZ.2022.3227400},
  _timestamp    = {Thu, 08 Feb 2024 10:33:06 +0100},
  _biburl       = {https://dblp.org/rec/journals/tfs/QiaoZF23.bib},
  _bibsource    = {dblp computer science bibliography, https://dblp.org}
}

@article{DBLP:journals/isci/Nguyen23,
  author       = {L. A. Nguyen},
  title        = {Computing the fuzzy partition corresponding to the greatest fuzzy
                  auto-bisimulation of a fuzzy graph-based structure under the {G{\"o}del}
                  semantics},
  journal      = {Inf. Sci.},
  volume       = {630},
  pages        = {482--506},
  year         = {2023},
  _url          = {https://doi.org/10.1016/j.ins.2023.02.029},
  doi          = {10.1016/J.INS.2023.02.029},
  _timestamp    = {Fri, 24 Mar 2023 08:43:31 +0100},
  _biburl       = {https://dblp.org/rec/journals/isci/Nguyen23.bib},
  _bibsource    = {dblp computer science bibliography, https://dblp.org}
}

@article{FBSML,
  author       = {L. A. Nguyen},
  title        = {Logical characterizations of fuzzy bisimulations in fuzzy modal logics
                  over residuated lattices},
  journal      = {Fuzzy Sets Syst.},
  volume       = {431},
  pages        = {70--93},
  year         = {2022},
  _url          = {https://doi.org/10.1016/j.fss.2021.08.009},
  doi          = {10.1016/J.FSS.2021.08.009},
  _timestamp    = {Fri, 25 Feb 2022 09:52:00 +0100},
  _biburl       = {https://dblp.org/rec/journals/fss/Nguyen22.bib},
  _bibsource    = {dblp computer science bibliography, https://dblp.org}
}

@BOOK{Belohlavek2002,
  title  = {Fuzzy Relational Systems: Foundations and Principles},
  author = {R. B\v{e}lohl{\'a}vek},
  publisher = {Kluwer},
  year = 2002
}

@BOOK{Hajek1998,
  title  = {Metamathematics of Fuzzy Logics},
  author = {P. H{\'a}jek},
  publisher = {Kluwer Academic Publishers},
  year = 1998
}

@article{DBLP:journals/tfs/NguyenN23,
  author       = {L. A. Nguyen and
                  N. T. Nguyen},
  title        = {Logical Characterizations of Crisp Bisimulations in Fuzzy Description
                  Logics},
  journal      = {{IEEE} Trans. Fuzzy Syst.},
  volume       = {31},
  number       = {4},
  pages        = {1294--1304},
  year         = {2023},
  _url          = {https://doi.org/10.1109/TFUZZ.2022.3198853},
  doi          = {10.1109/TFUZZ.2022.3198853},
  _timestamp    = {Sat, 29 Apr 2023 19:26:30 +0200},
  _biburl       = {https://dblp.org/rec/journals/tfs/NguyenN23.bib},
  _bibsource    = {dblp computer science bibliography, https://dblp.org}
}

@article{fss/QiaoZP23,
  author       = {S. Qiao and
                  P. Zhu and
                  W. Pedrycz},
  title        = {Approximate bisimulations for fuzzy-transition systems},
  journal      = {Fuzzy Sets Syst.},
  volume       = {472},
  pages        = {108533},
  year         = {2023},
  _url          = {https://doi.org/10.1016/j.fss.2023.108533},
  doi          = {10.1016/J.FSS.2023.108533},
  _timestamp    = {Thu, 08 Feb 2024 10:33:06 +0100},
  _biburl       = {https://dblp.org/rec/journals/fss/QiaoZP23.bib},
  _bibsource    = {dblp computer science bibliography, https://dblp.org}
}

@article{DBLP:journals/fss/Nguyen24,
  author       = {L. A. Nguyen},
  title        = {Minimizing fuzzy interpretations in fuzzy description logics by using
                  crisp bisimulations},
  journal      = {Fuzzy Sets Syst.},
  volume       = {481},
  pages        = {108896},
  year         = {2024},
  _url          = {https://doi.org/10.1016/j.fss.2024.108896},
  doi          = {10.1016/J.FSS.2024.108896},
  _timestamp    = {Sun, 19 Jan 2025 14:22:01 +0100},
  _biburl       = {https://dblp.org/rec/journals/fss/Nguyen24.bib},
  _bibsource    = {dblp computer science bibliography, https://dblp.org}
}

@article{Jiang1993,
	author = {Jiang, T. and Ravikumar, B.},
	title = {Minimal {NFA} Problems are Hard},
	journal = {{SIAM} J. Comput.},
	volume = {22},
	number = {6},
	pages = {1117-1141},
	year = {1993},
	doi = {10.1137/0222067},
}

@article{FDSML,
  author = {L. A. Nguyen},
  title        = {Fuzzy directed simulations for fuzzy modal logics over residuated
                  lattices},
  journal      = {CoRR},
  volume       = {abs/2607.16346},
  year         = {2026},
  _url          = {https://doi.org/10.48550/arXiv.2607.16346},
  _doi          = {10.48550/ARXIV.2607.16346},
  eprinttype   = {arXiv},
  eprint       = {2607.16346},
  _timestamp    = {Sat, 08 Aug 2026 20:52:01 +0200},
  _biburl       = {https://dblp.org/rec/journals/corr/abs-2607-16346.bib},
  _bibsource    = {dblp computer science bibliography, https://dblp.org}
}

@article{PEEVA20084152,
	title = {Computing behavior of finite fuzzy machines – Algorithm and its application to reduction and minimization},
	journal = {Inf. Sci.},
	volume = {178},
	number = {21},
	pages = {4152-4165},
	year = {2008},
	doi = {10.1016/j.ins.2008.07.009},
	author = {K. Peeva and Z. Zahariev},
}

@article{CSIP.10,
    title = {Fuzzy relation equations and reduction of fuzzy automata},
    journal = {J. Comput. Syst. Sci.},
    volume = {76},
    number = {7},
    pages = {609-633},
    year = {2010},
    doi = {10.1016/j.jcss.2009.10.015},
    author = {\'Ciri\'c, M. and Stamenkovi\'c, A. and Ignjatovi\'c, J. and Petkovi\'c, T.},
}

@article{YangL.19,
	title = {Approximate bisimulations and state reduction of fuzzy automata under fuzzy similarity measures},
	journal = {Fuzzy Sets Syst.},
	volume = "391",
	pages = "72--95",
	year = "2020",
	author = {C. Yang and Y. Li},
	doi = {10.1016/j.fss.2019.07.010}
}

@article{FSS-D-25-00029,
  author = {S. Stanimirovi{\'c} and L. A. Nguyen and M. \'Ciri\'c and M. Stankovi\'c},
  title = {Approximate state reduction of fuzzy finite automata},
  journal      = {Fuzzy Sets Syst.},
  volume       = {519},
  pages        = {109535},
  year         = {2025},
  _url          = {https://doi.org/10.1016/j.fss.2025.109535},
  doi          = {10.1016/J.FSS.2025.109535},
  _timestamp    = {Wed, 09 Sep 2026 18:14:32 +0200},
  _biburl       = {https://dblp.org/rec/journals/fss/StanimirovicN0025a.bib},
  _bibsource    = {dblp computer science bibliography, https://dblp.org}
}

@article{nguyen2026approximateminimizationinterpretationsfuzzy,
  title={Approximate minimization of interpretations in fuzzy description logics under the {G\"odel} semantics}, 
  author={L. A. Nguyen},
  journal = {Inf. Sci.},
  volume = {761},
  pages = {124171},
  year = {2027},
  issn = {0020-0255},
  doi = {10.1016/j.ins.2026.124171},
  _url = {https://www.sciencedirect.com/science/article/pii/S0020025526011023},
}

@misc{MinFKM-prog,
  author = {L. A. Nguyen},
  title = {An implementation of the algorithms given in the current paper},
  howpublished = {\url{http://mimuw.edu.pl/~nguyen/MinFKM}},
  year = 2026
}

@article{CompMinimaxG,
  author = {L. A. Nguyen},
  title = {Computing markings for fuzzy minimax nets over the {G\"odel} structure},
  journal      = {CoRR},
  volume       = {abs/2607.15494},
  year         = {2026},
  _url          = {https://doi.org/10.48550/arXiv.2607.15494},
  _doi          = {10.48550/ARXIV.2607.15494},
  eprinttype   = {arXiv},
  eprint       = {2607.15494},
  _timestamp    = {Fri, 07 Aug 2026 08:53:35 +0200},
  _biburl       = {https://dblp.org/rec/journals/corr/abs-2607-15494.bib},
  _bibsource    = {dblp computer science bibliography, https://dblp.org}
}
\bibliographystyle{elsarticle-harv}

%===============================================================================
\appendix

\section{Additional examples}
\renewcommand{\thesection}{A}

Example~\ref{example: HGDHJ} shows that the quotient fuzzy Kripke model $\mME$ of a fuzzy Kripke model $\mM$ with respect to its greatest fuzzy $\Phi$-auto-bisimulation~$E$ may have different sizes depending on whether the underlying residuated lattice is the product or {\L}ukasiewicz structure. 
Likewise, Example~\ref{example: HFJSK2} shows that Algorithm~\ref{alg1} may return minimal reductions of a finite fuzzy Kripke model (that preserve local assertions in \fKP) of different sizes depending on whether the underlying residuated lattice is the product or {\L}ukasiewicz structure. 
Examples~\ref{example: HFJSK}, \ref{example: HGDHJ}, \ref{example: HFJSK2}, and \ref{example: HGDHJ2}, however, all deal with the case where $\Phi \subseteq \{\cup, \circ, *, \lor, \Box, U\}$.
In this appendix, we present examples involving cases where $\Phi$ also includes $C$ and/or $O$, i.e., cases involving converse programs and/or nominals.

%----------------------------------------------------------------------

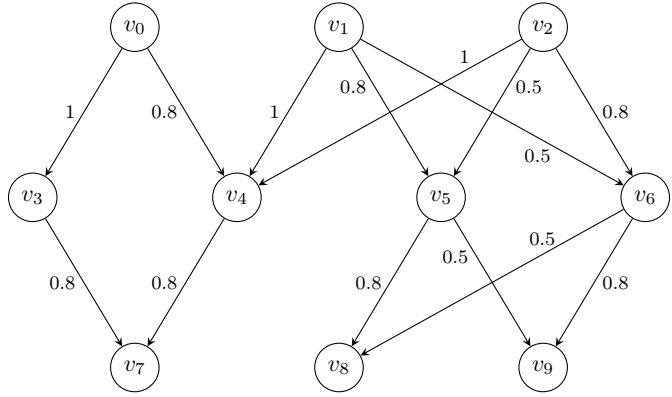
\begin{figure}
\centering
\begin{tikzpicture}[scale=0.9, transform shape,->,>=stealth,auto,black]
\node (top) {};
\node (btop) [node distance=0.2cm, below of=top] {};
\node[draw, circle] (v0) [node distance=1.5cm, right of=btop] {$v_0$};
\node[draw, circle] (v1) [node distance=3.0cm, right of=v0] {$v_1$};
\node[draw, circle] (v2) [node distance=3.0cm, right of=v1] {$v_2$};
\node[draw, circle] (v3) [node distance=2.5cm, below of=btop] {$v_3$};
\node[draw, circle] (v4) [node distance=3.0cm, right of=v3] {$v_4$};
\node[draw, circle] (v5) [node distance=3.0cm, right of=v4] {$v_5$};
\node[draw, circle] (v6) [node distance=3.0cm, right of=v5] {$v_6$};
\node[draw, circle] (v7) [node distance=5.0cm, below of=v0] {$v_7$};
\node[draw, circle] (v8) [node distance=3.0cm, right of=v7] {$v_8$};
\node[draw, circle] (v9) [node distance=3.0cm, right of=v8] {$v_9$};
\draw (v0) to node [left]{\footnotesize{1}} (v3);
\draw (v0) to node [left]{\footnotesize{0.8}} (v4);
\draw (v1) to node [left]{\footnotesize{1}} (v4);
\draw (v1) to node [left, pos=0.3]{\footnotesize{0.8}} (v5);
\draw (v1) to node [left, pos=0.8, xshift=-5pt]{\footnotesize{0.5}} (v6);
\draw (v2) to node [left, pos=0.12, xshift=-5pt]{\footnotesize{1}} (v4);
\draw (v2) to node [right, pos=0.3]{\footnotesize{0.5}} (v5);
\draw (v2) to node [right]{\footnotesize{0.8}} (v6);
\draw (v3) to node [left]{\footnotesize{0.8}} (v7);
\draw (v4) to node [left]{\footnotesize{0.8}} (v7);
\draw (v5) to node [left]{\footnotesize{0.8}} (v8);
\draw (v5) to node [left, pos=0.3]{\footnotesize{0.5}} (v9);
\draw (v6) to node [left, pos=0.2, xshift=-3pt]{\footnotesize{0.5}} (v8);
\draw (v6) to node [right]{\footnotesize{0.8}} (v9);
\end{tikzpicture}
\caption{An illustration of the fuzzy Kripke model $\mM$ considered in Example~\ref{example: JHRUA}.\label{fig: JHRUA1}}
\end{figure}

%----------------------------------------------------------------------

\begin{example}\label{example: JHRUA}
Let $\SA = \{\sigma\}$, $\SP = \emptyset$, $\Ss = \{v_0, v_1, v_2\}$, $\So = \{v_0, v_1\}$, and $\{\cup, \circ, *, \lor, \Box\} \subseteq \Phi \subseteq \{\cup, \circ, *, \lor, \Box, C,O\}$. Let $\mM$ be the fuzzy Kripke model illustrated in Figure~\ref{fig: JHRUA1} and specified as follows:
\begin{itemize}
\item $\DeltaM = \{v_0,\ldots,v_9\}$, 
\item $v_{i}^\mM = v_i$ for $v_i \in \Ss$, 
\item $\sigma^\mM = \{\tuple{v_0,v_3}\!:\!1$, 
$\tuple{v_0,v_4}\!:\!0.8$, 
$\tuple{v_1,v_4}\!:\!1$, 
$\tuple{v_1,v_5}\!:\!0.8$, 
$\tuple{v_1,v_6}\!:\!0.5$, 
$\tuple{v_2,v_4}\!:\!1$, 
$\tuple{v_2,v_5}\!:\!0.5$, 
$\tuple{v_2,v_6}\!:\!0.8$, 
$\tuple{v_3,v_7}\!:\!0.8$, 
$\tuple{v_4,v_7}\!:\!0.8$, 
$\tuple{v_5,v_8}\!:\!0.8$, 
$\tuple{v_5,v_9}\!:\!0.5$, 
$\tuple{v_6,v_8}\!:\!0.5$, 
$\tuple{v_6,v_9}\!:\!0.8\}$.
\end{itemize}

%----------------------------------------------------------------------

\begin{table}
\caption{The greatest fuzzy $\Phi$-auto-bisimulation of the fuzzy Kripke model $\mM$ considered in Example~\ref{example: JHRUA} in different cases. The occurrences of $a$ in the bottom-left subtable denote $0.625$. \label{table: JSIWH}}
\scriptsize

\bigskip

\noindent
\begin{tabular}{ccc}
& Under the product semantics: & Under the {\L}ukasiewicz semantics: \\[1ex]
\!\!\!\!\!\!\!\!
\begin{tabular}{l}
Case \\
$C \notin \Phi$,\\ 
$O \notin \Phi$:
\end{tabular}\!\!\!\!\!
& 
\begin{tabular}{|*{11}{>{\centering\arraybackslash}m{0.22cm}|}}
\hline
$E$   & $v_0$ & $v_1$ & $v_2$ & $v_3$ & $v_4$ & $v_5$ & $v_6$ & $v_7$ & $v_8$ & $v_9$ \\ \hline
$v_0$ &   1   &   1   &   1   &   0   &   0   &   0   &   0   &   0   &   0   &   0   \\ \hline
$v_1$ &   1   &   1   &   1   &   0   &   0   &   0   &   0   &   0   &   0   &   0   \\ \hline
$v_2$ &   1   &   1   &   1   &   0   &   0   &   0   &   0   &   0   &   0   &   0   \\ \hline
$v_3$ &   0   &   0   &   0   &   1   &   1   &   1   &   1   &   0   &   0   &   0   \\ \hline
$v_4$ &   0   &   0   &   0   &   1   &   1   &   1   &   1   &   0   &   0   &   0   \\ \hline
$v_5$ &   0   &   0   &   0   &   1   &   1   &   1   &   1   &   0   &   0   &   0   \\ \hline
$v_6$ &   0   &   0   &   0   &   1   &   1   &   1   &   1   &   0   &   0   &   0   \\ \hline
$v_7$ &   0   &   0   &   0   &   0   &   0   &   0   &   0   &   1   &   1   &   1   \\ \hline
$v_8$ &   0   &   0   &   0   &   0   &   0   &   0   &   0   &   1   &   1   &   1   \\ \hline
$v_9$ &   0   &   0   &   0   &   0   &   0   &   0   &   0   &   1   &   1   &   1   \\ \hline
\end{tabular}
& 
\begin{tabular}{|*{11}{>{\centering\arraybackslash}m{0.22cm}|}}
\hline
$E$   & $v_0$ & $v_1$ & $v_2$ & $v_3$ & $v_4$ & $v_5$ & $v_6$ & $v_7$ & $v_8$ & $v_9$ \\ \hline
$v_0$ &   1   &   1   &   1   &   0   &   0   &   0   &   0   &   0   &   0   &   0   \\ \hline
$v_1$ &   1   &   1   &   1   &   0   &   0   &   0   &   0   &   0   &   0   &   0   \\ \hline
$v_2$ &   1   &   1   &   1   &   0   &   0   &   0   &   0   &   0   &   0   &   0   \\ \hline
$v_3$ &   0   &   0   &   0   &   1   &   1   &   1   &   1   & \!0.2 & \!0.2 & \!0.2 \\ \hline
$v_4$ &   0   &   0   &   0   &   1   &   1   &   1   &   1   & \!0.2 & \!0.2 & \!0.2 \\ \hline
$v_5$ &   0   &   0   &   0   &   1   &   1   &   1   &   1   & \!0.2 & \!0.2 & \!0.2 \\ \hline
$v_6$ &   0   &   0   &   0   &   1   &   1   &   1   &   1   & \!0.2 & \!0.2 & \!0.2 \\ \hline
$v_7$ &   0   &   0   &   0   & \!0.2 & \!0.2 & \!0.2 & \!0.2 &   1   &   1   &   1   \\ \hline
$v_8$ &   0   &   0   &   0   & \!0.2 & \!0.2 & \!0.2 & \!0.2 &   1   &   1   &   1   \\ \hline
$v_9$ &   0   &   0   &   0   & \!0.2 & \!0.2 & \!0.2 & \!0.2 &   1   &   1   &   1   \\ \hline
\end{tabular}
\\
\\
\!\!\!\!\!\!\!\!
\begin{tabular}{l}
Case \\
$C \notin \Phi$,\\ 
$O \in \Phi$:
\end{tabular}\!\!\!\!\!
& 
\begin{tabular}{|*{11}{>{\centering\arraybackslash}m{0.22cm}|}}
\hline
$E$   & $v_0$ & $v_1$ & $v_2$ & $v_3$ & $v_4$ & $v_5$ & $v_6$ & $v_7$ & $v_8$ & $v_9$ \\ \hline
$v_0$ &   1   &   0   &   0   &   0   &   0   &   0   &   0   &   0   &   0   &   0   \\ \hline
$v_1$ &   0   &   1   &   0   &   0   &   0   &   0   &   0   &   0   &   0   &   0   \\ \hline
$v_2$ &   0   &   0   &   1   &   0   &   0   &   0   &   0   &   0   &   0   &   0   \\ \hline
$v_3$ &   0   &   0   &   0   &   1   &   1   &   1   &   1   &   0   &   0   &   0   \\ \hline
$v_4$ &   0   &   0   &   0   &   1   &   1   &   1   &   1   &   0   &   0   &   0   \\ \hline
$v_5$ &   0   &   0   &   0   &   1   &   1   &   1   &   1   &   0   &   0   &   0   \\ \hline
$v_6$ &   0   &   0   &   0   &   1   &   1   &   1   &   1   &   0   &   0   &   0   \\ \hline
$v_7$ &   0   &   0   &   0   &   0   &   0   &   0   &   0   &   1   &   1   &   1   \\ \hline
$v_8$ &   0   &   0   &   0   &   0   &   0   &   0   &   0   &   1   &   1   &   1   \\ \hline
$v_9$ &   0   &   0   &   0   &   0   &   0   &   0   &   0   &   1   &   1   &   1   \\ \hline
\end{tabular}
& 
\begin{tabular}{|*{11}{>{\centering\arraybackslash}m{0.22cm}|}}
\hline
$E$   & $v_0$ & $v_1$ & $v_2$ & $v_3$ & $v_4$ & $v_5$ & $v_6$ & $v_7$ & $v_8$ & $v_9$ \\ \hline
$v_0$ &   1   &   0   &   0   &   0   &   0   &   0   &   0   &   0   &   0   &   0   \\ \hline
$v_1$ &   0   &   1   &   0   &   0   &   0   &   0   &   0   &   0   &   0   &   0   \\ \hline
$v_2$ &   0   &   0   &   1   &   0   &   0   &   0   &   0   &   0   &   0   &   0   \\ \hline
$v_3$ &   0   &   0   &   0   &   1   &   1   &   1   &   1   & \!0.2 & \!0.2 & \!0.2 \\ \hline
$v_4$ &   0   &   0   &   0   &   1   &   1   &   1   &   1   & \!0.2 & \!0.2 & \!0.2 \\ \hline
$v_5$ &   0   &   0   &   0   &   1   &   1   &   1   &   1   & \!0.2 & \!0.2 & \!0.2 \\ \hline
$v_6$ &   0   &   0   &   0   &   1   &   1   &   1   &   1   & \!0.2 & \!0.2 & \!0.2 \\ \hline
$v_7$ &   0   &   0   &   0   & \!0.2 & \!0.2 & \!0.2 & \!0.2 &   1   &   1   &   1   \\ \hline
$v_8$ &   0   &   0   &   0   & \!0.2 & \!0.2 & \!0.2 & \!0.2 &   1   &   1   &   1   \\ \hline
$v_9$ &   0   &   0   &   0   & \!0.2 & \!0.2 & \!0.2 & \!0.2 &   1   &   1   &   1   \\ \hline
\end{tabular}
\\
\\
\!\!\!\!\!\!\!\!
\begin{tabular}{l}
Case \\
$C \in \Phi$,\\ 
$O \notin \Phi$:
\end{tabular}\!\!\!\!\!
& 
\begin{tabular}{|*{11}{>{\centering\arraybackslash}m{0.22cm}|}}
\hline
$E$   & $v_0$ & $v_1$ & $v_2$ & $v_3$ & $v_4$ & $v_5$ & $v_6$ & $v_7$ & $v_8$ & $v_9$ \\ \hline
$v_0$ &   1   &   1   &   1   &   0   &   0   &   0   &   0   &   0   &   0   &   0   \\ \hline
$v_1$ &   1   &   1   &   1   &   0   &   0   &   0   &   0   &   0   &   0   &   0   \\ \hline
$v_2$ &   1   &   1   &   1   &   0   &   0   &   0   &   0   &   0   &   0   &   0   \\ \hline
$v_3$ &   0   &   0   &   0   &   1   &   1   & \!0.8 & \!0.8 &   0   &   0   &   0   \\ \hline
$v_4$ &   0   &   0   &   0   &   1   &   1   & \!0.8 & \!0.8 &   0   &   0   &   0   \\ \hline
$v_5$ &   0   &   0   &   0   & \!0.8 & \!0.8 &   1   &   1   &   0   &   0   &   0   \\ \hline
$v_6$ &   0   &   0   &   0   & \!0.8 & \!0.8 &   1   &   1   &   0   &   0   &   0   \\ \hline
$v_7$ &   0   &   0   &   0   &   0   &   0   &   0   &   0   &   1   & \!0.8 & \!0.8 \\ \hline
$v_8$ &   0   &   0   &   0   &   0   &   0   &   0   &   0   & \!0.8 &   1   &   1   \\ \hline
$v_9$ &   0   &   0   &   0   &   0   &   0   &   0   &   0   & \!0.8 &   1   &   1   \\ \hline
\end{tabular}
& 
\begin{tabular}{|*{11}{>{\centering\arraybackslash}m{0.22cm}|}}
\hline
$E$   & $v_0$ & $v_1$ & $v_2$ & $v_3$ & $v_4$ & $v_5$ & $v_6$ & $v_7$ & $v_8$ & $v_9$ \\ \hline
$v_0$ &   1   &   1   &   1   &   0   &   0   &   0   &   0   &   0   &   0   &   0   \\ \hline
$v_1$ &   1   &   1   &   1   &   0   &   0   &   0   &   0   &   0   &   0   &   0   \\ \hline
$v_2$ &   1   &   1   &   1   &   0   &   0   &   0   &   0   &   0   &   0   &   0   \\ \hline
$v_3$ &   0   &   0   &   0   &   1   &   1   & \!0.8 & \!0.8 &   0   &   0   &   0   \\ \hline
$v_4$ &   0   &   0   &   0   &   1   &   1   & \!0.8 & \!0.8 &   0   &   0   &   0   \\ \hline
$v_5$ &   0   &   0   &   0   & \!0.8 & \!0.8 &   1   &   1   & \!0.2 & \!0.2 & \!0.2 \\ \hline
$v_6$ &   0   &   0   &   0   & \!0.8 & \!0.8 &   1   &   1   & \!0.2 & \!0.2 & \!0.2 \\ \hline
$v_7$ &   0   &   0   &   0   &   0   &   0   & \!0.2 & \!0.2 &   1   & \!0.8 & \!0.8 \\ \hline
$v_8$ &   0   &   0   &   0   &   0   &   0   & \!0.2 & \!0.2 & \!0.8 &   1   &   1   \\ \hline
$v_9$ &   0   &   0   &   0   &   0   &   0   & \!0.2 & \!0.2 & \!0.8 &   1   &   1   \\ \hline
\end{tabular}
\\
\\
\!\!\!\!\!\!\!\!
\begin{tabular}{l}
Case \\
$C \in \Phi$,\\ 
$O \in \Phi$:
\end{tabular}\!\!\!\!\!
& 
\begin{tabular}{|*{11}{>{\centering\arraybackslash}m{0.22cm}|}}
\hline
$E$   & $v_0$ & $v_1$ & $v_2$ & $v_3$ & $v_4$ & $v_5$ & $v_6$ & $v_7$ & $v_8$ & $v_9$ \\ \hline
$v_0$ &   1   &   0   &   0   &   0   &   0   &   0   &   0   &   0   &   0   &   0   \\ \hline
$v_1$ &   0   &   1   &   0   &   0   &   0   &   0   &   0   &   0   &   0   &   0   \\ \hline
$v_2$ &   0   &   0   &   1   &   0   &   0   &   0   &   0   &   0   &   0   &   0   \\ \hline
$v_3$ &   0   &   0   &   0   &   1   &   0   &   0   &   0   &   0   &   0   &   0   \\ \hline
$v_4$ &   0   &   0   &   0   &   0   &   1   &   0   &   0   &   0   &   0   &   0   \\ \hline
$v_5$ &   0   &   0   &   0   &   0   &   0   &   1   &   a   &   0   &   0   &   0   \\ \hline
$v_6$ &   0   &   0   &   0   &   0   &   0   &   a   &   1   &   0   &   0   &   0   \\ \hline
$v_7$ &   0   &   0   &   0   &   0   &   0   &   0   &   0   &   1   &   0   &   0   \\ \hline
$v_8$ &   0   &   0   &   0   &   0   &   0   &   0   &   0   &   0   &   1   &   a   \\ \hline
$v_9$ &   0   &   0   &   0   &   0   &   0   &   0   &   0   &   0   &   a   &   1   \\ \hline
\end{tabular}
& 
\begin{tabular}{|*{11}{>{\centering\arraybackslash}m{0.22cm}|}}
\hline
$E$   & $v_0$ & $v_1$ & $v_2$ & $v_3$ & $v_4$ & $v_5$ & $v_6$ & $v_7$ & $v_8$ & $v_9$ \\ \hline
$v_0$ &   1   &   0   &   0   &   0   &   0   &   0   &   0   &   0   &   0   &   0   \\ \hline
$v_1$ &   0   &   1   &   0   &   0   &   0   &   0   &   0   &   0   &   0   &   0   \\ \hline
$v_2$ &   0   &   0   &   1   &   0   &   0   &   0   &   0   &   0   &   0   &   0   \\ \hline
$v_3$ &   0   &   0   &   0   &   1   &   0   &   0   &   0   &   0   &   0   &   0   \\ \hline
$v_4$ &   0   &   0   &   0   &   0   &   1   & \!0.2 & \!0.2 &   0   &   0   &   0   \\ \hline
$v_5$ &   0   &   0   &   0   &   0   & \!0.2 &   1   & \!0.7 & \!0.2 & \!0.2 & \!0.2 \\ \hline
$v_6$ &   0   &   0   &   0   &   0   & \!0.2 & \!0.7 &   1   & \!0.2 & \!0.2 & \!0.2 \\ \hline
$v_7$ &   0   &   0   &   0   &   0   &   0   & \!0.2 & \!0.2 &   1   & \!0.2 & \!0.2 \\ \hline
$v_8$ &   0   &   0   &   0   &   0   &   0   & \!0.2 & \!0.2 & \!0.2 &   1   & \!0.7 \\ \hline
$v_9$ &   0   &   0   &   0   &   0   &   0   & \!0.2 & \!0.2 & \!0.2 & \!0.7 &   1   \\ \hline
\end{tabular}
\end{tabular}
\end{table}

%----------------------------------------------------------------------

The greatest fuzzy $\Phi$-auto-bisimulation $E$ of~$\mM$ is displayed in Table~\ref{table: JSIWH} for different cases of~$\Phi$ under the product and {\L}ukasiewicz semantics. The computation of $E$ and the other results reported in this example can be checked using the {\em MinFKMPL.py} module of the package~\cite{MinFKM-prog}, by running this module and examining the outputs involving the input files {\em m1.in} and {\em m1o.in}.   

%----------------------------------------------------------------------

\newcommand{\EMPq}{
\begin{tikzpicture}[scale=0.9, transform shape,->,>=stealth,auto,black]
\node (top) {};
\node[draw] (v0) [node distance=0.4cm, below of=top] {$E_{v_0}$};
\node[draw] (v3) [node distance=2.5cm, below of=v0] {$E_{v_3}$};
\node[draw] (v7) [node distance=2.5cm, below of=v3] {$E_{v_7}$};
\draw (v0) to node [left]{\footnotesize{1}} (v3);
\draw (v3) to node [left, xshift=2pt]{\footnotesize{0.8}} (v7);
\end{tikzpicture}
}

\newcommand{\EMLq}{
\begin{tikzpicture}[scale=0.9, transform shape,->,>=stealth,auto,black]
\node (top) {};
\node[draw] (v0) [node distance=0.4cm, below of=top] {$E_{v_0}$};
\node[draw] (v3) [node distance=2.5cm, below of=v0] {$E_{v_3}$};
\node[draw] (v7) [node distance=2.5cm, below of=v3] {$E_{v_7}$};
\draw (v0) to node [left]{\footnotesize{1}} (v3);
\draw (v3) to node [left, xshift=2pt]{\footnotesize{0.8}} (v7);
\draw (v0) edge[bend left=20] node[right, pos=0.21, xshift=-2pt]{\footnotesize{0.2}} (v7);
\end{tikzpicture}
}

\newcommand{\EMPqO}{
\begin{tikzpicture}[scale=0.9, transform shape,->,>=stealth,auto,black]
\node (top) {};
\node[draw] (v1) [node distance=0.4cm, below of=top] {$E_{v_1}$};
\node[draw] (v0) [node distance=1.0cm, left of=v1] {$E_{v_0}$};
\node[draw] (v2) [node distance=1.0cm, right of=v1] {$E_{v_2}$};
\node[draw] (v3) [node distance=2.5cm, below of=v1] {$E_{v_3}$};
\node[draw] (v7) [node distance=2.5cm, below of=v3] {$E_{v_7}$};
\draw (v0) to node [left, xshift=1pt]{\footnotesize{1}} (v3);
\draw (v1) to node [left, xshift=2pt]{\footnotesize{1}} (v3);
\draw (v2) to node [left, xshift=2pt]{\footnotesize{1}} (v3);
\draw (v3) to node [left, xshift=2pt]{\footnotesize{0.8}} (v7);
\end{tikzpicture}
}

\newcommand{\EMLqO}{
\begin{tikzpicture}[scale=0.9, transform shape,->,>=stealth,auto,black]
\node (top) {};
\node[draw] (v1) [node distance=0.4cm, below of=top] {$E_{v_1}$};
\node[draw] (v0) [node distance=1.0cm, left of=v1] {$E_{v_0}$};
\node[draw] (v2) [node distance=1.0cm, right of=v1] {$E_{v_2}$};
\node[draw] (v3) [node distance=2.5cm, below of=v1] {$E_{v_3}$};
\node[draw] (v7) [node distance=2.5cm, below of=v3] {$E_{v_7}$};
\draw (v0) to node [left, xshift=1pt]{\footnotesize{1}} (v3);
\draw (v1) to node [left, xshift=2pt]{\footnotesize{1}} (v3);
\draw (v2) to node [right, pos=0.4, xshift=-2pt]{\footnotesize{1}} (v3);
\draw (v3) to node [left, pos=0.2, xshift=2pt]{\footnotesize{0.8}} (v7);
\draw (v0) edge[bend right=20] node[left, pos=0.80, xshift=1pt]{\footnotesize{0.2}} (v7);
\draw (v1) edge[bend left=20] node[right, pos=0.3, xshift=-2pt]{\footnotesize{0.2}} (v7);
\draw (v2) edge[bend left=20] node[right, pos=0.80, xshift=-1pt]{\footnotesize{0.2}} (v7);
\end{tikzpicture}
}

\newcommand{\EMPqC}{
\begin{tikzpicture}[scale=0.9, transform shape,->,>=stealth,auto,black]
\node (top) {};
\node[draw] (v0) [node distance=0.4cm, below of=top] {$E_{v_0}$};
\node (bv0) [node distance=2.5cm, below of=v0] {};
\node[draw] (v3) [node distance=1.0cm, left of=bv0] {$E_{v_3}$};
\node[draw] (v5) [node distance=1.0cm, right of=bv0] {$E_{v_5}$};
\node[draw] (v7) [node distance=2.5cm, below of=v3] {$E_{v_7}$};
\node[draw] (v8) [node distance=2.5cm, below of=v5] {$E_{v_8}$};
\draw (v0) to node [left]{\footnotesize{1}} (v3);
\draw (v0) to node [right]{\footnotesize{0.8}} (v5);
\draw (v3) to node [left, xshift=2pt]{\footnotesize{0.8}} (v7);
\draw (v3) to node [left, pos=0.9]{\footnotesize{0.64}} (v8);
\draw (v5) to node [right, xshift=-2pt]{\footnotesize{0.8}} (v8);
\draw (v5) to node [left, pos=0.1]{\footnotesize{0.64}} (v7);
\end{tikzpicture}
}

\newcommand{\EMLqC}{
\begin{tikzpicture}[scale=0.9, transform shape,->,>=stealth,auto,black]
\node (top) {};
\node[draw] (v0) [node distance=0.4cm, below of=top] {$E_{v_0}$};
\node (bv0) [node distance=2.5cm, below of=v0] {};
\node[draw] (v3) [node distance=1.0cm, left of=bv0] {$E_{v_3}$};
\node[draw] (v5) [node distance=1.0cm, right of=bv0] {$E_{v_5}$};
\node[draw] (v7) [node distance=2.5cm, below of=v3] {$E_{v_7}$};
\node[draw] (v8) [node distance=2.5cm, below of=v5] {$E_{v_8}$};
\draw (v0) to node [left]{\footnotesize{1}} (v3);
\draw (v0) to node [right]{\footnotesize{0.8}} (v5);
\draw (v3) to node [left, xshift=2pt]{\footnotesize{0.8}} (v7);
\draw (v3) to node [left, pos=0.9]{\footnotesize{0.6}} (v8);
\draw (v5) to node [right, xshift=-2pt]{\footnotesize{0.8}} (v8);
\draw (v5) to node [left, pos=0.1]{\footnotesize{0.6}} (v7);
\end{tikzpicture}
}

\begin{figure}
\begin{center}
\begin{tabular}{|c|c|c||c|c|c|}
\hline
\multicolumn{3}{|c||}{Under the product semantics} & \multicolumn{3}{c|}{Under the {\L}ukasiewicz semantics} 
\\ \hline
$C \notin \Phi$ & $C \notin \Phi$ & $C \in \Phi$ & 
$C \notin \Phi$ & $C \notin \Phi$ & $C \in \Phi$
\\
$O \notin \Phi$ & $O \in \Phi$ & $O \notin \Phi$ & 
$O \notin \Phi$ & $O \in \Phi$ & $O \notin \Phi$
\\ \hline
  \EMPq
& \EMPqO
& \EMPqC
& \EMLq
& \EMLqO
& \EMLqC 
\\ \hline
\end{tabular}
\end{center}
\caption{The quotient model $\mME$ of the fuzzy Kripke model $\mM$ considered in Example~\ref{example: JHRUA} with respect to its greatest fuzzy $\Phi$-auto-bisimulation $E$, for different cases of~$\Phi$ under the product and {\L}ukasiewicz semantics.\label{fig: JHRUA2}}
\end{figure}

%----------------------------------------------------------------------

The quotient fuzzy Kripke model $\mME$ is illustrated in Figure~\ref{fig: JHRUA2} for different cases of~$\Phi$ under the product and {\L}ukasiewicz semantics. This figure does not cover the case where $\{C,O\} \subset \Phi$. In this remaining case, $\mME$ is isomorphic to $\mM$ under both the product and {\L}ukasiewicz semantics, via the bijection that maps $E_{v_i}$ to $v_i$, for $0 \leq i \leq 9$.

%----------------------------------------------------------------------

\newcommand{\EMPl}{
\begin{tikzpicture}[scale=0.9, transform shape,->,>=stealth,auto,black]
\node (top) {};
\node[draw, circle] (v0) [node distance=0.4cm, below of=top] {$v_0$};
\node[draw, circle] (v3) [node distance=2.5cm, below of=v0] {$v_3$};
\node[draw, circle] (v7) [node distance=2.5cm, below of=v3] {$v_7$};
\draw (v0) to node [left]{\footnotesize{1}} (v3);
\draw (v3) to node [left, xshift=2pt]{\footnotesize{0.8}} (v7);
\end{tikzpicture}
}

\newcommand{\EMPlO}{
\begin{tikzpicture}[scale=0.9, transform shape,->,>=stealth,auto,black]
\node (top) {};
\node[draw, circle] (v1) [node distance=0.4cm, below of=top] {$v_1$};
\node[draw, circle] (v0) [node distance=1.2cm, left of=v1] {$v_0$};
\node[draw, circle] (v2) [node distance=1.2cm, right of=v1] {$v_2$};
\node[draw, circle] (v3) [node distance=2.5cm, below of=v1] {$v_3$};
\node[draw, circle] (v7) [node distance=2.5cm, below of=v3] {$v_7$};
\draw (v0) to node [left, xshift=1pt]{\footnotesize{1}} (v3);
\draw (v1) to node [left, xshift=2pt]{\footnotesize{1}} (v3);
\draw (v2) to node [left, xshift=2pt]{\footnotesize{1}} (v3);
\draw (v3) to node [left, xshift=2pt]{\footnotesize{0.8}} (v7);
\end{tikzpicture}
}

\newcommand{\EMPlCO}{
\begin{tikzpicture}[scale=0.9, transform shape,->,>=stealth,auto,black]
\node (top) {};
\node (btop) [node distance=0.4cm, below of=top] {};
\node[draw, circle] (v0) [node distance=1.2cm, right of=btop] {$v_0$};
\node[draw, circle] (v1) [node distance=2.4cm, right of=v0] {$v_1$};
\node[draw, circle] (v2) [node distance=2.4cm, right of=v1] {$v_2$};
\node[draw, circle] (v3) [node distance=2.5cm, below of=btop] {$v_3$};
\node[draw, circle] (v4) [node distance=2.4cm, right of=v3] {$v_4$};
\node[draw, circle] (v5) [node distance=2.4cm, right of=v4] {$v_5$};
\node[draw, circle] (v6) [node distance=2.4cm, right of=v5] {$v_6$};
\node[draw, circle] (v7) [node distance=5.0cm, below of=v0] {$v_7$};
\node[draw, circle] (v8) [node distance=2.4cm, right of=v7] {$v_8$};
\node[draw, circle] (v9) [node distance=2.4cm, right of=v8] {$v_9$};
\draw (v0) to node [left]{\footnotesize{1}} (v3);
\draw (v0) to node [left]{\footnotesize{0.8}} (v4);
\draw (v1) to node [left]{\footnotesize{1}} (v4);
\draw (v1) to node [left, pos=0.3]{\footnotesize{0.8}} (v5);
\draw (v1) to node [left, pos=0.8, xshift=-5pt]{\footnotesize{0.5}} (v6);
\draw (v2) to node [left, pos=0.12, xshift=-5pt]{\footnotesize{1}} (v4);
\draw (v2) to node [right, pos=0.3]{\footnotesize{0.5}} (v5);
\draw (v2) to node [right]{\footnotesize{0.8}} (v6);
\draw (v3) to node [left]{\footnotesize{0.8}} (v7);
\draw (v4) to node [left]{\footnotesize{0.8}} (v7);
\draw (v5) to node [left]{\footnotesize{0.8}} (v8);
%\draw (v5) to node [left, pos=0.3]{\footnotesize{0.5}} (v9);
%\draw (v6) to node [left, pos=0.2, xshift=-3pt]{\footnotesize{0.5}} (v8);
\draw (v6) to node [right]{\footnotesize{0.8}} (v9);
\end{tikzpicture}
}

\begin{figure}
\centering
\begin{tabular}{|c|c|c|}
\hline
\ \ $O \notin \Phi$\ \ & 
$O \in \Phi$, $C \notin \Phi$ & 
$O \in \Phi$, $C \in \Phi$
\\ \hline 
  \EMPl
& \EMPlO
& \EMPlCO
\\ \hline 
\end{tabular}
\caption{The reduced fuzzy Kripke model obtained by applying Algorithm~\ref{alg1} to the fuzzy Kripke model $\mM$ considered in Example~\ref{example: JHRUA}, for different cases of~$\Phi$ under the product and {\L}ukasiewicz semantics.\label{fig: JHRUA3}}
\end{figure}

%----------------------------------------------------------------------

The reduced fuzzy Kripke model $\mMp$ obtained by applying Algorithm~\ref{alg1} to~$\mM$ is illustrated in Figure~\ref{fig: JHRUA3} for different cases of~$\Phi$ under the product and {\L}ukasiewicz semantics. The reduced model in each case is identical under the product and {\L}ukasiewicz semantics. The case where $O \notin \Phi$ is independent of whether $C \in \Phi$. 
Observe that:
\begin{itemize}
\item $\mMp$ in the case where $O \notin \Phi$ is isomorphic to $\mME$ in the case where $C \notin \Phi$ and $O \notin \Phi$.
\item When $O \in \Phi$ and $C \notin \Phi$, $\mMp$ is isomorphic to $\mME$.
\item When $O \in \Phi$ and $C \in \Phi$, $\mMp$ differs from $\mM$ only in that it does not contain the transitions $\tuple{v_5,\sigma,v_9}$ and $\tuple{v_6,\sigma,v_8}$ (i.e., $\sigma^\mMp(v_5,v_9) = \sigma^\mMp(v_6,v_8) = 0$).
\end{itemize}

Consider the case where $O \notin \Phi$, $C \in \Phi$, and the underlying residuated lattice is the product or {\L}ukasiewicz structure. We examine the execution of Algorithm~\ref{alg1} for $\mM$:
\begin{itemize}
\item After executing the statements~\ref{step: alg1 1}--\ref{step: alg1 2}, we have $\DeltaMp = \{v_0\}$, $v_i^\mMp = v_0$ for $v_i \in \Ss$, and $\unprocessed = \{v_0\}$.

\item In the first iteration of the ``while'' loop, $v_0$ is extracted from $\unprocessed$ and processed. Since $\sigma^\mM(v_0, v_3) > \sigma^\mM(v_0, v_4) > 0$, the states $v_3$ and $v_4$ are processed in subsequent iterations of the inner ``foreach'' loop~\ref{step: alg1 5}.
When $v_3$ is processed, it is added to $\DeltaMp$ and $\unprocessed$ by the statement~\ref{step: alg1 10}, with $\sigma^\mMp(v_0, v_3)$ set to 1 by the statement~\ref{step: alg1 12}. When $v_4$ is processed, it does not satisfy the condition of the ``if'' statement~\ref{step: alg1 6}, since
\[ \sigma^\mMp(v_0, v_3) \fand E(v_3, v_4) = 1 > 0.8 = \sigma^\mM(v_0, v_4). \]

\item In the second iteration of the ``while'' loop, $v_3$ is extracted from $\unprocessed$ and processed. For $\alpha = \sigma$, only the state $v_7$ is processed by the inner ``foreach'' loop~\ref{step: alg1 5}. As a result, $v_7$ is added to $\DeltaMp$ and $\unprocessed$ by the statement~\ref{step: alg1 10}, with $\sigma^\mMp(v_3, v_7)$ set to~$0.8$ by the statement~\ref{step: alg1 12}. For $\alpha = \sigma^-$, only the state $v_0$ is processed by the inner ``foreach'' loop~\ref{step: alg1 5}. However, $v_0$ does not satisfy the condition of the ``if'' statement~\ref{step: alg1 6}, since
\[ (\sigma^-)^\mMp(v_3, v_0) \fand E(v_0, v_0) = 1 = (\sigma^-)^\mM(v_3, v_0). \]

\item In the third iteration of the ``while'' loop, $v_7$ is extracted from $\unprocessed$ and processed. The inner ``foreach'' loop~\ref{step: alg1 5} executes only two iterations, with $\alpha = \sigma^-$ and $y \in \{v_3, v_4\}$. Neither of these values of $y$ satisfies the condition of the ``if'' statement~\ref{step: alg1 6}, since
\begin{eqnarray*}
&& (\sigma^-)^\mMp(v_7, v_3) \fand E(v_3, v_3) = 0.8 = (\sigma^-)^\mM(v_7, v_3),\\
&& (\sigma^-)^\mMp(v_7, v_3) \fand E(v_3, v_4) = 0.8 = (\sigma^-)^\mM(v_7, v_4).
\end{eqnarray*}
Thus, no changes are made to $\mMp$ or $\unprocessed$. 

\item After the third iteration, the ``while'' loop terminates, and the algorithm returns $\mMp$, which is illustrated in the left subfigure in Figure~\ref{fig: JHRUA3}.
\end{itemize}

Now consider the case where $O \in \Phi$, $C \notin \Phi$, and the underlying residuated lattice is the product or {\L}ukasiewicz structure. We examine the execution of Algorithm~\ref{alg1} for $\mM$:
\begin{itemize}
\item After executing the statements~\ref{step: alg1 1}--\ref{step: alg1 2}, we have $\DeltaMp = \{v_0,v_1,v_2\}$, $v_i^\mMp = v_i$ for $v_i \in \Ss$, and $\unprocessed = \{v_0,v_1,v_2\}$.

\item Consider the first iteration of the ``while'' loop and assume that $v_0$ is extracted from $\unprocessed$ and processed. As in the previous case, the states $v_3$ and $v_4$ are processed in subsequent iterations of the inner ``foreach'' loop~\ref{step: alg1 5}, but only $v_3$ is added to $\DeltaMp$ and $\unprocessed$, with $\sigma^\mMp(v_0, v_3)$ set to~1.

\item Consider the second iteration of the ``while'' loop and assume that $v_1$ is extracted from $\unprocessed$ and processed. Since $\sigma^\mM(v_1, v_4) > \sigma^\mM(v_1, v_5) > \sigma^\mM(v_1, v_6) > 0$, the states $v_4$, $v_5$, and $v_6$ are processed (as $y$) in subsequent iterations of the inner ``foreach'' loop~\ref{step: alg1 5}. When $v_4$ is processed, it satisfies the condition of the ``if'' statement~\ref{step: alg1 6}, as well as the condition of the ``if'' statement~\ref{step: alg1 7} with $y' = v_3$. As a result, $\sigma^\mMp(v_1,v_3)$ is set to $\sigma^\mM(v_1,v_4)$, i.e., to~1. When $v_5$ and $v_6$ are processed, neither satisfies the condition of the ``if'' statement~\ref{step: alg1 6}, since
\begin{eqnarray*}
&& \sigma^\mMp(v_1, v_3) \fand E(v_3, v_5) = 1 > 0.8 = \sigma^\mM(v_1, v_5),\\
&& \sigma^\mMp(v_1, v_3) \fand E(v_3, v_6) = 1 > 0.5 = \sigma^\mM(v_1, v_6).
\end{eqnarray*}
Thus, no further changes are made to $\mMp$ or $\unprocessed$.

\item Consider the third iteration of the ``while'' loop and assume that $v_2$ is extracted from $\unprocessed$ and processed. As in the second iteration, $\sigma^\mMp(v_2,v_3)$ is set to~1, and no further changes are made to $\mMp$ or $\unprocessed$.

\item In the fourth iteration of the ``while'' loop, $v_3$ is extracted from $\unprocessed$ and processed. Only the state $v_7$ is processed by the inner ``foreach'' loop~\ref{step: alg1 5}. As a result, $v_7$ is added to $\DeltaMp$ and $\unprocessed$ by the statement~\ref{step: alg1 10}, with $\sigma^\mMp(v_3, v_7)$ set to~$0.8$ by the statement~\ref{step: alg1 12}.

\item In the fifth iteration of the ``while'' loop, $v_7$ is extracted from $\unprocessed$ and processed, but no iterations of the inner ``foreach'' loop~\ref{step: alg1 5} are executed.

\item After the fifth iteration, the ``while'' loop terminates, and the algorithm returns $\mMp$, which is illustrated in the middle subfigure in Figure~\ref{fig: JHRUA3}.
\myend
\end{itemize}
\end{example}

%===============================================================================

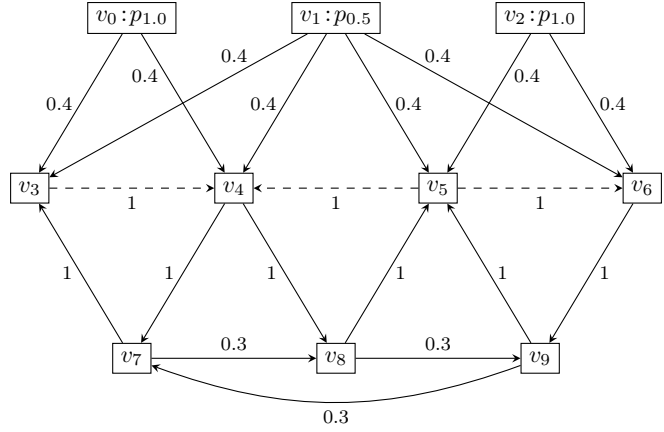
\begin{figure}
\centering
\begin{tikzpicture}[scale=0.9, transform shape,->,>=stealth,auto,black]
\node (top) {};
\node (btop) [node distance=0.2cm, below of=top] {};
\node[draw] (v0) [node distance=1.5cm, right of=btop] {$v_0\!:\!p_{1.0}$};
\node[draw] (v1) [node distance=3.0cm, right of=v0] {$v_1\!:\!p_{0.5}$};
\node[draw] (v2) [node distance=3.0cm, right of=v1] {$v_2\!:\!p_{1.0}$};
\node[draw] (v3) [node distance=2.5cm, below of=btop] {$v_3$};
\node[draw] (v4) [node distance=3.0cm, right of=v3] {$v_4$};
\node[draw] (v5) [node distance=3.0cm, right of=v4] {$v_5$};
\node[draw] (v6) [node distance=3.0cm, right of=v5] {$v_6$};
\node[draw] (v7) [node distance=5.0cm, below of=v0] {$v_7$};
\node[draw] (v8) [node distance=3.0cm, right of=v7] {$v_8$};
\node[draw] (v9) [node distance=3.0cm, right of=v8] {$v_9$};
\draw (v0) to node [left]{\footnotesize{0.4}} (v3);
\draw (v0) to node [left,pos=0.3]{\footnotesize{0.4}} (v4);
\draw (v1) to node [left,pos=0.2,yshift=3pt]{\footnotesize{0.4}} (v3);
\draw (v1) to node [left]{\footnotesize{0.4}} (v4);
\draw (v1) to node [right]{\footnotesize{0.4}} (v5);
\draw (v1) to node [right,pos=0.2,yshift=3pt]{\footnotesize{0.4}} (v6);
\draw (v2) to node [right,pos=0.3]{\footnotesize{0.4}} (v5);
\draw (v2) to node [right]{\footnotesize{0.4}} (v6);
\draw[dashed] (v3) to node [below]{\footnotesize{1}} (v4);
\draw[dashed] (v5) to node [below]{\footnotesize{1}} (v4);
\draw[dashed] (v5) to node [below]{\footnotesize{1}} (v6);
\draw (v7) to node [left]{\footnotesize{1}} (v3);
\draw (v4) to node [left]{\footnotesize{1}} (v7);
\draw (v4) to node [left]{\footnotesize{1}} (v8);
\draw (v8) to node [right]{\footnotesize{1}} (v5);
\draw (v9) to node [right]{\footnotesize{1}} (v5);
\draw (v6) to node [right]{\footnotesize{1}} (v9);
\draw (v7) to node [above]{\footnotesize{0.3}} (v8);
\draw (v8) to node [above]{\footnotesize{0.3}} (v9);
\draw (v9) edge[bend left=20] node [below]{\footnotesize{0.3}} (v7);
\end{tikzpicture}
\caption{An illustration of the fuzzy Kripke model $\mM$ considered in Example~\ref{example: HGDHG}.\label{fig: HGDHG1}}
\end{figure}

%----------------------------------------------------------------------

\begin{example}\label{example: HGDHG}
Let $\SA = \{\sigma,\varrho\}$, $\SP = \{p\}$, $\Ss = \{v_0\}$, $\So = \{v_0\}$, and $\{\cup, \circ, *, \lor, \Box\} \subseteq \Phi \subseteq \{\cup, \circ, *, \lor, \Box, C,O\}$. Let $\mM$ be the fuzzy Kripke model illustrated in Figure~\ref{fig: JHRUA1} and specified as follows:
\begin{itemize}
\item $\DeltaM = \{v_0,\ldots,v_9\}$, $v_0^\mM = v_0$, $p^\mM = \{v_0\!:\!1, v_1\!:\!0.5, v_2\!:\!1\}$, 
\item $\sigma^\mM = \{\tuple{v_0,v_3}\!:\!0.4$, 
$\tuple{v_0,v_4}\!:\!0.4$, 
$\tuple{v_1,v_3}\!:\!0.4$, 
$\tuple{v_1,v_4}\!:\!0.4$, 
$\tuple{v_1,v_5}\!:\!0.4$, 
$\tuple{v_1,v_6}\!:\!0.4$, 
$\tuple{v_2,v_5}\!:\!0.4$, 
$\tuple{v_2,v_6}\!:\!0.4$, 
$\tuple{v_4,v_7}\!:\!1$, 
$\tuple{v_4,v_8}\!:\!1$, 
$\tuple{v_6,v_9}\!:\!1$, 
$\tuple{v_7,v_3}\!:\!1$, 
$\tuple{v_7,v_8}\!:\!0.3$, 
$\tuple{v_8,v_5}\!:\!1$, 
$\tuple{v_8,v_9}\!:\!0.3$, 
$\tuple{v_9,v_5}\!:\!1$, 
$\tuple{v_9,v_7}\!:\!0.3\}$, 
\item $\varrho^\mM = \{\tuple{v_3,v_4}\!:\!1$, $\tuple{v_5,v_4}\!:\!1$, $\tuple{v_5,v_6}\!:\!1\}$.
\end{itemize}

%----------------------------------------------------------------------

\begin{table}
\caption{The greatest fuzzy $\Phi$-auto-bisimulation of the fuzzy Kripke model $\mM$ considered in Example~\ref{example: HGDHG} in different cases. \label{table: HGDHG}}
\scriptsize

\bigskip

\noindent
\begin{tabular}{ccc}
& Under the product semantics: & Under the {\L}ukasiewicz semantics: \\[1ex]
\!\!\!\!\!\!\!\!
\begin{tabular}{l}
Case \\
$C \notin \Phi$,\\ 
$O \notin \Phi$:
\end{tabular}\!\!\!\!\!
& 
\begin{tabular}{|*{11}{>{\centering\arraybackslash}m{0.22cm}|}}
\hline
$E$   & $v_0$ & $v_1$ & $v_2$ & $v_3$ & $v_4$ & $v_5$ & $v_6$ & $v_7$ & $v_8$ & $v_9$ \\ \hline
$v_0$ &   1   & \!0.5 &   1   &   0   &   0   &   0   &   0   &   0   &   0   &   0   \\ \hline
$v_1$ & \!0.5 &   1   & \!0.5 &   0   &   0   &   0   &   0   &   0   &   0   &   0   \\ \hline
$v_2$ &   1   & \!0.5 &   1   &   0   &   0   &   0   &   0   &   0   &   0   &   0   \\ \hline
$v_3$ &   0   &   0   &   0   &   1   &   0   &   1   &   0   &   0   &   0   &   0   \\ \hline
$v_4$ &   0   &   0   &   0   &   0   &   1   &   0   &   1   &   0   &   0   &   0   \\ \hline
$v_5$ &   0   &   0   &   0   &   1   &   0   &   1   &   0   &   0   &   0   &   0   \\ \hline
$v_6$ &   0   &   0   &   0   &   0   &   1   &   0   &   1   &   0   &   0   &   0   \\ \hline
$v_7$ &   0   &   0   &   0   &   0   &   0   &   0   &   0   &   1   &   1   &   1   \\ \hline
$v_8$ &   0   &   0   &   0   &   0   &   0   &   0   &   0   &   1   &   1   &   1   \\ \hline
$v_9$ &   0   &   0   &   0   &   0   &   0   &   0   &   0   &   1   &   1   &   1   \\ \hline
\end{tabular}
& 
\begin{tabular}{|*{11}{>{\centering\arraybackslash}m{0.22cm}|}}
\hline
$E$   & $v_0$ & $v_1$ & $v_2$ & $v_3$ & $v_4$ & $v_5$ & $v_6$ & $v_7$ & $v_8$ & $v_9$ \\ \hline
$v_0$ &   1   & \!0.5 &   1   &   0   &   0   &   0   &   0   &   0   &   0   &   0   \\ \hline
$v_1$ & \!0.5 &   1   & \!0.5 &   0   &   0   &   0   &   0   & \!0.4 & \!0.4 & \!0.4 \\ \hline
$v_2$ &   1   & \!0.5 &   1   &   0   &   0   &   0   &   0   &   0   &   0   &   0   \\ \hline
$v_3$ &   0   &   0   &   0   &   1   &   0   &   1   &   0   &   0   &   0   &   0   \\ \hline
$v_4$ &   0   &   0   &   0   &   0   &   1   &   0   &   1   &   0   &   0   &   0   \\ \hline
$v_5$ &   0   &   0   &   0   &   1   &   0   &   1   &   0   &   0   &   0   &   0   \\ \hline
$v_6$ &   0   &   0   &   0   &   0   &   1   &   0   &   1   &   0   &   0   &   0   \\ \hline
$v_7$ &   0   & \!0.4 &   0   &   0   &   0   &   0   &   0   &   1   &   1   &   1   \\ \hline
$v_8$ &   0   & \!0.4 &   0   &   0   &   0   &   0   &   0   &   1   &   1   &   1   \\ \hline
$v_9$ &   0   & \!0.4 &   0   &   0   &   0   &   0   &   0   &   1   &   1   &   1   \\ \hline
\end{tabular}
\\
\\
\!\!\!\!\!\!\!\!
\begin{tabular}{l}
Case \\
$C \notin \Phi$,\\ 
$O \in \Phi$:
\end{tabular}\!\!\!\!\!
& 
\begin{tabular}{|*{11}{>{\centering\arraybackslash}m{0.22cm}|}}
\hline
$E$   & $v_0$ & $v_1$ & $v_2$ & $v_3$ & $v_4$ & $v_5$ & $v_6$ & $v_7$ & $v_8$ & $v_9$ \\ \hline
$v_0$ &   1   &   0   &   0   &   0   &   0   &   0   &   0   &   0   &   0   &   0   \\ \hline
$v_1$ &   0   &   1   & \!0.5 &   0   &   0   &   0   &   0   &   0   &   0   &   0   \\ \hline
$v_2$ &   0   & \!0.5 &   1   &   0   &   0   &   0   &   0   &   0   &   0   &   0   \\ \hline
$v_3$ &   0   &   0   &   0   &   1   &   0   &   1   &   0   &   0   &   0   &   0   \\ \hline
$v_4$ &   0   &   0   &   0   &   0   &   1   &   0   &   1   &   0   &   0   &   0   \\ \hline
$v_5$ &   0   &   0   &   0   &   1   &   0   &   1   &   0   &   0   &   0   &   0   \\ \hline
$v_6$ &   0   &   0   &   0   &   0   &   1   &   0   &   1   &   0   &   0   &   0   \\ \hline
$v_7$ &   0   &   0   &   0   &   0   &   0   &   0   &   0   &   1   &   1   &   1   \\ \hline
$v_8$ &   0   &   0   &   0   &   0   &   0   &   0   &   0   &   1   &   1   &   1   \\ \hline
$v_9$ &   0   &   0   &   0   &   0   &   0   &   0   &   0   &   1   &   1   &   1   \\ \hline
\end{tabular}
& 
\begin{tabular}{|*{11}{>{\centering\arraybackslash}m{0.22cm}|}}
\hline
$E$   & $v_0$ & $v_1$ & $v_2$ & $v_3$ & $v_4$ & $v_5$ & $v_6$ & $v_7$ & $v_8$ & $v_9$ \\ \hline
$v_0$ &   1   &   0   &   0   &   0   &   0   &   0   &   0   &   0   &   0   &   0   \\ \hline
$v_1$ &   0   &   1   & \!0.5 &   0   &   0   &   0   &   0   & \!0.4 & \!0.4 & \!0.4 \\ \hline
$v_2$ &   0   & \!0.5 &   1   &   0   &   0   &   0   &   0   &   0   &   0   &   0   \\ \hline
$v_3$ &   0   &   0   &   0   &   1   &   0   &   1   &   0   &   0   &   0   &   0   \\ \hline
$v_4$ &   0   &   0   &   0   &   0   &   1   &   0   &   1   &   0   &   0   &   0   \\ \hline
$v_5$ &   0   &   0   &   0   &   1   &   0   &   1   &   0   &   0   &   0   &   0   \\ \hline
$v_6$ &   0   &   0   &   0   &   0   &   1   &   0   &   1   &   0   &   0   &   0   \\ \hline
$v_7$ &   0   & \!0.4 &   0   &   0   &   0   &   0   &   0   &   1   &   1   &   1   \\ \hline
$v_8$ &   0   & \!0.4 &   0   &   0   &   0   &   0   &   0   &   1   &   1   &   1   \\ \hline
$v_9$ &   0   & \!0.4 &   0   &   0   &   0   &   0   &   0   &   1   &   1   &   1   \\ \hline
\end{tabular}
\\
\\
\!\!\!\!\!\!\!\!
\begin{tabular}{l}
Case \\
$C \in \Phi$,\\ 
$O \notin \Phi$:
\end{tabular}\!\!\!\!\!
& 
\begin{tabular}{|*{11}{>{\centering\arraybackslash}m{0.22cm}|}}
\hline
$E$   & $v_0$ & $v_1$ & $v_2$ & $v_3$ & $v_4$ & $v_5$ & $v_6$ & $v_7$ & $v_8$ & $v_9$ \\ \hline
$v_0$ &   1   & \!0.5 &   1   &   0   &   0   &   0   &   0   &   0   &   0   &   0   \\ \hline
$v_1$ & \!0.5 &   1   & \!0.5 &   0   &   0   &   0   &   0   &   0   &   0   &   0   \\ \hline
$v_2$ &   1   & \!0.5 &   1   &   0   &   0   &   0   &   0   &   0   &   0   &   0   \\ \hline
$v_3$ &   0   &   0   &   0   &   1   &   0   &   1   &   0   &   0   &   0   &   0   \\ \hline
$v_4$ &   0   &   0   &   0   &   0   &   1   &   0   &   1   &   0   &   0   &   0   \\ \hline
$v_5$ &   0   &   0   &   0   &   1   &   0   &   1   &   0   &   0   &   0   &   0   \\ \hline
$v_6$ &   0   &   0   &   0   &   0   &   1   &   0   &   1   &   0   &   0   &   0   \\ \hline
$v_7$ &   0   &   0   &   0   &   0   &   0   &   0   &   0   &   1   &   1   &   1   \\ \hline
$v_8$ &   0   &   0   &   0   &   0   &   0   &   0   &   0   &   1   &   1   &   1   \\ \hline
$v_9$ &   0   &   0   &   0   &   0   &   0   &   0   &   0   &   1   &   1   &   1   \\ \hline
\end{tabular}
& 
\begin{tabular}{|*{11}{>{\centering\arraybackslash}m{0.22cm}|}}
\hline
$E$   & $v_0$ & $v_1$ & $v_2$ & $v_3$ & $v_4$ & $v_5$ & $v_6$ & $v_7$ & $v_8$ & $v_9$ \\ \hline
$v_0$ &   1   & \!0.5 &   1   &   0   &   0   &   0   &   0   &   0   &   0   &   0   \\ \hline
$v_1$ & \!0.5 &   1   & \!0.5 &   0   &   0   &   0   &   0   &   0   &   0   &   0   \\ \hline
$v_2$ &   1   & \!0.5 &   1   &   0   &   0   &   0   &   0   &   0   &   0   &   0   \\ \hline
$v_3$ &   0   &   0   &   0   &   1   &   0   &   1   &   0   &   0   &   0   &   0   \\ \hline
$v_4$ &   0   &   0   &   0   &   0   &   1   &   0   &   1   &   0   &   0   &   0   \\ \hline
$v_5$ &   0   &   0   &   0   &   1   &   0   &   1   &   0   &   0   &   0   &   0   \\ \hline
$v_6$ &   0   &   0   &   0   &   0   &   1   &   0   &   1   &   0   &   0   &   0   \\ \hline
$v_7$ &   0   &   0   &   0   &   0   &   0   &   0   &   0   &   1   &   1   &   1   \\ \hline
$v_8$ &   0   &   0   &   0   &   0   &   0   &   0   &   0   &   1   &   1   &   1   \\ \hline
$v_9$ &   0   &   0   &   0   &   0   &   0   &   0   &   0   &   1   &   1   &   1   \\ \hline
\end{tabular}
\\
\\
\!\!\!\!\!\!\!\!
\begin{tabular}{l}
Case \\
$C \in \Phi$,\\ 
$O \in \Phi$:
\end{tabular}\!\!\!\!\!
& 
\begin{tabular}{|*{11}{>{\centering\arraybackslash}m{0.22cm}|}}
\hline
$E$   & $v_0$ & $v_1$ & $v_2$ & $v_3$ & $v_4$ & $v_5$ & $v_6$ & $v_7$ & $v_8$ & $v_9$ \\ \hline
$v_0$ &   1   &   0   &   0   &   0   &   0   &   0   &   0   &   0   &   0   &   0   \\ \hline
$v_1$ &   0   &   1   &   0   &   0   &   0   &   0   &   0   &   0   &   0   &   0   \\ \hline
$v_2$ &   0   &   0   &   1   &   0   &   0   &   0   &   0   &   0   &   0   &   0   \\ \hline
$v_3$ &   0   &   0   &   0   &   1   &   0   &   0   &   0   &   0   &   0   &   0   \\ \hline
$v_4$ &   0   &   0   &   0   &   0   &   1   &   0   &   0   &   0   &   0   &   0   \\ \hline
$v_5$ &   0   &   0   &   0   &   0   &   0   &   1   &   0   &   0   &   0   &   0   \\ \hline
$v_6$ &   0   &   0   &   0   &   0   &   0   &   0   &   1   &   0   &   0   &   0   \\ \hline
$v_7$ &   0   &   0   &   0   &   0   &   0   &   0   &   0   &   1   &   0   &   0   \\ \hline
$v_8$ &   0   &   0   &   0   &   0   &   0   &   0   &   0   &   0   &   1   &   0   \\ \hline
$v_9$ &   0   &   0   &   0   &   0   &   0   &   0   &   0   &   0   &   0   &   1   \\ \hline
\end{tabular}
& 
\begin{tabular}{|*{11}{>{\centering\arraybackslash}m{0.22cm}|}}
\hline
$E$   & $v_0$ & $v_1$ & $v_2$ & $v_3$ & $v_4$ & $v_5$ & $v_6$ & $v_7$ & $v_8$ & $v_9$ \\ \hline
$v_0$ &   1   &   0   &   0   &   0   &   0   &   0   &   0   &   0   &   0   &   0   \\ \hline
$v_1$ &   0   &   1   & \!0.5 &   0   &   0   &   0   &   0   &   0   &   0   &   0   \\ \hline
$v_2$ &   0   & \!0.5 &   1   &   0   &   0   &   0   &   0   &   0   &   0   &   0   \\ \hline
$v_3$ &   0   &   0   &   0   &   1   &   0   & \!0.6 &   0   &   0   &   0   &   0   \\ \hline
$v_4$ &   0   &   0   &   0   &   0   &   1   &   0   & \!0.6 &   0   &   0   &   0   \\ \hline
$v_5$ &   0   &   0   &   0   & \!0.6 &   0   &   1   &   0   &   0   &   0   &   0   \\ \hline
$v_6$ &   0   &   0   &   0   &   0   & \!0.6 &   0   &   1   &   0   &   0   &   0   \\ \hline
$v_7$ &   0   &   0   &   0   &   0   &   0   &   0   &   0   &   1   & \!0.6 & \!0.6 \\ \hline
$v_8$ &   0   &   0   &   0   &   0   &   0   &   0   &   0   & \!0.6 &   1   & \!0.6 \\ \hline
$v_9$ &   0   &   0   &   0   &   0   &   0   &   0   &   0   & \!0.6 & \!0.6 &   1   \\ \hline
\end{tabular}
\end{tabular}
\end{table}

%----------------------------------------------------------------------

The greatest fuzzy $\Phi$-auto-bisimulation $E$ of~$\mM$ is displayed in Table~\ref{table: HGDHG} for different cases of~$\Phi$ under the product and {\L}ukasiewicz semantics. The computation of $E$ and the other results reported in this example can be checked using the {\em MinFKMPL.py} module of the package~\cite{MinFKM-prog}, by running this module and examining the outputs involving the input files {\em m2.in} and {\em m2o.in}. 
Observe that:
\begin{itemize}
\item $E$ is the identity relation in the case where $\{C,O\} \subset \Phi$ and the product structure is used.
\item The relation $E$ is the same in the following three cases:
\begin{itemize}
\item $C \notin \Phi$, $O \notin \Phi$, and the product structure is used;
\item $C \in \Phi$, $O \notin \Phi$, and the product structure is used;
\item $C \in \Phi$, $O \notin \Phi$, and the {\L}ukasiewicz structure is used.
\end{itemize}
\end{itemize}

%----------------------------------------------------------------------

\newcommand{\EtMPq}{
\begin{tikzpicture}[scale=0.75, transform shape,->,>=stealth,auto,black]
\node (top) {};
\node[draw] (v0) [node distance=0.4cm, below of=top] {$E_{v_0}\!:\!p_{1.0}$};
\node[draw] (v1) [node distance=2.0cm, right of=v0] {$E_{v_1}\!:\!p_{0.5}$};
\node[draw] (v3) [node distance=2.5cm, below of=v0] {$E_{v_3}$};
\node[draw] (v4) [node distance=2.5cm, below of=v1] {$E_{v_4}$};
\node (bv3) [node distance=2.5cm, below of=v3] {};
\node[draw] (v7) [node distance=1.0cm, right of=bv3] {$E_{v_7}$};
\draw (v0) to node [left, xshift=1pt]{\footnotesize{0.4}} (v3);
\draw (v0) to node [left, pos=0.85]{\footnotesize{0.4}} (v4);
\draw (v1) to node [left, pos=0.15]{\footnotesize{0.4}} (v3);
\draw (v1) to node [right, xshift=-1pt]{\footnotesize{0.4}} (v4);
\draw[dashed] (v3) to node [below]{\footnotesize{1}} (v4);
\draw (v4) to node [right]{\footnotesize{1}} (v7);
\draw (v7) to node [left]{\footnotesize{1}} (v3);
\draw (v7) edge[out=-120,in=-60,looseness=6] node[below]{\footnotesize{0.3}} (v7);
\end{tikzpicture}
}

\newcommand{\EtMLq}{
\begin{tikzpicture}[scale=0.75, transform shape,->,>=stealth,auto,black]
\node (top) {};
\node[draw] (v0) [node distance=0.4cm, below of=top] {$E_{v_0}\!:\!p_{1.0}$};
\node[draw] (v1) [node distance=2.0cm, right of=v0] {$E_{v_1}\!:\!p_{0.5}$};
\node[draw] (v3) [node distance=2.5cm, below of=v0] {$E_{v_3}$};
\node[draw] (v4) [node distance=2.5cm, below of=v1] {$E_{v_4}$};
\node (bv3) [node distance=2.5cm, below of=v3] {};
\node[draw] (v7) [node distance=1.0cm, right of=bv3] {$E_{v_7}$};
\draw (v0) to node [left, xshift=1pt]{\footnotesize{0.4}} (v3);
\draw (v0) to node [left, pos=0.85]{\footnotesize{0.4}} (v4);
\draw (v1) to node [left, pos=0.15]{\footnotesize{0.4}} (v3);
\draw (v1) to node [left, xshift=1pt]{\footnotesize{0.4}} (v4);
\draw[dashed] (v3) to node [below]{\footnotesize{1}} (v4);
\draw (v4) to node [right]{\footnotesize{1}} (v7);
\draw (v7) to node [left]{\footnotesize{1}} (v3);
\draw (v7) edge[out=-120,in=-60,looseness=6] node[below]{\footnotesize{0.3}} (v7);
\draw (v4) edge[bend right=20] node [right, xshift=-1pt]{\footnotesize{0.4}} (v1);
\end{tikzpicture}
}

\newcommand{\EtMPqO}{
\begin{tikzpicture}[scale=0.75, transform shape,->,>=stealth,auto,black]
\node (top) {};
\node[draw] (v0) [node distance=0.4cm, below of=top] {$E_{v_0}\!:\!p_{1.0}$};
\node (rv0) [node distance=1.0cm, right of=v0] {};
\node[draw] (v1) [node distance=2.0cm, right of=v0] {$E_{v_1}\!:\!p_{0.5}$};
\node (rv1) [node distance=1.0cm, right of=v1] {};
\node[draw] (v2) [node distance=2.0cm, right of=v1] {$E_{v_2}\!:\!p_{1.0}$};
\node[draw] (v3) [node distance=2.5cm, below of=rv0] {$E_{v_3}$};
\node[draw] (v4) [node distance=2.5cm, below of=rv1] {$E_{v_4}$};
\node (bv3) [node distance=2.5cm, below of=v3] {};
\node[draw] (v7) [node distance=1.0cm, right of=bv3] {$E_{v_7}$};
\draw (v0) to node [left, xshift=1pt]{\footnotesize{0.4}} (v3);
\draw (v0) to node [left, pos=0.45, xshift=-1pt]{\footnotesize{0.4}} (v4);
\draw (v1) to node [left, pos=0.15]{\footnotesize{0.4}} (v3);
\draw (v1) to node [right, pos=0.15]{\footnotesize{0.4}} (v4);
\draw[dashed] (v3) to node [below]{\footnotesize{1}} (v4);
\draw (v4) to node [right]{\footnotesize{1}} (v7);
\draw (v7) to node [left]{\footnotesize{1}} (v3);
\draw (v7) edge[out=-120,in=-60,looseness=6] node[below]{\footnotesize{0.3}} (v7);
\draw (v2) to node [right, xshift=-1pt]{\footnotesize{0.4}} (v4);
\draw (v2) to node [right, pos=0.45, xshift=1pt]{\footnotesize{0.4}} (v3);
\end{tikzpicture}
}

\newcommand{\EtMLqO}{
\begin{tikzpicture}[scale=0.75, transform shape,->,>=stealth,auto,black]
\node (top) {};
\node[draw] (v0) [node distance=0.4cm, below of=top] {$E_{v_0}\!:\!p_{1.0}$};
\node (rv0) [node distance=1.0cm, right of=v0] {};
\node[draw] (v1) [node distance=2.0cm, right of=v0] {$E_{v_1}\!:\!p_{0.5}$};
\node (rv1) [node distance=1.0cm, right of=v1] {};
\node[draw] (v2) [node distance=2.0cm, right of=v1] {$E_{v_2}\!:\!p_{1.0}$};
\node[draw] (v3) [node distance=2.5cm, below of=rv0] {$E_{v_3}$};
\node[draw] (v4) [node distance=2.5cm, below of=rv1] {$E_{v_4}$};
\node (bv3) [node distance=2.5cm, below of=v3] {};
\node[draw] (v7) [node distance=1.0cm, right of=bv3] {$E_{v_7}$};
\draw (v0) to node [left, xshift=1pt]{\footnotesize{0.4}} (v3);
\draw (v0) to node [left, pos=0.45, xshift=-1pt]{\footnotesize{0.4}} (v4);
\draw (v1) to node [left, pos=0.15]{\footnotesize{0.4}} (v3);
\draw (v1) to node [left, pos=0.35, xshift=1pt]{\footnotesize{0.4}} (v4);
\draw[dashed] (v3) to node [below]{\footnotesize{1}} (v4);
\draw (v4) to node [right]{\footnotesize{1}} (v7);
\draw (v7) to node [left]{\footnotesize{1}} (v3);
\draw (v7) edge[out=-120,in=-60,looseness=6] node[below]{\footnotesize{0.3}} (v7);
\draw (v2) to node [right, xshift=-1pt]{\footnotesize{0.4}} (v4);
\draw (v2) to node [right, pos=0.95, xshift=1pt]{\footnotesize{0.4}} (v3);
\draw (v4) edge[bend right=20] node [right, pos=0.47, xshift=-2pt]{\footnotesize{0.4}} (v1);
\end{tikzpicture}
}

\begin{figure}
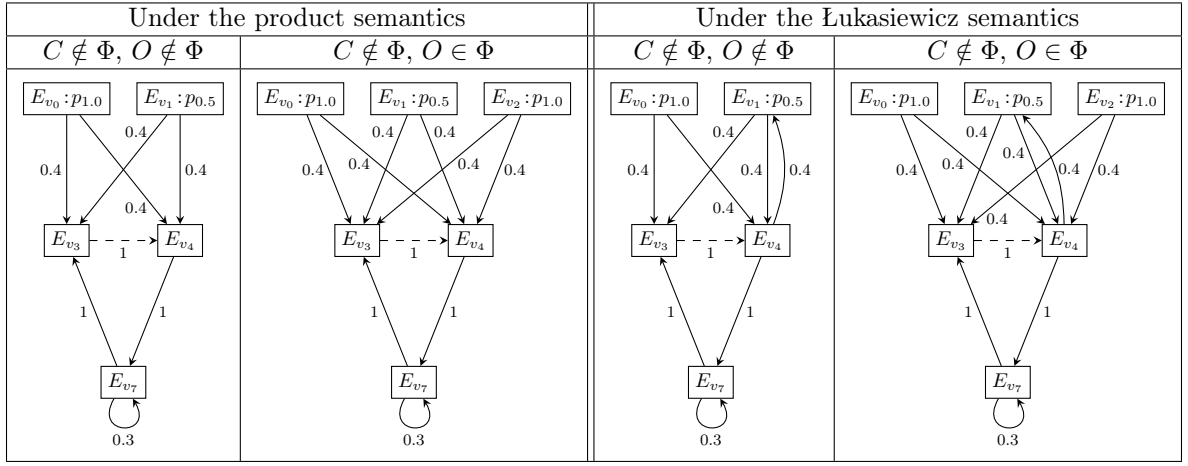

\begin{center}
\begin{tabular}{|c|c||c|c|}
\hline
\multicolumn{2}{|c||}{Under the product semantics} & \multicolumn{2}{c|}{Under the {\L}ukasiewicz semantics} 
\\ \hline
$C \notin \Phi$, $O \notin \Phi$ & $C \notin \Phi$, $O \in \Phi$ & $C \notin \Phi$, $O \notin \Phi$ & $C \notin \Phi$, $O \in \Phi$ 
\\ \hline
  \EtMPq
& \EtMPqO
& \EtMLq
& \EtMLqO
\\ \hline
\end{tabular}
\end{center}
\caption{The quotient model $\mME$ of the fuzzy Kripke model $\mM$ considered in Example~\ref{example: HGDHG} with respect to its greatest fuzzy $\Phi$-auto-bisimulation $E$, for different cases of~$\Phi$ under the product and {\L}ukasiewicz semantics.\label{fig: HGDHG2}}
\end{figure}

%----------------------------------------------------------------------

The quotient fuzzy Kripke model $\mME$ is illustrated in Figure~\ref{fig: HGDHG2} for some cases of~$\Phi$ under the product and {\L}ukasiewicz semantics.
$\mME$ in the case where $C \in \Phi$ and $O \notin \Phi$, under both the product and {\L}ukasiewicz semantics, is the same as $\mME$ in the case where $C \notin \Phi$ and $O \notin \Phi$ under the product semantics, which is illustrated in the leftmost subfigure in Figure~\ref{fig: HGDHG2}.
When $C \in \Phi$ and $O \in \Phi$, $\mME$ under the product semantics is isomorphic to $\mM$, whereas $\mME$ under the {\L}ukasiewicz semantics is an extension of $\mME$ under the product semantics with additional transitions of degree~$0.6$. 

%----------------------------------------------------------------------

\newcommand{\EtMPlC}{
\begin{tikzpicture}[scale=1.0, transform shape,->,>=stealth,auto,black]
\node (top) {};
\node[draw] (v0) [node distance=0.4cm, below of=top] {$v_0\!:\!p_{1.0}$};
\node[draw] (v1) [node distance=2.0cm, right of=v0] {$v_1\!:\!p_{0.5}$};
\node[draw] (v3) [node distance=2.5cm, below of=v0] {$v_3$};
\node[draw] (v4) [node distance=2.5cm, below of=v1] {$v_4$};
\node (bv3) [node distance=2.5cm, below of=v3] {};
\node[draw] (v7) [node distance=1.0cm, right of=bv3] {$v_7$};
\draw (v0) to node [left, xshift=1pt]{\footnotesize{0.4}} (v3);
\draw (v0) to node [left, pos=0.85]{\footnotesize{0.4}} (v4);
\draw (v1) to node [left, pos=0.15]{\footnotesize{0.4}} (v3);
\draw (v1) to node [right, xshift=-1pt]{\footnotesize{0.4}} (v4);
\draw[dashed] (v3) to node [below]{\footnotesize{1}} (v4);
\draw (v4) to node [right]{\footnotesize{1}} (v7);
\draw (v7) to node [left]{\footnotesize{1}} (v3);
\draw (v7) edge[out=-120,in=-60,looseness=6] node[below]{\footnotesize{0.3}} (v7);
\end{tikzpicture}
}

\newcommand{\EtMPl}{
\begin{tikzpicture}[scale=1.0, transform shape,->,>=stealth,auto,black]
\node (top) {};
\node[draw] (v0) [node distance=0.4cm, below of=top] {$v_0\!:\!p_{1.0}$};
\node (lv0) [node distance=1.0cm, left of=v0] {};
\node (rv0) [node distance=1.0cm, right of=v0] {};
\node[draw] (v3) [node distance=2.5cm, below of=lv0] {$v_3$};
\node[draw] (v4) [node distance=2.5cm, below of=rv0] {$v_4$};
\node[draw] (v7) [node distance=5.0cm, below of=v0] {$v_7$};
\draw (v0) to node [left]{\footnotesize{0.4}} (v3);
\draw (v0) to node [right]{\footnotesize{0.4}} (v4);
\draw[dashed] (v3) to node [below]{\footnotesize{1}} (v4);
\draw (v4) to node [right]{\footnotesize{1}} (v7);
\draw (v7) to node [left]{\footnotesize{1}} (v3);
\draw (v7) edge[out=-120,in=-60,looseness=6] node[below]{\footnotesize{0.3}} (v7);
\end{tikzpicture}
}

\begin{figure}
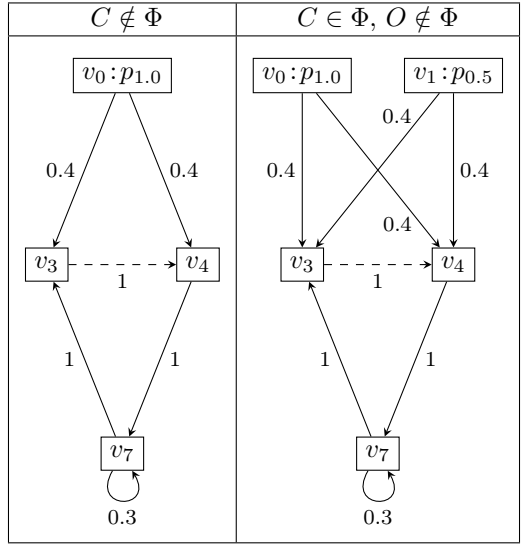

\centering
\begin{tabular}{|c|c|}
\hline
\ \ $C \notin \Phi$\ \ & 
$C \in \Phi$, $O \notin \Phi$ 
\\ \hline 
  \EtMPl
& \EtMPlC
\\ \hline 
\end{tabular}
\caption{The reduced fuzzy Kripke model obtained by applying Algorithm~\ref{alg1} to the fuzzy Kripke model $\mM$ considered in Example~\ref{example: HGDHG}, for some cases of~$\Phi$ under the product and {\L}ukasiewicz semantics.\label{fig: HGDHG3}}
\end{figure}

%----------------------------------------------------------------------

The reduced fuzzy Kripke model $\mMp$ obtained by applying Algorithm~\ref{alg1} to~$\mM$ is illustrated in Figure~\ref{fig: HGDHG3} for some cases of~$\Phi$ under the product and {\L}ukasiewicz semantics. The reduced model in each case is identical under the product and {\L}ukasiewicz semantics. The case where $C \notin \Phi$ is independent of whether $O \in \Phi$. Notice that, when $C \in \Phi$ and $O \notin \Phi$, $\mMp$ is isomorphic to $\mME$. The case where $\{C,O\} \subset \Phi$ is not covered in Figure~\ref{fig: HGDHG3}. In this case, $\mMp$ is identical to~$\mM$.

%----------------------------------------------------------------------

Consider the case where $C \in \Phi$, $O \notin \Phi$, and the underlying residuated lattice is the product or {\L}ukasiewicz structure. We examine the execution of Algorithm~\ref{alg1} for $\mM$:
\begin{itemize}
\item After executing the statements~\ref{step: alg1 1}--\ref{step: alg1 2}, we have $\DeltaMp = \{v_0\}$, $v_0^\mMp = v_0$, $p^\mMp(v_0) = 1$, and $\unprocessed = \{v_0\}$.

\item In the first iteration of the ``while'' loop, $v_0$ is extracted from $\unprocessed$ and processed. The states $v_3$ and $v_4$ are processed (as $y$) by the inner ``foreach'' loop~\ref{step: alg1 5} using $\alpha = \sigma$. As a result, they are both added to $\DeltaMp$ and $\unprocessed$ by the statement~\ref{step: alg1 10}, with $\sigma^\mMp(v_0, v_3)$ and $\sigma^\mMp(v_0, v_4)$ set to~$0.4$ by the statement~\ref{step: alg1 12}. 

\item Consider the second iteration of the ``while'' loop and assume that $v_3$ is extracted from $\unprocessed$ and processed. For $\alpha = \varrho$, only the state $v_4$ is processed by the inner ``foreach'' loop~\ref{step: alg1 5}. As a result, $\varrho^\mMp(v_3, v_4)$ is set to~$1$ by the statement~\ref{step: alg1 8}. For $\alpha = \sigma^-$, since $(\sigma^-)^\mM(v_3, v_7) > (\sigma^-)^\mM(v_3, v_1) = (\sigma^-)^\mM(v_3, v_0) > 0$, the states $v_7$, $v_1$, and $v_0$ are processed by the inner ``foreach'' loop~\ref{step: alg1 5}. As a result, $v_7$ and $v_1$ are added to $\DeltaMp$ and $\unprocessed$, with $\sigma^\mMp(v_7, v_3)$ set to~1, $p^\mMp(v_1)$ set to $0.5$, and $\sigma^\mMp(v_1, v_3)$ set to~$0.4$, by the statements~\ref{step: alg1 10}--\ref{step: alg1 12}. 

\item Consider the third iteration of the ``while'' loop and assume that $v_4$ is extracted from $\unprocessed$ and processed. 
    \begin{itemize}
    \item For $\alpha = \sigma$, the states $v_7$ and $v_8$ are processed by the inner ``foreach'' loop~\ref{step: alg1 5}. When $v_7$ is processed, $\sigma^\mMp(v_4, v_7)$ is set to~$1$ by the statement~\ref{step: alg1 8}. When $v_8$ is processed, it does not satisfy the condition of the ``if'' statement~\ref{step: alg1 6} due to $y' = v_7$. 
    \item For $\alpha = \sigma^-$, the states $v_0$ and $v_1$ are processed by the inner ``foreach'' loop~\ref{step: alg1 5}. As a result, $\sigma^\mMp(v_0, v_4)$ and $\sigma^\mMp(v_1, v_4)$ are set to~$0.4$ by the statement~\ref{step: alg1 8}. 
    \item For $\alpha = \varrho^-$, the states $v_3$ and $v_5$ are processed (as $y$) by the inner ``foreach'' loop~\ref{step: alg1 5}, but neither satisfies the condition of the ``if'' statement~\ref{step: alg1 6} due to $y' = v_3$. 
    \end{itemize}

\item Consider the fourth iteration of the ``while'' loop and assume that $v_7$ is extracted from $\unprocessed$ and processed. 
    \begin{itemize}
    \item For $\alpha = \sigma$, the states $v_3$ and $v_8$ are processed by the inner ``foreach'' loop~\ref{step: alg1 5}. When $v_3$ is processed, it does not satisfy the condition of the ``if'' statement~\ref{step: alg1 6} due to $y' = v_3$. When $v_8$ is processed, $\sigma^\mMp(v_7,v_7)$ is set to $0.3$ by the statement~\ref{step: alg1 8}. 
    \item For $\alpha = \sigma^-$, the states $v_4$ and $v_9$ are processed by the inner ``foreach'' loop~\ref{step: alg1 5}, but neither satisfies the condition of the ``if'' statement~\ref{step: alg1 6} (due to $y' = v_4$ and $y' = v_7$, respectively).
    \end{itemize}

\item In the fifth iteration of the ``while'' loop, $v_1$ is extracted from $\unprocessed$ and processed. The states $v_3$, $v_4$, $v_5$, and $v_6$ are processed (as $y$) by the inner ``foreach'' loop~\ref{step: alg1 5} using $\alpha = \sigma$, but none of them satisfies the condition of the ``if'' statement~\ref{step: alg1 6} (due to $y' = v_3$ or $y' = v_4$). 

\item After the fifth iteration, the ``while'' loop terminates, and the algorithm returns $\mMp$, which is illustrated in the right subfigure in Figure~\ref{fig: HGDHG3}. 
\myend
\end{itemize}
\end{example}

%===============================================================================

\end{document}